\documentclass[opre,nonblindrev]{informs_arxiv}

\DoubleSpacedXI % Made defa[]ult 4/4/2014 at request
\usepackage{color}
\usepackage{endnotes}
\let\footnote=\endnote

\usepackage{natbib}
 \bibpunct[, ]{(}{)}{,}{a}{}{,}%
 \def\bibfont{\small}%
\TheoremsNumberedThrough     % Preferred (Theorem 1, Lemma 1, Theorem 2)
\ECRepeatTheorems

\EquationsNumberedThrough    % Default: (1), (2), ...
\usepackage{subcaption}
\usepackage{algorithm}
\usepackage{algorithmic}
\usepackage{comment}
\newtheorem{condition}{Condition}

\newcommand{\red}[1]{\begin{color}{red}#1\end{color}}

\DeclareMathOperator*{\diag}{Diag}

\def\norm#1{\|#1 \|}
\def\inprod#1#2{\langle#1,\,#2 \rangle}

\def\cA{{\mathcal A}}

\def\cC{{\mathcal C}}

\def\cE{{\mathcal E}}
\def\cF{{\mathcal F}}
\def\cG{{\mathcal G}}

\def\cI{{\mathcal I}}

\def\cL{{\mathcal L}}
\def\cM{{\mathcal M}}
\def\cN{{\mathcal N}}

\def\cS{{\mathcal S}}
\def\cT{{\mathcal T}}

\def\cW{{\mathcal W}}
\def\cX{{\mathcal X}}
\def\cY{{\mathcal Y}}

\def\norm#1{\|#1 \|}

\def\S{\mathbb{S}}
\def\R{\mathbb{R}}

\def\mw#1{\textcolor{blue}{#1}}

\def\bP{\mathbf{P}}

\def\bQ{\mathbf{Q}}
\def\bH{\mathbf{H}}

\def\bX{\mathbf{X}}
\def\bZ{\mathbf{Z}}	
\def\bU{\mathbf{U}}	
\def\bV{\mathbf{V}}	
	
\def\bD{\mathbf{D}}

\def\bS{\mathbf{S}}

\def\bR{\mathbf{R}}

\def\bW{\mathbf{W}}		
		
\def\bzero{\mathbf{0}}	
\def\bone{\mathbf{1}}

\newcommand{\D}{{\mathbf{D}}}

\renewcommand{\P}{{\mathbf{P}}}
\newcommand{\bbP}{{\mathbb{P}}}
\newcommand{\Q}{{\mathbf{Q}}}

\renewcommand{\S}{{\mathbf{S}}}
\newcommand{\U}{{\mathbf{U}}}
\newcommand{\V}{{\mathbf{V}}}
\newcommand{\Var}{{\rm Var}}

\newcommand{\W}{{\mathbf{W}}}
\newcommand{\X}{{\mathbf{X}}}

\newcommand{\Z}{{\mathbf{Z}}}

\def \bGamma   {\boldsymbol{\Gamma}}       \def \bDelta   {\boldsymbol{\Delta}}

\newcommand{\rank}{{\rm rank}}

\newcommand{\lonenorm}[1]{\lVert#1\rVert_1}
\newcommand{\ltwonorm}[1]{\lVert#1\rVert_2}

\newcommand{\opnorm}[1]{\|#1\|_{2}}
\newcommand{\fnorm}[1]{\|#1\|_{\mathrm{F}}}
\newcommand{\nnorm}[1]{\lVert#1\rVert_*}
\newcommand{\supnorm}[1]{ \lVert#1  \rVert_{\max}}
\newcommand{\inn}[1]{\langle #1 \rangle}

\def \DKL {D_{\mathrm{KL}}}
\def \PP {\mathbb{P}}

\def \OO {\mathbb{O}}
\def \NN {\mathbb{N}}
\def \EE {\mathbb{E}}
\def \ind {1}

\def \KL {\mathrm{KL}}
\def \vec	{\text{vec}}

\def\mybox#1{\vskip1mm \begin{center} \bf \red
		\hspace{.0\textwidth}\vbox{\hrule\hbox{\vrule\kern6pt
				\parbox{.95\textwidth}{\kern6pt#1\vskip6pt}\kern6pt\vrule}\hrule}
	\end{center} \vskip-5mm}
\begin{document}
%%%%%%%%%%%%%%%%

%\setcounter{page}{1}
%\DOI{10.1287/xxxx.0000.0000}%

\RUNAUTHOR{Zhu et al.}

\RUNTITLE{Estimation of Markov Models}

% \TITLE{Bundling Information Goods of Decreasing Value}
\TITLE{Learning Markov models via low-rank optimization}

\ARTICLEAUTHORS{%
\AUTHOR{Ziwei Zhu}
\AFF{Department of Statistics, University of Michigan, Ann Arbor, MI 48109, \EMAIL{ziweiz@umich.edu}
\AUTHOR{Xudong Li}
\AFF{School of Data Science, Fudan University, Shanghai 200433, \EMAIL{lixudong@fudan.edu.cn}\\
Shanghai Center for Mathematical Sciences, Fudan University, Shanghai 200433
} %, \URL{}}
\AUTHOR{Mengdi Wang}
\AFF{Department of Operations Research and Financial Engineering, Princeton University, Princeton, NJ 08544,
\EMAIL{mengdiw@princeton.edu}}
\AUTHOR{Anru R. Zhang}
\AFF{Department of Statistics, University of Wisconsin-Madison, Madison, WI 53706,
\EMAIL{anruzhang@stat.wisc.edu}\\
Department of Biostatistics and Bioinformatics, Duke University, 27710}
}
} % end of the block
\def\mw#1{\textcolor{red}{[#1]}}
\ABSTRACT{
Modeling unknown systems from data is a precursor of system optimization and sequential decision making. In this paper, we focus on learning a Markov model from a single trajectory of states. Suppose that the transition model has a small rank despite of having a large state space, meaning that the system admits a low-dimensional latent structure. We show that one can estimate the full transition model accurately using a trajectory of length that is proportional to the total number of states. 
We propose two maximum likelihood estimation methods: a convex approach with  nuclear-norm regularization and a nonconvex approach with rank constraint. We explicitly derive the statistical rates of both estimators in terms of the Kullback-Leiber divergence and the $\ell_2$ error and also establish a minimax lower bound to assess the tightness of these rates. For computing the nonconvex estimator, we develop a novel DC (difference of convex function) programming algorithm that starts with the convex M-estimator and then successively refines the solution till convergence. Empirical experiments demonstrate consistent superiority of the nonconvex estimator over the convex one.

% Enter your abstract
}%

%\KEYWORDS{deterministic inventory theory; infinite linear programming duality;
%  existence of optimal policies; semi-Markov decision process; cyclic schedule}


\maketitle
%%%%%%%%%%%%%%%%%%%%%%%%%%%%%%%%%%%%%%%%%%%%%%%%%%%%%%%%%%%%%%%%%%%%%%

% Samples of sectioning (and labeling) in OPRE
% NOTE: (1) \section and \subsection do NOT end with a period
%       (2) \subsubsection and lower need end punctuation
%       (3) capitalization is as shown (title style).
%
%\section{Introduction.}\label{intro} %%1.
%\subsection{Duality and the Classical EOQ Problem.}\label{class-EOQ} %% 1.1.
%\subsection{Outline.}\label{outline1} %% 1.2.
%\subsubsection{Cyclic Schedules for the General Deterministic SMDP.}
%  \label{cyclic-schedules} %% 1.2.1
%\section{Problem Description.}\label{problemdescription} %% 2.

% Text of your paper here

\section{Introduction}

In engineering and management applications, one often has to collect data from unknown systems, learn their transition functions, and learn to make predictions and decisions. A critical precursor of decision making is to model the system from data. We study how to learn an unknown Markov model of the system from its state-transition trajectories. When the system admits a large number of states, recovering the full model becomes sample expensive.

%Consider the discrete-state discrete-time Markov chain with a finite but potentially large state space. Suppose we are given a sequence of state transitions generated by an unknown Markov process. %The true transition matrix is unknown, and the physical states of the system and its law of transition is hidden under massive noisy observations. 
%We are interested in the recovering the full transition model from finite-length trajectories. 
In this paper, we focus on Markov processes where the transition matrix has a small rank.
The small rank implies that the observed process is governed by a low-dimensional latent process which we cannot see in a straightforward manner. It is a property that is (approximately) satisfied in a wide range of practical systems. Despite the large state space, the low-rank property unlocks potential of accurate learning of a full set of transition density functions based on short empirical trajectories. 

\subsection{Motivating Examples}

Practical state-transition processes with a large number of states often exhibit low-rank structures. For example, the sequence of stops made by a taxi turns out to follow a Markov model with approximately low rank structure \citep{liu2012understanding, benson2017spacey}. 
For another example, random walk on a lumpable network has a low-rank transition matrix \citep{buchholz1994exact,e2008optimal}. The transition kernel with fast decaying eigenvalues has been also observed in molecular dynamics \citep{rohrdanz2011determination}, which can be used to find metastable states, coresets and manifold structures of complicated dynamics \citep{chodera2007automatic, coifman2008diffusion}.

%The low-rank Markov model is also related to dimension reduction of  control and reinforcement learning, where latent structures like Markov process with rich observations \citep{azizzadenesheli2016reinforcement1} and state aggregation \citep{singh1995reinforcement} are useful for reducing the computation overhead due to a large state space.

Low-rank Markov models are also related to dimension reduction for control systems and reinforcement learning. For example, the state aggregation approach for modeling a high-dimensional system can be viewed as a low-rank approximation approach \citep{bertsekas1995dynamic, bertsekas1995neuro,singh1995reinforcement}. In state aggregation, one assumes that there exists a latent stochastic process $\{z_t\} \subset [r]$ such that
$\PP( s_{t+1}  \mid s_t ) =  \sum_z \PP( z_t =z \mid  s_t) \PP( s_{t+1} \mid z_t=z),$
which is equivalent to a factorization model of the transition kernel $\P$.
%In this case, one can verify that the transition matrix $\P$ can be factorized as
%f a $r$-latent-variable Markov chain has a nonnegative rank at most $r$ and there exists $\U,\V\in\Re^{p\times r}, \tilde\P\in \Re^{r\times r}$ such that 
%$\P = \U \tilde \P \V^{\top},$
%where $\tilde\P$ is a stochastic matrix, rows of $\U$ and columns of $\V$ are vectors of probability distributions. 
%In particular, the model implies a hard state aggregation structure if $\U$ and $\V$ are block functions. 
In the context of reinforcement learning, the nonnegative factorization model was referred to as the generalization to the rich-observation model \citep{azizzadenesheli2016reinforcement1}. The low-rank structure allows us to model and optimize the system using significantly fewer observations and less computation. Effective methods for estimating the low-rank Markov model would pave the way to better understanding of process data and more efficient decision making.

\subsection{Our approach}
We propose to estimate the low-rank Markov model based on an empirical trajectory of states, whose length is only proportional to the total number of states. 
We propose two approaches based on the maximum likelihood principle and low-rank optimization. The first approach uses a convex nuclear-norm regularizer to enforce the low-rank structure and a polyhedral constraint to ensure that optimization is over all probabilistic matrices. The second approach is to solve a rank-constrained optimization problem using difference-of-convex (DC) programming. For both approaches, we provide statistical upper bounds for the Kullback-Leibler (KL) divergence between the estimator and the true transition matrix as well as the $\ell_2$ risk. We also provide an information-theoretic lower bound to show that the proposed estimators are nearly rate-optimal. Note that the low-rank estimation of the Markov model was considered in \cite{zhang2018optimal} where a spectral method with total variation bound is given. In comparison, the novelty of our methods lies in the use of maximum likelihood principle and low-rank optimization, which allows us to obtain the first KL divergence bound for learning low-rank Markov models.

Our second approach involves solving a rank constraint optimization problem over probabilistic matrices, which is a refinement of the convex nuclear-norm approach. Due to the non-convex rank constraint, the optimization problem is difficult -  to the best of our knowledge, there is no efficient approach that directly solves the rank-constraint problem. In this paper, we develop a penalty approach to relax the rank constraint and transform the original problem into a DC (difference of convex functions) programming one. Furthermore, we develop a particular DC algorithm to solve the problem by initiating at the solution to the convex problem and successively refining the solution through solving a sequence of inner subproblems. Each subroutine is based on the multi-block alternating direction method of multipliers (ADMM). Empirical experiments show that the successive refinements through DC programming do improve the learning quality. As a byproduct of this research, we develop a new class of DC algorithms and a unified convergence analysis for solving non-convex non-smooth problems, which were not available in the literature to our best knowledge.  %The proposed DC programming algorithms and convergence analysis are given in Section \ref{sec-opt}.

\subsection{Contributions and paper outline}

The paper provides a full set of solutions for learning low-rank Markov models. The main contributions are: (1) We develop statistical methods for learning low-rank Markov model with rate-optimal Kullback-Leiber divergence guarantee for the first time; (2) We develop low-rank optimization methods that are tailored to the computation problems for nuclear-norm regularized and rank-constrained M-estimation; (3) A byproduct is a generalized DC algorithm that applies to nonsmooth nonconvex optimization with convergence guarantee. 

The rest of the paper is organized as follows. Section 2 surveys related literature. Section 3 proposes two maximum likelihood estimators based on low-rank optimization and establishes their statistical properties. Section 4 develops computation methods and establishes convergence of the methods. Section 5 presents the results of our numerical experiments.

%Finally, the performance of the proposed estimator and algorithm is verified through simulation studies and the dataset of Manhattan taxi trips. The experiment successfully recover the zoning of Manhattan without using any coordinate information. 
%We show that our method can be used in conjunction with spectral clustering to recover the zoning of Manhattan. 

\section{Related literature}

%This work draws motivation from several prior developments in system reduction, latent variable models and nonconvex optimization.

Model reduction for complicated systems has a long history. It traces back to variable-resolution dynamic programming \citep{moore1991variable} and state aggregation for decision process \citep{sutton1998reinforcement}. In the case of Markov process, \citep{deng2011optimal, deng2012model} considered low-rank reduction of Markov models with explicitly known transition probability matrix, but not the estimation of the reduced models. Low-rank matrix approximation has been proved powerful in analysis of large-scale panel data, with numerous applications including network analysis \citep{e2008optimal}, community detection \citep{newman2013spectral}, ranking \citep{negahban2016rank}, product recommendation \citep{keshavan2010matrix} and many more. The main goal is to impute corrupted or missing entries of a large data matrix. Statistical theory and computation methods are well understood in the settings where a low-rank signal matrix is corrupted with independent Gaussian noise or its entries are missed independently.

In contrast, our problem is to estimate the transition density functions from dependent state trajectories, where statistical theory and efficient methods are underdeveloped. When the Markov model has rank $1$, it becomes an independent process. In this case, our problem reduces to estimation of a discrete distribution from independent samples \citep{steinhaus1957problem,lehmann2006theory,han2015minimax}. For a rank-$2$ transition matrix, \cite{huang2016recovering} proposed an estimation method using a small number of independent samples. Very recently there have been some works on minimax learning of Markov chains. \citet{HOP18} derived the minimax rates of estimating a Markov model in terms of a smooth class of $f$-divergences. They considered the family of $\alpha$-minorated Markov chains, i.e., all the transition probabilities are greater than $\alpha$. \citet{WKo19} computed the finite-sample PAC-type minimax sample complexity of recovering the transition matrix from a state trajectory of a Markov chain, up to a tolerance in a total-variation-based (TV-based) metric. This TV-based metric does not belong to the family of the smooth $f$-divergences in \citet{HOP18}, and their class of Markov models strictly contains the class of the $\alpha$-minorated ones. Neither of these works considered low-rank Markov models though. 

The closest work to ours is \cite{zhang2018optimal}, in which a spectral method via truncated singular value decomposition was introduced and the upper and lower error bounds in terms of total variation were established. \cite{yang2017dynamic} developed an online stochastic gradient method for computing the leading singular space of a transition matrix from random walk data. To our best knowledge, none of the existing works has analyzed efficient recovery of a low-rank Markov model with Kullback-Leiber divergence guarantee.

Hidden Markov Models (HMMs) are closely related with our low-rank Markov models. Note that the observation trajectory of an HMM is not necessarily Markovian. Therefore, an HMM can be regarded as a relaxed variant of low-rank Markov models. There have been many works on estimating HMM, in particular through spectral approaches, e.g., \citet{hsu2012spectral} and \citet{anandkumar2014tensor}. A critical difference is: States are not fully observable in HMM, but are fully observable in low-rank Markov models. Although HMM is more general, the low-rank Markov model is more suitable for dynamical processes where the state space is large but fully observable, for which we will establish tighter error bounds. 
%Note that though the observations from an HMM do not have to follow an MC, each observation needs to be independent of all the historical observations and hidden states given the current hidden state. This condition may contradict practical cases as illustrated in Section \ref{sec:taxi}. 

%The online stochastic approximation algorithms were also analyzed recently in \cite{li2017diffusion,Li2018nearoptimal} for principal component estimation.
%The estimation of discrete distributions in the form of vectors is another line of related research, where the procedure and estimation risk has been considered in both classic and recent literature \citep{steinhaus1957problem,lehmann2006theory,han2015minimax}.

On the optimization side, we adopt DC programming to handle the rank constraint and replace it with the difference of two convex functions. DC programming was first introduced by \cite{tao1997convex} and has become a prominent tool for handling a class of nonconvex optimization problems (see also \cite{tao2005dc,le2012exact,le2017stochastic,lethi2018}). 
%In subsection \ref{sec:DCA}, we study a general class of DCA. 
%As one shall see in Section \ref{sec:DCA}, our algorithm enjoys more flexibility than the classic DCA.  Especially, a unified convergence analysis is also provided for the proposed algorithm.  %when $f\equiv 0$.
%Our algorithm and analysis are closely related to \cite{van2015convergence,wen2017proximal}. Similar optimization models are considered in both works. However, the smooth part in the objective function is required to be a convex function in \cite{wen2017proximal}.
%While both works use the majorization technique, the choices of majorization functions in their works are quite restricted. 
%%Our approach is more flexible than the others since we allow a properly chosen  majorization function to capture the Hessian information of the smooth function.   
%Moreover, none of these two works consider the introduction of the indefinite proximal terms. 
%The {majorized indefinite-proximal DC algorithm} considered in this paper is closely related to the ones in \cite{van2015convergence,wen2017proximal}. %, where similar optimization models are also analyzed. 
In particular, \citet{van2015convergence} and \cite{wen2017proximal} considered the {majorized DC algorithm}, which motivated the optimization method developed in this paper.
However, both \cite{van2015convergence} and \citet{wen2017proximal} used the majorization technique with restricted choices of majorants, and neither considered the introduction of the indefinite proximal terms. In addition, \citet{wen2017proximal} further assumes the smooth part in the objective to be convex. 
%However, the existing DC methods do not work for our problem. 
In comparison with the existing methods, our DC programming method applies to nonsmooth problems and is compatible with a more flexible and possibly indefinite proximal term.

Finally, we would like to mention the probabilistic tools we used to derive the statistical results. Recent years have witnessed many works on measure concentration of dependent random variables, e.g., \citet{Mar96, Kon07, KRa08, Pau15, JFS18}, etc. Nevertheless, these results do not suffice to establish the desired statistical guarantee, because exploiting low-rank structure requires studying the concentration of a matrix martingale in terms of the spectral norm, as shown in Lemma \ref{lem:gradient}. The matrix Freedman inequality \citep[][Corollary~1.3]{tropp2011freedman} turns out to be the right tool for analyzing the concentration of the matrix martingale. We also used an variant of Bernstein's inequality for general Markov chains \citep[][Theorem~1.2]{JFS18} to derive an exponential tail bound for the status counts of the Markov chain $\cX$.

\section{Minimax rate-optimal estimation of low-rank Markov chains}
\label{sec:StatProperty}

Consider an ergodic Markov chain $\cX = \{X_0, X_1, \ldots, X_n\}$ on $p$ states $\cS=\{s_j\}_{j=1}^p$ with the transition probability matrix $\bP\in\mathbb{R}^{p\times p}$ and stationary distribution $\pi$, where $P_{ij} = \PP(X_1 = s_j|X_0 = s_i)$ for any $i, j \in [p]$. Let $\pi_{\min} := \min_{j \in [p]} \pi_j$ and $\pi_{\max} := \max_{j \in [p]} \pi_j$. We quantify the distance between two transition matrices $\bP$ and $\widehat{\bP}$ in Frobenius norm $\|\widehat{\bP} - \bP\|_{\rm F} = \bigl\{\sum_{i,j  = 1}^p (\widehat P_{ij} - P_{ij})^2\bigr\}^{1/2}$ and Kullback--Leibler divergence  $\DKL(\bP, \widehat{\bP}) = \sum_{i, j=1}^p \pi_iP_{ij}\log(P_{ij}/ \widehat P_{ij})1_{\{P_{ij} \neq 0\}}$.
Suppose that the unknown transition matrix $\bP$ has a small rank $r \ll p$. Our goal is to estimate the transition matrix $\P$ via a state trajectory of length $n$. 

\subsection{Spectral gap of nonreversible Markov chains}
We first introduce the {\it right $\cL_2$-spectral gap} of $\bP$ \citep{Fil91, JFS18}, a quantity that determines the convergence speed of the Markov chain $\cX$ to its invariant distribution $\pi$. 
%This quantity determines the `effective' sample size in statistical rate of our proposed M-estimators. 
Let $\cL_2(\pi):= \{h \in \Re^p: \sum_{j \in [p]} h_j^2 \pi_j < \infty\}$ be a Hilbert space endowed with the following inner product: 
\[
\inn{h_1, h_2}_{\pi} := \sum_{j \in [p]} h_{1j} h_{2j}\pi_j. 
\]
The matrix $\bP$ induces a linear operator on $\cL_2(\pi)$: $h \mapsto \bP h$, which we abuse $\bP$ to denote. Let $\bP^*$ be the adjoint operator of $\bP$ with respect to $\cL_2(\pi)$:
$$\P^* = \diag(\pi)^{-1} \P^\top \diag(\pi).$$ 
Note that the following four statements are equivalent: (a) $\P$ is self-adjoint; (b) $\P^\ast = \P$; (c) the detailed balance condition holds: $\pi_i P_{ij} = \pi_j P_{ji}$; (d) the Markov chain is reversible. In our analysis, we do {\it not} require the Markov chain to be reversible. We therefore introduce the {\it additive reversiblization of $\bP$}: $(\bP + \bP^*) / 2$, which is a self-adjoint operator on $\cL_2(\pi)$ and has the largest eigenvalue as 1. The right spectral gap of $\bP$ is defined as follows: 
\begin{definition}[Right $\cL_2$-spectral gap]
	We say  the right $\cL_2$-spectral gap of $\bP$ is $1 - \rho_+ $ if 
	\[
	\rho_+:= \sup_{\inn{h, 1}_{\pi} = 0, \inn{h, h}_{\pi} = 1} \frac{1}{2}\inn{(\bP + \bP^*)h, h}_{\pi} < 1, 
	\]
	where $1$ in $\inn{h, 1}$ refers to the all-one $p$-dimensional vector. %In the sequel, we use $\rho_{+}$ to denote $\rho_{+} (\bP)$ for simplicity. 
\end{definition}
Define the $\epsilon$-mixing time of the Markov chain $\cX$ as
\[
\tau(\epsilon) := \min\{t: \max_{j \in [p]} \|(\bP^t)_{j\cdot} - \pi\|_{\mathrm{TV}} \le \epsilon\}, 
\]
where $\|(\bP^t)_{j\cdot} - \pi\|_{\mathrm{TV}} := 2 ^ {-1} \lonenorm{(\bP^t)_{j\cdot} - \pi}$ is the total variation distance between $\bP ^ t_{j \cdot}$ and $\pi$. For reversible and ergodic Markov chains, \citet[][Theorem~12.3]{LPe17} show that 
\be
	\label{eq:mixing_gap}
	\tau(\epsilon) \le \frac{1}{1 - \rho_+}\log\biggl(\frac{1}{\epsilon \pi_{\min}}\biggr), 
\ee
which implies that the larger the spectral gap is, the faster the Markov chain converges to the stationary distribution.

\subsection{Estimation methods and statistical results}
\label{sec:main_results}
%\subsection{Nuclear-norm regularized M-estimation}

Now we are in position to present our methods and statistical results. Given the trajectory $\{X_1,\ldots,X_n\}$, we count the number of times that the state $s_i$ transitions to $s_j$:
\[ n_{ij}: = \left|\{1\le k\le n \mid \, X_{k-1} = s_i, X_k = s_j\}\right|.\]
Let $n_i: = \sum_{j=1}^{p} n_{ij}$ for $i=1,\ldots, p$ and $n := \sum_{i=1}^p n_i$. The averaged negative log-likelihood function of $\P$ based on the state-transition trajectory $\{x_0, \ldots, x_n\}$ is
\begin{equation}\label{eq:log-likelihood}
{\ell_n}(\P):= - \frac{1}{n} \sum_{k = 1}^n \log (\inn{\bP, \bX_k}) = -\frac{1}{n} \sum_{i=1}^{p}\sum_{j=1}^{p} n_{ij}\log(P_{ij}), 
\end{equation}
where $\bX_k := e_ie_j^\top \in \Re^{p \times p}$ if $x_k = s_i$ and $x_{k + 1} = s_j$. We first impose the following assumptions on $\bP$ and $\pi$.  %\mw{mw: maybe this assumption can be moved to the beginning of the next section}
\begin{assumption}
	\label{asp:1}
	(i) $\rank(\bP) = r$; (ii) There exist some positive constants $\alpha, \beta > 0$ such that for any $1 \le j, k \le p$, $P_{jk} \in \{0\} \cup [\alpha / p, \beta / p]$. 
%	; (iii) $\pi_j \le \beta / p$ for all $1 \le j \le p$. 
\end{assumption}
\begin{remark}
The entrywise constraints on $\bP$ are imposed by our theoretical analysis and may not be necessary in practice. 
%The element-wise upper bound constraint, $P_{ij} \leq \beta/p$, is only used for the theoretical analysis; the element-wise lower bound constraint, $P_{ij} \geq \alpha/p$, ensures that the estimator has positive sign and reasonable upper bound in KL-divergence loss. \mw{need more justification here}
	Specifically, the upper and lower bounds for the nonzero entries of $\bP$ %(Assumption \ref{asp:1}) in Theorems \ref{thm:nuclear} and \ref{thm:rank} 
	ensure that (i) the gradient of the log-likelihood $\nabla \ell_n(\bP)$ is well controlled and exhibits exponential concentration around its population mean (see \eqref{eq:log_likelihood} for the reason we need $\alpha$ there); 
%	(ii) $\ell_n(\bQ)$ enjoys restricted strong convexity around $\bP$ as characterized in Lemma \ref{lem:restricted_strong_convexity} (see \eqref{eq:rsc_first_step} for the reason we need $\beta$ there);
	(ii) the converter between the $\ell_2$-risk $\fnorm{\widehat \bP - \bP}$ ($\fnorm{\widehat \bP^r - \bP}$ resp.) and the KL-divergence $D_{\KL}(\bP, \widehat \bP)$ ($D_{\KL}(\bP, \widehat \bP^r)$ resp.) depends on $\alpha$ and $\beta$, as per Lemma \ref{lem:kl_to_l2}. The entry-wise upper and lower bounds are common in statistical analysis of count data, e.g., Poisson matrix completion \citep[Equation (10)]{cao2016poisson}, Poisson sparse regression \citep[Assumption 2.1]{jiang2015minimax}, point autoregressive model \citep[Definition of $\mathcal{A}_s$]{hall2016inference}, etc.
\end{remark}%\mw{MW: shall this remark be moved? so we explain the alpha beta bounds in one place?} 

\begin{remark}
	%If we remove $0$ in the feasible range of $P_{jk}$, then $P_{jk} \ge \alpha / p$ for all $j, k \in [p]$. This is called $(\alpha/p)$-minoration condition and it implies strong mixing. 
	If we remove $0$ in the feasible range of $P_{jk}$, we obtain the $(\alpha/p)$-minoration condition: $P_{jk} \ge \alpha / p$ for all $j, k \in [p]$. The $(\alpha/p)$-minoration condition implies strong mixing since combining \citet[][pp. 237-238]{Bre99} and \citet[][Lemma~2.2.2]{Kon07} yields $1 - \rho_+ \ge \alpha$ and we can deduce that $\tau(\epsilon) \le \alpha^{-1}\log\{(\epsilon \pi_{\min})^{-1}\}$ given \eqref{eq:mixing_gap}. 
\end{remark}

Next we propose and analyze a nuclear-norm regularized maximum likelihood estimator (MLE) of $\bP$ defined as follows: 
\begin{equation}
\label{prob:convex-nuclear}
\begin{array}{rllll}
%\widehat{\P} =\argmin L(\Q)  \; \mbox{s.t.}~~ \Q \geq 0, \Q{\bf 1}_p = {\bf 1}_p, {\rm rank}(\Q) \le r.
%  
\widehat{\P} := & \argmin ~\ell_n(\Q) + \lambda \nnorm{\Q}\\
\mbox{s.t.} & \Q 1_p = 1_p,\quad \alpha / p \le Q_{ij}\le \beta / p,\quad \forall\, 1\le i,j\le p, 
\end{array}
\end{equation}
where $\lambda>0$ is a tuning parameter. Note that we cannot allow $\bQ$ to have zero entries as in Assumption \ref{asp:1}, because otherwise we may have that $\widehat P_{ij} = 0$ and $P_{ij} > 0$ for some $(i, j)$, violating the requirement of the definition of $\DKL(\bP, \widehat\bP)$. Our first theorem shows that with an appropriate choice of $\lambda$, $\widehat \bP$ exhibits a sharp statistical rate. For simplicity, from now on we say $a \gtrsim b$ ($a \lesssim b$) if there exists a universal constant $c > 0$ ($C > 0$) such that $a \ge cb$ ($a \le Cb$). 

\begin{theorem}[Statistical guarantee for the nuclear-norm regularized estimator]
	\label{thm:nuclear}
	Suppose the initial state $X_0$ is drawn from the stationary distribution $\pi$ and Assumption \ref{asp:1} holds. There exists a universal constant $C_1 > 0$, such that for any $\xi > 1$, if we choose 
	\[
		\lambda =  C_1 \biggl\{\biggl(\frac{\xi p ^ 2\pi_{\max} \log p}{n\alpha}\biggr)^{\! 1 / 2} + \frac{\xi p\log p}{n \alpha}\biggr\}, 
	\]
	then whenever $n\pi_{\max}(1 - \rho_+) \ge \max\{\max(20, \xi ^ 2) \log p, \log n\}$, we have that 
%	\[
%	\begin{aligned}
%		\PP \biggl( \DKL(\bP, \widehat\bP) \gtrsim \frac{\xi r\pi_{\max}\beta ^ 2 p \log p}{\pi_{\min}\alpha ^ 3n} + \frac{\xi\alpha ^ 2}{rp ^ 2 \pi_{\max}\log p}\biggr) \lesssim e^{- \xi} + p^{-(\xi - 1)} + p^{-10},
%	\end{aligned}
%	\]
	\[
	\begin{aligned}
		\PP \biggl( \DKL(\bP, \widehat\bP) \gtrsim \frac{\xi r\pi_{\max}\beta ^ 2 p \log p}{\pi_{\min}\alpha ^ 3n} + \frac{\xi\pi_{\min}}{rp\pi_{\max}\log p}\biggr) \lesssim e^{- \xi} + p^{-(\xi - 1)} + p^{-10},
	\end{aligned}
	\]
	and that 
%	\[
%		\begin{aligned}
%		\PP \biggl( \fnorm{\widehat \bP - \bP} ^ 2 \gtrsim \frac{\xi r\pi_{\max}\beta ^ 4\log p}{\pi_{\min} ^ 2\alpha ^ 4n} + \frac{\xi \alpha\beta ^ 2}{rp ^ 2 \pi_{\max}\pi_{\min} \log p}\biggr) \lesssim e^{- \xi} + p^{-(\xi - 1)} + p^{-10}. 
%		\end{aligned}
%	\]	
	\[
		\begin{aligned}
		\PP \biggl( \fnorm{\widehat \bP - \bP} ^ 2 \gtrsim \frac{\xi r\pi_{\max}\beta ^ 4\log p}{\pi_{\min} ^ 2\alpha ^ 4n} + \frac{\xi\beta ^ 2}{\alpha rp ^ 2 \pi_{\max} \log p}\biggr) \lesssim e^{- \xi} + p^{-(\xi - 1)} + p^{-10}. 
		\end{aligned}
	\]	

\end{theorem}

%\begin{remark}
%		Theorem \ref{thm:nuclear} suggests that the `effective' sample size of learning the Markov model is $n(1 - \rho_+)$, where the discount factor is the spectral gap of the true Markov kernel. When $\rho_+ = 0$, the  Markov model has rank $1$ and our result reduces to the typical results under independent sampling scheme. 
%	\end{remark}

\begin{remark}
		When $n \lesssim \{rp \pi_{\max}(\log p) \beta/ (\pi_{\min} \alpha ^ {\!3 / 2})\} ^ 2$, the second terms of both the KL-Divergence and Frobenius-norm error bounds are dominated by the first terms respectively, so that 
		\be
			\label{eq:hd_result}
			\DKL(\bP, \widehat\bP) = O_{\PP}\biggl(\frac{r\pi_{\max}\beta ^ 2 p \log p}{\pi_{\min}\alpha ^ 3n}\biggr)~\text{and}~\fnorm{\widehat\bP - \bP} ^ 2 = O_{\PP}\biggl(\frac{r\pi_{\max}\beta ^ 4\log p}{\pi_{\min} ^ 2\alpha ^ 4n}\biggr).
		\ee
		When $\alpha \asymp \beta$ and $\pi_{\max} \asymp \pi_{\min}$, we have that $\alpha, \beta \asymp 1$ and that $\pi_{\max}, \pi_{\min} \asymp 1 / p$. Therefore, $\DKL(\bP, \widehat \bP) = O_{\PP}(rp\log p / n)$ and $\fnorm{\widehat \bP - \bP} ^ 2 = O_{\PP}(rp\log p / n)$. These rates are consistent with those derived in the literature of low-rank matrix estimation \citep{NWa11, KLT11}. For a big $n$, the current error bounds are sub-optimal: the second terms of the bounds are independent of $n$ and thus do not converge to zero as $n$ goes to infinity. These terms are due to the requirement of the uniform concentration of $\widetilde D_{\KL}(\bP, \widehat\bP)$, the empirical counterpart of $\DKL(\bP, \widehat \bP)$, to $\DKL(\bP, \widehat \bP)$ (see Lemma \ref{lem:uniform_law}). We eliminate these trailing terms through an alternative proof strategy in Section \ref{sec:alt}, though the resulting statistical rates have worse dependence on $\alpha$ and $\beta$ and are thus relegated to the appendix.
\end{remark}

\begin{remark}
	When $r = 1$, $\bP$ can be written as $1v^\top$ for some vector $v\in \Re^{p}$, and then estimating $\bP$ essentially reduces to estimating a discrete distribution from multinomial count data. The first term of the upper bounds in Theorem \ref{thm:nuclear} nearly matches (up to a log factor) the classical results of discrete distribution estimation $\ell_2$ risks (see, e.g., \citet[Pg. 349]{lehmann2006theory}). 
%	 \mw{mw: Is there any literature on KL divergence for estimating discrete distribution? Those results may also assume entrywise upper and lower bounds, so they can justify our assumption} \red{Maybe discussing \cite{kamath2015learning}?}
	%	\textcolor{red}{To Anru: (1) Could you provide a minimax lower bound as Theorem 2? It can be a straightforward extension from our compression paper; (2) Is it easy to add a result on recovery of the subspace? That could be our Theorem 3. (3) Discuss the results of Theorem 1 in comparison with existing results on estimating discrete distributions. }
\end{remark}

Next we move on to the second approach -- using rank-constrained MLE to estimate $\bP$: 
\begin{equation}\label{prob:nonconvex-lowrank}
\begin{array}{rllll}
\widehat{\P}^r  :=  &  \argmin {\ell_n}(\Q) \\
 \mbox{s.t.} & \Q 1_p = 1_p,\quad \alpha / p \le Q_{ij}\le \beta / p,\quad \forall\, 1\le i,j\le p,\quad \text{rank}(\Q) \le r. 
\end{array}
\end{equation}
Similarly to \eqref{prob:convex-nuclear}, we cannot allow $\bQ$ to have zero entries. In contrast to $\widehat \bP$, the rank-constrained MLE $\widehat \bP^r$ enforces the prior knowledge ``$\bP$ is low-rank" exactly without inducing any additional bias. It requires solving a non-convex and non-smooth optimization problem, for which we will provide an algorithm based on DC programming in Section \ref{sec:6}. Here we first present its statistical guarantee. 

\begin{theorem}[Statistical guarantee for the rank-constrained estimator]
	\label{thm:rank}
		Suppose that Assumption \ref{asp:1} holds and that $n\pi_{\max} (1 - \rho_+) > \max(20 \log p, \log n)$. There exist universal constants $C_1, C_2 > 0$ such that for any $\xi > 0$, 
		\[
			\PP\biggl\{D_{\KL}(\bP, \widehat \bP^r) \ge \max\biggl(\frac{C_1 r\pi_{\max} \beta^2 p \log p }{\pi_{\min}\alpha ^ 3n}, \frac{\xi \pi_{\min}}{rp \pi_{\max} \log p}\biggr)\biggr\} \le C_2e^{- \xi}, 
		\]
		and 
		\[
			\PP\biggl\{\fnorm{\widehat \bP^r - \bP} ^ 2 \ge \max\biggl(\frac{C_1 r\pi_{\max} \beta^4 \log p }{\pi ^ 2_{\min}\alpha ^ 4n}, \frac{\xi \beta ^ 2}{\alpha rp ^ 2 \pi_{\max}\log p}\biggr)\biggr\} \le C_2 e^{-\xi}. 
		\]
%	Furthermore, 
%	\[
%	\begin{aligned}
%	D_{\KL}(\widehat\bP^r, \bP) \le \frac{C^2\xi^2\beta^2}{\alpha^2} \cdot \frac{ rp\log p}{(1 - \rho_+)n}. 
%	\end{aligned}			
%	\]
\end{theorem}

\begin{remark}
	The proof of the rank constrained method requires fewer inequality steps and is more straightforward than the that of the nuclear method. Although our upper bounds of the nuclear norm regularized method and the rank constrained one have the same rate, the difference of their proofs may implicitly suggest the advantage of the rank constrained method in the constant, as futher illustrated by our numerical studies.
\end{remark}%\mw{MW: this remark can be moved to the result subsection. Readers might skip the proof section.}
%\mw{MW: add a comment about theorems 1,2 match and why the rank-estimator might perform better in practice}

To assess the quality of the established statistical guarantee, we further provide a lower bound result below. It shows that when $\alpha, \beta$ are constants, both estimators $\widehat\bP$ and $\widehat \bP^r$ are rate-optimal up to a logarithmic factor. Informally speaking, they are not improvable for estimating the class of rank-$r$ Markov chains. 
\begin{theorem}[Minimax error lower bound for estimating low-rank Markov models]
	\label{thm:lower_bound}
	Consider the following set of low-rank transition matrices
	$$\Theta := \bigl\{\bP: \forall j, k \in [p], P_{jk} \in \{0\} \cup [\alpha / p, +\infty),~\bP1_p = 1_p,~\rank(\bP) \le r\bigr\}.$$
	There exists a universal constant $c > 0$ such that when $p(r -1) \ge 192\log 2$, we have
	\[
		\inf_{\widehat{\bP}}\sup_{\bP\in \Theta}\mathbb{E}\|\widehat{\bP} - \bP\|_F^2 \geq \frac{cp(r - 1)}{n\alpha}. 
	\]
%	\begin{equation*}
%	\begin{split}
%	\inf_{\widehat{\bP}}\sup_{\bP\in \Theta}\mathbb{E} \{D_{\KL}(\bP, \widehat{\bP})\} \geq c\frac{rp}{n}~~~\text{and}
%	\end{split}
%	\end{equation*}
%	\begin{equation*}
%	\begin{split}
%\qquad
%	\inf_{\widehat{\bP}}\sup_{\bP\in \Theta}\mathbb{E}\|\widehat{\P} - \P\|_F^2 \geq c\frac{p(r - 1)}{n\alpha}. 
%	\end{split}
%	\end{equation*}
\end{theorem}
\begin{remark}	
	Theorem \ref{thm:lower_bound} shows that a smaller $\alpha$ makes the estimation problem harder. It still remains to be an open problem whether $\beta$ in Assumption \ref{asp:1} should be in this minimax risk or not. 
\end{remark}
%as long as $n = O(p^2r^2\log p)$

Besides the full transition matrix $\bP$, the leading left and right singular vectors of $\bP$, denoted by $\bU, \bV \in \OO^{p\times r}$, also play important roles in Markov chain analysis. For example, performing $k$-means on reliable estimate of $\bU$ or $\bV$ can give rise to state aggregation of the Markov chain \citep{zhang2018optimal}. In the following, we further establish the statistical rate of estimating the singular subspace of the Markov transition matrix, based on the previous results. 

\begin{theorem}\label{thm:uv}
	Under the setting of Theorem \ref{thm:nuclear}, let $\widehat{\bU}, \widehat \bV \in \OO^{p\times r}$ be the left and right singular vectors of $\widehat{\bP}$  respectively. Then there exist universal constants $C_1, C_2$, such that for any $\xi > 0$, we have that 
	\begin{equation*}
	\max\left\{\|\sin\Theta(\widehat{\bU}, \bU)\|_F^2, \|\sin\Theta(\widehat{\bV}, \bV)\|_F^2\right\} 
	\le  \min\biggl\{\max\biggl(\frac{C r\pi_{\max} \beta^4\log p }{\pi ^ 2_{\min}\alpha ^ 4n\sigma_r^2(\bP)}, \frac{\xi \beta ^ 2}{\alpha rp ^ 2 \pi_{\max}(\log p)\sigma_r^2(\bP) }\biggr), r\biggr\}
	\end{equation*}
	with probability at least $1 - C_2(e^{- \xi} + p^{-(\xi - 1)} + p^{-10})$. Here, $\sigma_r(\bP)$ is the $r$-th largest singular value of $\bP$ and $\|\sin\Theta(\widehat{\bU}, \bU)\|_F := (r - \|\widehat{\bU}^\top\bU\|_F^2)^{1/2}$ is the Frobenius norm $\sin \Theta$ distance between $\widehat{\bU}$ and $\bU$.
\end{theorem}

%The theorem below gives a KL-divergence bound of $\widehat\bP^r$. 
%
%\begin{theorem}
%		\label{thm:kl_rank}
%		Suppose that $n\pi_{\max}(1 - \rho_+)\ge \max (20 \log p, \log n)$. There exist universal constants $C_1, C_2 > 0$, such that for any $\xi > 1$, 
%		\[
%		\PP\biggl\{D_{\KL}(\bP, \widehat \bP^r) \ge \max\biggl(\frac{C_1 r\pi_{\max} \beta^2 p^2 \log p }{\alpha ^ 4n}, \frac{\alpha ^ 2\xi}{rp^ 2 \pi_{\max} \log p}\biggr)\biggr\} \le C_2\exp(-\xi). 
%		\]
%\end{theorem}

%\begin{remark}
%	As in Theorem \ref{thm:lower_bound}, one can develop the corresponding lower bound to show the minimax optimality for the upper bound in Theorem \ref{thm:lower_bound_uv}.
%\end{remark}

\subsection{Proof outline of Theorems \ref{thm:nuclear}, \ref{thm:rank}}
\label{sec:statistical_analysis}
%\red{(I think these discussions on the proof procedure are perfect if we target journal of machine learning research. If we want to try Operations Research, do you think it would be better to postpone them to a separate sections later?)} 
In this section, we elucidate the roadmap to proving Theorems \ref{thm:nuclear} and \ref{thm:rank}. Complete proofs are deferred to the supplementary materials. We mainly focus on Theorem \ref{thm:nuclear} for the nuclear-norm penalized MLE $\widehat \bP$, as we use similar strategies to prove Theorem \ref{thm:rank}. 

We first show in the forthcoming Lemma \ref{lem:large_lambda} that when the regularization parameter $\lambda$ is sufficiently large, the statistical error $\widehat \bDelta := \widehat \bP - \bP$ falls in a restricted nuclear-norm cone. This cone structure is crucial to establishing strong  statistical guarantee for estimation of low-rank matrices with high-dimensional scaling \citep{NWa11}. Define a linear subspace $\cN := \{ \bQ: \bQ 1_p = 1_p \}$ 
%\[
%\cN := \biggl\{ \bQ: \sum\limits_{k = 1}^p P_{jk}Q_{jk} = 0 , \forall j = 1, \ldots, p \biggr\}
%\]
and denote the corresponding projection operator by $\Pi_{\cN}$. In other words, for any $\bQ \in \cN$ and any $j = 1, \ldots, p$, the summation of all the entries in the $j$th row of $\bQ$ equals one. One can verify that for any $\bQ \in \Re^{p \times p}$, $\Pi_{\cN}(\bQ) = \bQ - \bQ 11^\top / p$. Let $\bP=\bU\bD\bV^\top$ be an SVD of $\bP$, where $\bU, \bV\in \Re^{p \times r}$ are orthonormal and the diagonals of $\bD$ are in the non-increasing order. Define
\[
\begin{aligned}
	& \cM:=\{\bQ \in \Re^{p \times p}\ |\ \text{row}(\bQ)\subseteq \text{col}(\bV), \text{col}(\bQ)\subseteq \text{col}(\bU)\}, \\
	& \overline{\cM}^{\perp}:=\{\bQ \in \Re^{p \times p}\ |\ \text{row}(\bQ)\perp \text{col}(\bV), \text{col}(\bQ)\perp \text{col}(\bU)\},
\end{aligned}
\]
where col$(\cdot)$ and row$(\cdot)$ denote the column space and row space respectively. We can write any $\bDelta\in \Re^{p \times p}$ as
\[
\bDelta=[\bU, \bU^\perp]\left[
\begin{array}{cc}
\bGamma_{11} & \bGamma_{12} \\
\bGamma_{21} & \bGamma_{22}
\end{array} \right] [\bV, \bV^\perp]^\top.
\]
Define $\bDelta_{\cW}$ as the projection of $\bDelta$ onto any Hilbert space $\cW\subseteq\Re^{p \times p}$. Then,
\be
\begin{aligned}
	\bDelta_{\cM} =\bU\bGamma_{11}{\bV}^\top,\quad\bDelta_{\overline\cM^\perp} = \bU^{\perp} \bGamma_{22}(\bV^{\perp})^{\top}, \quad \bDelta_{\overline\cM} = [\bU, \bU^{\perp}]\left[
	\begin{array}{cc}
		\bGamma_{11} & \bGamma_{12} \\
		\bGamma_{21} & \bzero
	\end{array} \right] [\bV, \bV^{\perp}]^{\top}.
\end{aligned}
\ee
The lemma below shows that $\widehat \bDelta := \widehat\bP - \bP$ falls in a nuclear-norm cone if $\lambda$ is sufficiently large. 
\begin{lemma}
	\label{lem:large_lambda}
	If $\lambda\ge 2 \opnorm{\Pi_{\cN}(\nabla \ell_n(\bP))}$ in \eqref{prob:convex-nuclear}, then we have that 
	\[
		\nnorm{ \widehat\bDelta_{\overline\cM^\perp}} \le 3\nnorm{\widehat\bDelta_{\overline\cM}} + 4\nnorm{\bP_{\cM^\perp}}. 
	\]
	In particular, when $\bP \in \cM$, we have that $\nnorm{\bP_{\cM^{\perp}}} = 0$ and that
	\be
	\label{eq:converter}
	\nnorm{\widehat \bDelta} \le \nnorm{\bDelta_{\overline\cM^{\perp}}} + \nnorm{\bDelta_{\overline \cM}} \le 4 \nnorm{\widehat \bDelta_{\overline \cM}} \le 4(2r)^{1 / 2}\fnorm{\widehat \bDelta}.
	\ee
\end{lemma}

Lemma \ref{lem:large_lambda} implies that the converting factor between the nuclear and Frobenius norms of $\widehat \bDelta$ is merely $4(2r)^{1 / 2}$ when $\P\in \mathcal{M}$, which is much smaller than the worst-case factor $p^{1 / 2}$ between nuclear and Frobenius norms of general $p$-by-$p$ matrices. This property of $\widehat \bDelta$ is one cornerstone for establishing Theorem \ref{thm:nuclear}.

\begin{comment}
\begin{remark}
	When $\bP \in \cM$, $\nnorm{\bP_{\cM^{\perp}}} = 0$ and Lemma \ref{lem:large_lambda} implies that 
\be
	\label{eq:converter}
	\nnorm{\widehat \bDelta} \le \nnorm{\bDelta_{\overline\cM^{\perp}}} + \nnorm{\bDelta_{\overline \cM}} \le 4 \nnorm{\widehat \bDelta_{\overline \cM}} \le 4\sqrt{2r}\fnorm{\widehat \bDelta}.
\ee
Thus the converting factor between the nuclear and Frobenius norms of $\widehat \bDelta$ is merely $4\sqrt{2r}$, which is much smaller than the worst-case factor $\sqrt{p}$ between nuclear and Frobenius norms of general $p$-by-$p$ matrices. This property of $\widehat \bDelta$ is one cornerstone for establishing Theorem \ref{thm:nuclear} (see \eqref{eq:stat_error_fnorm} for details).
\end{remark}
\end{comment}

%\red{(This Lemma seems overlaps with the description right before Lemma 3. Consider to merge them?)
%\begin{remark}
%	 In Lemma \ref{lem:restricted_strong_convexity}, it holds with high probability that the loss function $\ell_n(\bP + \bDelta)$ is strongly convex with respect to $\bDelta$ if $\bDelta$ satisfies \eqref{eq:converter}. Combining this with Lemma \ref{lem:large_lambda}, one can bound $\fnorm{\widehat \bDelta}$ by $|\ell_n(\widehat \bP) - \ell_n(\bP)|$, which motivates the key step \eqref{eq:stat_error_fnorm} in the proof of Theorem \ref{thm:nuclear}. 
%\end{remark}
%}

Next, we derive the rate of $\opnorm{\Pi_{\cN}(\nabla \ell_n(\bP))}$ to determine the order of $\lambda$ that ensures the condition of Lemma \ref{lem:large_lambda} to hold. 

\begin{lemma}
	\label{lem:gradient} 
	Under Assumption \ref{asp:1}, whenever $n \pi_{\max} (1 - \rho_+) \ge 2 \log p$, for any $\xi > 1$,  
	\[
		\PP\biggl\{ \opnorm{\Pi_{\cN}(\nabla \ell_n(\bP))} \gtrsim \biggl(\frac{\xi p ^ 2\pi_{\max}\log p}{n\alpha}\biggr)^{1 / 2} + \frac{\xi p\log p}{n \alpha}\biggr\} \le 4p^{-(\xi - 1)} + \exp\biggl(- \frac{n\pi_{\max} (1 - \rho_+)}{2}\biggr). 
	\]
\end{lemma}

\begin{remark}
	Lemma \ref{lem:gradient} is essentially due to concentration of a matrix martingale. Many existing results on measure concentration of dependent random variables \citep{Mar96, Kon07, KRa08, Pau15} are not directly applicable because of the matrix structure of $\nabla \ell_n(\bP)$. The main probabilistic tool we use here is the matrix Freedman inequality \citep[][Corollary~1.3]{tropp2011freedman} that characterizes concentration behavior of a matrix martingale (See \eqref{eq:matrix_freedman} for details). We notice two recent works, \citet{WKo19} and \citet{WKon19}, that use the same matrix Freedman inequality. Specifically, \citet{WKon19} applied the matrix Freedman inequality to derive a confidence interval for the mixing time of a Markov chain based on its single trajectory, and \citet{WKo19} used the same inequality to establish an upper bound for the sample complexity of learning a Markov chain. Finally, we also use an variant of Bernstein's inequality for general Markov chains \citep[][Theorem~1.2]{JFS18} to derive an exponential tail bound for the status counts of the Markov chain $\cX$ (See \eqref{eq:mc_bernstein} for details).
\end{remark}

Let $\cC := \{\bQ \in \R^{p \times p}: \nnorm{\bQ - \bP} \le 4\times 2^{1 /2} \fnorm{\bQ - \bP}, \bQ 1_p = 1_p, \alpha / p \le Q_{jk} \le \beta / p, \forall (j, k) \in  [p] \times [p]\}$. For any $\bQ \in \cC$, define $\cL(\bQ) := \mathbb{E} \{- \log (\inn{\bQ, \bX_i})\}$ and $\ell_n(\bQ) := n^{-1}\sum_{i = 1}^n -\log(\inn{\bQ, \bX_i})$. Recall that $\DKL(\bP, \bQ) = \cL(\bQ) - \cL(\bP) = \sum_{i=1}^p \pi_i \DKL(P_{i\cdot}, Q_{i\cdot}) = \sum_{i = 1}^p \sum_{j=1}^p \pi_iP_{ij} \log(P_{ij}/Q_{ij})$. Define the empirical KL divergence of $\bQ$ from $\bP$ as
\[
\widetilde{D}_{\KL}(\bP,\bQ) := \frac{1}{n}\sum_{i=1}^n \langle \log (\bP) -\log(\bQ), \bX_i\rangle = \ell_n(\bQ) - \ell_n(\bP). 
\]
The final ingredient of the analysis is the uniform converegence of $\widetilde D_{\KL}(\bP, \bQ)$ to $D_{\KL}(\bP, \bQ)$ when $\DKL(\bP, \bQ)$ is large. 
\begin{lemma}
	\label{lem:uniform_law}
	Suppose that $n\pi_{\max}(1 - \rho_+) \ge \max(20\log p, \log n)$. For any $\eta > \pi_{\min}  / (2\pi_{\max} rp \log p)$, define $\cC(\eta) := \{\bQ \in \cC: \DKL(\bP, \bQ) \ge \eta\}$. Then there exist universal constants $C_1, C_2 > 0$ such that
	\begin{equation}\label{ineq:to-show-1}
	\begin{aligned}
	\PP\biggl\{\forall \bQ \in \cC(\eta), ~|\widetilde D_{\KL}(\bP, \bQ) - D_{\KL}(\bP, \bQ)| \leq \frac{1}{2}\DKL(\bP, \bQ) + &\frac{C_1\pi_{\max} \beta ^ 2rp\log p}{\pi_{\min}\alpha ^ 3n}\biggr\} \\
	& \geq 1- C_2 \exp\biggl(- \frac{\eta \pi_{\max} rp \log p}{\pi_{\min}}\biggr).
	\end{aligned}
	\end{equation}
\end{lemma}

Theorem \ref{thm:nuclear} follows immediatley after one combines Lemmas \ref{lem:large_lambda}, \ref{lem:gradient} and \ref{lem:uniform_law}. As for the rank-constrained MLE $\widehat\bP^r$, let $\widehat {\bDelta}(r) := \widehat \bP^r - \bP$. Note that the rank constraint in \eqref{prob:nonconvex-lowrank} implies that $\rank(\widehat \bDelta(r)) \le 2r$. Thus, $\nnorm{\widehat\bDelta(r)} \le (2r)^{1 / 2} \fnorm{\widehat\bDelta(r)}$ and Lemma \ref{lem:uniform_law} remains applicable in the statistical anlaysis of $\widehat \bP^r$.

\section{Computing Markov models using low-rank optimization}

In this section we develop efficient optimization methods to compute the proposed estimators for the low-rank Markov model. From now on, we drop the constraint that $\alpha / p \le Q_{ij}\le \beta / p$, which is used only to derive the statistical guarantees. In other words, $\alpha$ and $\beta$ are motivated by statistical theory, and do not need to be taken into account in the optimization. 

\subsection{Optimization methods for the nuclear-norm regularized likelihood problem}
\label{sec:optNuc}

We first consider the nuclear-norm regularized likelihood problem \eqref{prob:convex-nuclear}. It is a special case of the following linearly constrained optimization problem:
\begin{equation}\label{prob:gen-convex-nuc}
\min \, \left\{ g(\X) + c \norm{\X}_{*}\,\mid\, \cA(\bX) = b \right\}, 
\end{equation}
where $g:\Re^{p \times p} \to (-\infty,+\infty]$ is a closed, convex, but possibly non-smooth function, $\cA:\Re^{p\times p} \to \Re^m$ is a linear map, $b\in\Re^m$ and $c > 0$ are given data. If we take $\alpha=0, \beta=p$ in  problem \eqref{prob:convex-nuclear}, it becomes a special case of the general problem \eqref{prob:gen-convex-nuc} with $g(\X) = -\ell_n(\X) + \delta(\X\ge 0)$, $\cA(\X) = \X {\bf 1}_p$, $b = {\bf 1}_p$, and $\delta(\cdot)$ being the indicator function. %\mw{[Note that i changed the notation of the indicator function]}

Despite of its convexity, problem \eqref{prob:gen-convex-nuc} is highly nontrivial due to the nonsmoothness of $g$ and the presence of the nuclear norm regularizer. Here, we propose to solve it via the dual approach.  
% It is important and non-trivial to solve the convex subproblem \eqref{subprob:mmalg} in Algorithm \ref{alg:dc}. 
% Note that \eqref{subprob:mmalg} is a nuclear norm penalized convex optimization problem, we propose to apply an efficient multi-block alternating direction method of multipliers (ADMM) for solving the dual of \eqref{subprob:mmalg}. A comprehensive numerical study has been conducted in \cite{li2016schur} and justifies our procedure. %we propose to solve these subproblems via the dual approach.
% in order to solve the non-convex DC problem \eqref{prob:nonconvex-pen-DC}, we propose to solve a sequence of convex subproblems \eqref{subprob:mmalg}. Below we discuss the algorithm for solving the subproblems in Algorithm 1.
The dual of problem \eqref{prob:gen-convex-nuc} is
%Instead of working directly on \eqref{subprob:mmalg}, we study a slightly general model as follows. Given $\W\in\Re^{p\times p}$ and $\alpha \ge 0$, consider
%\begin{equation*} 
%\begin{array}{rll}
%({\bf P}) \quad \min  & f(\X) + \inprod{\W}{\X} + c \norm{\X}_* + \frac{\alpha}{2}\norm{\X}^2_F \\
%\mbox{s.t.} & \cA(\X) = b.
%\end{array}
%\end{equation*}
%Its dual problem can be written as
%\mw{[Let's not use (D) as a label. Just use a regular equation label to denote it.]}
\begin{equation} 
\label{prob:D}
\begin{array}{rll}
\min  & g^*(-{\bf\Xi}) - \inprod{b}{y}  \\
\mbox{s.t.} & {\bf\Xi} + \cA^*(y) + \S = 0,\quad \norm{\S}_2 \le c,
\end{array}
\end{equation}
where  $\norm{\cdot}_2$ denotes the spectral norm, and $g^*$ is the conjugate function of $g$ given by
\begin{align*} 
g^*({\bf\Xi}) {}
%&= \sup \left\{ \inprod{X}{\Xi} - f(X) \, \mid \, X\in\Re^{p\times p}\right\} \\
%&= \sum_{(i,j)\in\Omega} \sup \left\{ X_{ij}\Xi_{ij} + \frac{n_{ij}}{n}\log X_{ij}\right\} \\
%& + \sum_{(i,j)\in\overline{\Omega}} \sup\{\Xi_{ij}X_{ij}\,\mid\, X_{ij}\ge 0\}\\
= \sum_{(i,j)\in\Omega} \frac{n_{ij}}{n}(\log \frac{n_{ij}}{n} - 1 - \log(-\Xi_{ij})) + \delta( \Xi \le 0) \quad \forall\,{\bf \Xi}\,\in\Re^{p\times p}
\end{align*}
with $\Omega = \{(i,j) \mid n_{ij} \neq 0\}$ and $\overline \Omega = \{(i,j) \mid n_{ij} = 0\}$. 
Given $\sigma >0$, the augmented Lagrangian function $\cL_{\sigma}$ associated with \eqref{prob:D} is
\begin{align*}
\cL_{\sigma}({\bf\Xi}, y, \S; \X) = g^*(-{\bf\Xi}) - \inprod{b}{y} + \frac{\sigma}{2} \norm{{\bf\Xi} + \cA^*(y) + \S + \X/\sigma}^2 - \frac{1}{2\sigma}\norm{\X}^2.
\end{align*}
We consider popular ADMM type methods for solving problem \eqref{prob:D} (A comprehensive numerical study has been conducted in \citep{li2016schur} and justifies our procedure).
Since there are three separable blocks in \eqref{prob:D} (namely ${\bf\Xi}$, $y$, and $\S$), the direct extended ADMM is not applicable. Indeed, it has been shown in \citep{chen2016direct} that the direct extended ADMM for multi-block convex minimization problem is not necessarily convergent. Fortunately, the functions corresponding to block $y$ in the objective of \eqref{prob:D} is linear. Thus we can apply the multi-block symmetric Gauss-Sediel based ADMM (sGS-ADMM) \citep{li2016schur}.
In literature \citep{chen2017efficient,ferreira2017semidefinite,lam2017fast,li2016schur,wang2018another}, extensive numerical experiments demonstrate that sGS-ADMM is not only convergent but also faster than the directly extended multi-block ADMM and its many other variants. Specifically, the algorithmic framework of sGS-ADMM for solving \eqref{prob:D} is presented in Algorithm \ref{alg:sGS-ADMM}. 
%Note that when $\alpha = 0$, the computation steps corresponding to block $\Z$ will not be performed.

\begin{algorithm} 
	\caption{An sGS-ADMM for solving \eqref{prob:D}}
	\label{alg:sGS-ADMM}
	\begin{algorithmic}
		\STATE {\bfseries Input:} initial point $({\bf\Xi}^0, y^0, \S^0, \X^0)$, penalty parameter $\sigma > 0$, maximum iteration number $K$, and the step-length $\gamma \in (0,(1+ \sqrt{5})/2)$
		\FOR{$k=0$ {\bfseries to} $K$}
		\STATE $y^{k+\frac{1}{2}} = \argmin_{y} \cL_{\sigma}({\bf\Xi}^k, y, \S^k; \X^k)$
		\STATE ${\bf\Xi}^{k+1} = \argmin_{{\bf\Xi}} \cL_{\sigma}({\bf\Xi} ,y^{k+\frac{1}{2}}, \S^k; \X^k) $
		\STATE $y^{k+1} = \argmin_{y} \cL_{\sigma}({\bf\Xi}^{k+1}, y, \S^k; \X^k)$
		\STATE $\S^{k+1} = \argmin_{\S} \cL_{\sigma}({\bf\Xi}^{k+1}, y^{k+1}, \S, ; \X^k)$
		\STATE
		$\X^{k+1} = \X^k + \gamma\sigma({\bf\Xi}^{k+1} + \cA^*(y^{k+1}) + \S^{k+1})$
		\ENDFOR
		%\UNTIL{$noChange$ is $true$}
	\end{algorithmic}
\end{algorithm}

%\mw{[Make the following an itemized list to discuss the steps of Algorithm 1 one-by-one ]}
Next, we discuss how the $k$-th iteration of Algorithm \ref{alg:sGS-ADMM} is performed:
%In the $k$-th iteration of Algorithm \ref{alg:sGS-ADMM}, 
\begin{description}
	\item[{\bf Computation of $y^{k+\frac{1}{2}}$ and $y^{k+1}$}.]
	Simple calculations show that $y^{k+\frac{1}{2}}$ and $y^{k+1}$ can be obtained by solving the following linear systems:
	\begin{equation*}
	\left\{ 
	\begin{aligned}
	&y^{k+\frac{1}{2}} = \frac{1}{\sigma}(\cA \cA^*)^{-1}\big( b - X^k - \sigma({\bf\Xi}^k + \S^k)\big ), \\[5pt]
	&y^{k+1} =  \frac{1}{\sigma}(\cA \cA^*)^{-1}\big( b - X^k - \sigma({\bf\Xi}^{k+1} + \S^k)\big ).
	\end{aligned}
	\right. 
	\end{equation*}
	In our estimation problem, it is not difficult to verify that 
	$ \cA \cA^* y = p y$ for any $y\in \Re^p$.
	Thanks to this special structure, the above formulas can be further reduced to 
	\begin{equation*}
	y^{k+\frac{1}{2}} = \frac{1}{\sigma p}\big( b - \X^k - \sigma({\bf\Xi}^k + \S^k)\big) \mbox{ and }
	y^{k+1} =  \frac{1}{\sigma p}\big( b - \X^k - \sigma({\bf\Xi}^{k+1} + \S^k)\big).
	\end{equation*}
	\item [{\bf Computation of ${\bf \Xi}^{k+1}$}.] 
	To compute ${\bf \Xi}^{k+1}$, we need to solve the following optimization problem:
	%\mw{MW: what about other steps in the algorithm? do they all have closed-form solutions? if yes, please give the formula. If not, explain how this step is carried out}
	\[
	\min_{\bf\Xi} \left\{ g^*(-{\bf\Xi}) +\frac{\sigma}{2}\norm{{\bf\Xi} + {\bf R}^k}^2 \right\},
	\]
	where ${\bf R}^k \in \Re^{p\times p}$ is given. Careful calculations, together with the Moreau identity \citep[Theorem 31.5]{rockafellar2015convex}, show that
	\[
	{\bf \Xi}^{k+1} = \frac{1}{\sigma} [\Z^k - \sigma {\bf R}^k]\, \mbox{ and } \,
	\Z^k = \argmin_{\Z} \left\{\sigma g(\Z) + \frac{1}{2}\norm{\Z - \sigma {\bf R}^k}^2 \right\}.
	\]
	For our estimation problem, i.e., $g(\X) = \ell_n(\X) + \delta( \X\ge 0)$, it is easy to see that 
	$Z^k$ admits the following form:
	\[
	Z^k_{ij} = \frac{\sigma R_{ij}^k + \sigma\sqrt{(R^k_{ij})^2 + 4 n_{ij}/(n\sigma)}}{2} \quad \mbox{if } (i,j)\in \Omega \, \mbox{ and } \,
	Z^k_{ij} = \sigma \max(R_{ij}^k,0) \quad \mbox{if }  (i,j)\in \overline{\Omega}.
	\]
	\item [{\bf Computation of $\S^{k+1}$}.]
	The computation of $\S^{k+1}$ can be simplified as:
	\[
	\S^{k+1} = \argmin_{\S} \left\{
	\frac{\sigma}{2} \norm{\S + {\bf\Xi}^{k+1} + \cA^* y^{k+1} + \X^k/\sigma}^2 \mid 
	\norm{\S}_2 \le c  
	\right\}.
	\]
	Let ${\bf W}_{k} := -({\bf\Xi}^{k+1} + \cA^* y^{k+1} + X^k/\sigma)$ admit the following singular value decomposition (SVD)
	$
	{\bf W}_{k} = {\bf U}_{k} {\bf\Sigma}_k {\bf V}_k^\top,
	$
	where ${\bf U}_{k}$ and ${\bf V}_k$ are orthogonal matrices,
	${\bf \Sigma}_k = {\rm Diag}(\alpha_1^k,\ldots, \alpha_p^k)$ is
	the diagonal matrix of singular values of $W_k$, with $\alpha_1^k\ge\ldots\ge \alpha_p^k\ge 0.$ Then, by Lemma 2.1 in \citep{jiang2014partial}, we know that
	\[
	\S^{k+1} = {\bf U}_k\min({\bf \Sigma}_k,c){\bf V}_k^\top,
	\]
	where  $\min({\bf \Sigma}_k,c) = {\rm Diag}\big(\min(\alpha_1^k,c),\ldots, \min(\alpha_p^k,c)\big)$.
	%We also note that when implementing Algorithm \ref{alg:sGS-ADMM}, only partial SVD, which is much cheaper than full SVD, is needed as $r \ll p$. 
	We also note that in the implementation, only partial SVD, which is much cheaper than full SVD, is needed as $r\ll p$.
\end{description}
The nontrivial convergence results and the sublinear non-ergodic iteration
complexity of Algorithm \ref{alg:sGS-ADMM} can be  obtained from \cite{li2016schur} and \cite{chen2017efficient}. %Thus, we state the theorem here without the proof.
{We put the convergence theorem and a sketch of the proof in the supplementary material.}
%\begin{theorem}
%	\label{thm:sGS-ADMM}
%	Suppose that the solution sets of {\rm({\bf P})} and {\rm({\bf D})} are nonempty. Let $\{({\bf\Xi}^k,y^k,\S^k,\X^k)\}$ be the sequence generated by Algorithm \ref{alg:sGS-ADMM}. If $\tau\in(0,(1+\sqrt{5}\,)/2)$, then the sequence $\{({\bf\Xi}^k,y^k,\S^k)\}$ converges to an optimal solution of {\rm ({\bf D})} and $\{\X^k\}$ converges to an optimal solution of {\rm ({\bf P})}.
%\end{theorem}
%\mw{MW: what is new in this result? if it is a restatement of the existing result, we might not need it in the main text}

\subsection{Optimization methods for the rank-constrained likelihood problem}
\label{sec:6}

%\begin{figure*}[tb]
%	\centering
%	\subfigure[$t = 00:00 \sim 05:59$.]{
%		\label{figure:lrh0}
%		\includegraphics[width=0.23\columnwidth]{figs//icml-lowrank-r4hh0}
%	}
%	%\hspace{0.5cm}
%	\subfigure[$t = 06:00 \sim 11:59$.]{
%		\label{figure:lrh1}
%		\includegraphics[width=0.23\columnwidth]{figs//icml-lowrank-r4hh1}
%	} 
%	\subfigure[$t = 12:00 \sim 17:59$.]{
%		\label{figure:lrh2}
%		\includegraphics[width=0.23\columnwidth]{figs//icml-lowrank-r4hh2}
%	}
%	%\hspace{0.5cm}
%	\subfigure[$t = 18:00 \sim 23:59$.]{
%		\label{figure:lrh3}
%		\includegraphics[width=0.23\columnwidth]{figs//icml-lowrank-r4hh3}
%	}
%	\caption{The meta-states compression of Manhattan traffic network via rank-constrained approach with $ r = 4$. Each color or symbol represents a meta-state.}
%	\label{figure:lr4}
%\end{figure*}
%\begin{figure*}[tb]
%	%\centering
%	\subfigure[$t = 00:00 \sim 05:59$.]{
%		\label{figure:lr10h0}
%		\includegraphics[width=0.23\columnwidth]{figs//icml-lowrank-r10hh0}
%	}
%	%\hspace{0.5cm}
%	\subfigure[$t = 06:00 \sim 11:59$.]{
%		\label{figure:lr10h1}
%		\includegraphics[width=0.23\columnwidth]{figs//icml-lowrank-r10hh1}
%	}  
%	\subfigure[$t = 12:00 \sim 17:59$.]{
%		\label{figure:lr10h2}
%		\includegraphics[width=0.23\columnwidth]{figs//icml-lowrank-r10hh2}
%	}
%	%\hspace{0.5cm}
%	\subfigure[$t = 18:00 \sim 23:59$.]{
%		\label{figure:lr10h3}
%		\includegraphics[width=0.23\columnwidth]{figs//icml-lowrank-r10hh3}
%	}
%	\caption{
%		The meta-states compression of Manhattan traffic network via rank-constrained approach with $r = 10$. Each color or symbol represents a meta-state.}
%	\label{figure:lr10}
%\end{figure*}
Next we develop the optimization method for computing the rank-constrained likelihood maximizer from \eqref{prob:nonconvex-lowrank}. In Subsection \ref{sec:pen}, a penalty approach is applied to transform the original intractable rank-constrained problem into a DC programming problem. Then we solve this problem by a proximal DC (PDC) algorithm  in Subsection \ref{subsec:proxdca}. 
%Then in Section \ref{subsec:proxdca}, a proximal DC algorithm (DCA) is applied to solve the resulted DC programming problem. 
We also discuss the solver for the subproblems involved in the proximal DC algorithm. Lastly, a unified convergence analysis of a class of majorized indefinite-proximal DC (Majorized iPDC) algorithms is provided in Subsection \ref{sec:DCA}.

\subsubsection{A penalty approach for problem \eqref{prob:nonconvex-lowrank}. } 
\label{sec:pen}
Recall \eqref{prob:nonconvex-lowrank} is intractable due to the non-convex rank constraint, we introduce a penalty approach to relax. We particularly study the following optimization problem:
\begin{equation}\label{prob:gen-nonconvex-lowrank}
\min \, \left\{ f(\X) \,\mid\, \cA(\X) = b,\, {\rm rank}(\X)\le r \right\}, 
\end{equation}
where $f:\Re^{p \times p} \to (-\infty,+\infty]$ is a closed proper convex, but possibly non-smooth, function. %Compared to \eqref{prob:gen-convex-nuc}, the nuclear norm regularizer is replaced by a rank constraint now. Here, we present a penalty approach to handle the rank constraint.
%
%$f:\Re^{p \times p} \to (-\infty,+\infty]$ is a closed, convex, but possibly non-smooth function, $\cA:\Re^{p\times p} \to \Re^m$ is a linear map, $b\in\Re^m$ and $r > 0$ are given data. 
%Similar to the discussions in Section \ref{sec:optNuc}, 
The original rank-constraint maximum likelihood problem \eqref{prob:nonconvex-lowrank} can be viewed as a special case of the general model \eqref{prob:gen-nonconvex-lowrank}. %Let ${\bf 1}_p \in \Re^{p}$ be a vector of all ones. 
%Especially when $f(\X) = -\frac{1}{n} \sum_{i=1}^p \sum_{j=1}^{p} n_{ij}\log(X_{ij}) + \delta(\X\mid \X\ge 0)$, $\cA(\X) = \X {\bf 1}_p$, $b = {\bf 1}_p$, and $\delta(\cdot\mid \X\ge 0)$ is the indicator function of the closed convex set $\{ \X\in \Re^{p\times p} \mid X\ge 0\}$, the general problem \eqref{prob:gen-nonconvex-lowrank} becomes  the original rank-constraint maximum likelihood problem \eqref{prob:nonconvex-lowrank}. 

%The possible non-smoothness of $f$ and the rank constraint make the  problem \eqref{prob:gen-nonconvex-lowrank} extremely difficult.
%In order to solve \eqref{prob:gen-nonconvex-lowrank}, we propose a penalty approach. This approach first appears in \cite{sun2010majorized} where the authors introduced this technique to solve rank constrained correlation matrix problems. See also \cite{gotoh2017dc}.

%\textcolor{red}{To Xudong: do not use  $\cA$ for matrix. Use $A$ or $\mathbf{A}$}
%
%\blue{Ans: There might be a misunderstanding here. $\cA$ is a linear map defined on the matrix space $\Re^{p\times p}$. It is a tensor instead of a matrix. Hence, I use $\cA$ (instead of $A$ or ${\A}$) here. }

Given $\X\in\Re^{p\times p}$, let $\sigma_1(\X) \geq \cdots \geq \sigma_p(\X)\geq 0$ be the singular values of $\X$. Since ${\rm rank}(\X)\le r$ if and only $\sigma_{r+1}(\X) + \ldots + \sigma_{p}(\X) = \|\X\|_\ast - \|\X\|_{(r)} = 0$ ($\|\X\|_{(r)} = \sum_{i=1}^r \sigma_i(\X)$ is the Ky Fan $r$-norm of $\X$), \eqref{prob:gen-nonconvex-lowrank} can be equivalently formulated as
\begin{equation*}
%\label{prob:nonconvex-lowrank}
\min \left\{ f(\X)  \mid \|\X\|_\ast-\|\X\|_{(r)}=0, \, \cA(\X) = b \right\}.
\end{equation*}
See also \citep[Equation (29)]{sun2010majorized}.
The penalized formulation of problem \eqref{prob:gen-nonconvex-lowrank} is
\begin{equation}\label{prob:nonconvex-pen-DC}
%\label{prob:nonconvex-pen-lowrank}
\min%_{\X\in\Re^{p\times p}} 
\left\{ f(\X) +  c (\|\X\|_\ast - \|\X\|_{(r)}) \mid  \cA(\X) = b \right\},
\end{equation}
where  $c>0$ is a penalty parameter. 
%Since
%\[ \sum_{i=r+1}^p \sigma_i(Q) = \norm{Q}_{*} - \norm{Q}_{(r)}\] 
%with $\norm{Q}_{(r)} = \sum_{i=1}^r \sigma_i(Q)$ denoting the Ky Fan $r$-norm, i.e., the summation of the $r$ largest  singular values.
%Therefore, \eqref{prob:nonconvex-pen-lowrank} can be further written as
%\begin{equation}
%\label{prob:nonconvex-pen-DC}
%\hspace{-0.1cm}\min_{X\in\Re^{p\times p}} \left\{ f(X) +  c \big(\norm{X}_{*} - \norm{X}_{(r)} \big) \mid  \, \cA X = b \right\}.
%\end{equation}
Since $\norm{\cdot}_{(r)}$ is convex, the objective in problem \eqref{prob:nonconvex-pen-DC} is a difference
of two convex functions: $f(\X) + c\norm{\X}_{*}$ and $c\norm{\X}_{(r)}$, i.e., \eqref{prob:nonconvex-pen-DC} is a DC program.

%Note that problem \eqref{prob:nonconvex-pen-DC} is not equivalent to the original problem \eqref{prob:gen-nonconvex-lowrank}. 
Let $\X_c^*$ be an optimal solution to the penalized problem \eqref{prob:nonconvex-pen-DC}. The following proposition shows that $\X_c^\ast$ is also the optimizer to \eqref{prob:gen-nonconvex-lowrank} when it is low-rank.
\begin{proposition}\label{prop:penlowrank}
	If ${\rm rank}(\X_c^*)\le r$, then $\X_c^*$ is also an optimal {solution} to the original problem \eqref{prob:gen-nonconvex-lowrank}.
\end{proposition}

%\begin{proof}
%	Since ${\rm rank}(X_c^*)\le r$, we know that $X_c^*$ is in fact a feasible solution to the original problem \eqref{prob:gen-nonconvex-lowrank} and $\norm{X_c^*}_{*} - \norm{X_c^*}_{(r)} = 0$. Therefore, for any feasible solution $X$ to
%	\eqref{prob:gen-nonconvex-lowrank}, it holds that 
%	\begin{align*} 
%	f(X_c^*) ={}& f(X_c^*) + c(\norm{X_c^*}_{*} - \norm{X_c^*}_{(r)})\\[5pt]
%	\le{}& f(X) + c(\norm{X}_* - \norm{X}_{(r)})
%	= f(X).
%	\end{align*}
%	This completes the proof of the proposition.
%\end{proof}

In practice, one can gradually increase the penalty parameter $c$ to obtain a sufficient low rank solution $\X_c^*$. In our numerical experiments, we can obtain solutions with the desired rank with a properly chosen parameter $c$.

\subsubsection{A PDC algorithm for the penalized problem \eqref{prob:nonconvex-pen-DC}.}
\label{subsec:proxdca}

The central idea of the DC algorithm \citep{tao1997convex} is as follows: at each iteration, one approximates the concave part of the objective function by its affine majorant, then solves the resulting convex optimization problem. In this subsection, we present a variant of the classic DC algorithm for solving \eqref{prob:nonconvex-pen-DC}. For the execution of the algorithm, we recall that the sub-gradient of Ky Fan $r$-norm at a point $\X\in\Re^{p\times p}$ \citep{watson1993matrix} is
%\begin{equation*}
\[\partial \norm{\X}_{(r)}=\left\{ 
\U \, {\rm Diag}(q^*)\V^\top \, \mid q^*\in \Delta \right\},
\]%\end{equation*}
where {$\bU$ and $\bV$ are the singular vectors of $\X$,} and $\Delta$ is the optimal solution set of the following problem
\begin{equation*}
\max_{q\in \Re^{p}} \left\{ 
\sum_{i=1}^p \sigma_i(\X)q_i \mid 
%\begin{aligned}
\inprod{{\bf 1}_p}{q} \le r, \, 0\le q \le 1
%\end{aligned}
\right\}.
\end{equation*}
%\textcolor{red}{To Xudong: max over what?}
%\begin{equation}
%\partial \norm{X}_{(r)}=\left\{ 
%\begin{aligned}
%& U{\rm Diag}(q^*)V^\top \\
%& \mid q^*\in \argmax \left\{ 
%\sum_{i=1}^p \sigma_i(X)q_i \mid 
%\begin{aligned}
%\inprod{1_q}{q} = r,\\
%0\le q \le 1.
%\end{aligned}
%\right\}
%\end{aligned}
%\right\}
%\end{equation}
Note that one can efficiently obtain a component of $\partial \norm{\X}_{(r)}$ by computing the
SVD of $X$ and picking up the SVD vectors corresponding to the $r$ largest singular values.
After these preparations, we are ready to state the PDC algorithm for problem \eqref{prob:nonconvex-pen-DC} in Algorithm \ref{alg:dc}.
Different from the classic DC algorithm, an additional proximal term is added to ensure that solutions of subproblems \eqref{subprob:mmalg} exist and the difference of two consecutive iterations converges. {See Theorem \ref{thm:convergence-alg-MM} and Remark \ref{remk:pdc} for more details.} 
%\textcolor{red}{To Xudong: why need a proximal term to ensure conv?}
%
%\blue{Ans: The proximal term is added to ensure the convergence of 
%$\{\norm{X^k - X^{k+1}}_F\}$.}
%For any given $R\in \Re^{p\times p}$, we know from the convexity of $\norm{\cdot}_{(r)}$ and the definition of the subgradient that 
%\begin{equation} \label{eq:subg_kyfanknorm}
%\norm{Q}_{(r)} \ge \norm{R}_{(r)} + \inprod{W}{Q - R}, \, \forall Q\in \Re^{p\times p}, W\in \partial \norm{R}_{(r)}. 
%\end{equation}

\begin{algorithm}
	\caption{A PDC algorithm for solving \eqref{prob:nonconvex-pen-DC}}
	\label{alg:dc}
	Given $c>0$, $\alpha \ge 0$, and the stopping tolerance $\eta$, choose initial point $\X^0\in \Re^{p\times p}$.
	Iterate the following steps for $k=0,1,\ldots:$
	
	{\bf 1.}  Choose  $\W_k\in \partial \norm{\X^k}_{(r)}$. Compute
	\begin{equation}\label{subprob:mmalg}
	\begin{split}
	\bX^{k+1}\quad = \quad & \argmin f(\X) + c\left(\norm{\X}_* - \inprod{\W_k}{\X - \X^k} - \norm{\X^k}_{(r)}\right) 
	+ \frac{\alpha}{2}\norm{\X - \X^k}_F^2  \\
	& \text{subject to } \cA(\X) = b.
	\end{split}
	\end{equation}
	{\bf 2.} If $\norm{\X^{k+1} - \X^k}_F\le \eta$, stop.
\end{algorithm}
We say that $\X$ is a critical point of problem \eqref{prob:nonconvex-pen-DC} if 
\[\partial (f(\X) + c\norm{\X}_* + \delta(\cA(X) = b) ) \cap (c\partial \norm{\X}_{(r)}) \ne \emptyset.\]
We have the following convergence results for Algorithm \ref{alg:dc}. 
%For Algorithm \ref{alg:dc}, we have the following convergence results which can be obtained directly from Theorem \ref{thm: MMconvergence}.
\begin{theorem}[Convergence of Algorithm \ref{alg:dc}]
	\label{thm:convergence-alg-MM}
	Let $\{\X^k\}$ be the sequence generated by Algorithm \ref{alg:dc} and $\alpha \ge 0$. Then $\{ f(\X^k) + c(\norm{\X^k}_* - \norm{\X^k}_{(r)})\}$ is a non-increasing sequence. If $\X^{k+1} = \X^k$ for some integer $k\ge 0$, then $\X^k$ is a critical point of \eqref{prob:nonconvex-pen-DC}. Otherwise, it holds that
	\begin{align*} 
	&\big(f(\X^{k+1}) + c(\norm{\X^{k+1}}_* - \norm{\X^{k+1}}_{(r)})\big) 
	- \big( f(\X^k) + c(\norm{\X^k}_* - \norm{\X^k}_{(r)} )\big)
	\le {} -\frac{\alpha}{2}\norm{\X^{k+1} - \X^k}^2_F.
	\end{align*}
	Moreover, 
	any accumulation point of the bounded sequence $\{\X^k\}$ is a critical point of problem \eqref{prob:nonconvex-pen-DC}.
	In addition, if $\alpha >0$, it holds that $\lim_{k\to \infty}\norm{\X^{k+1} - \X^k}_F = 0$.
\end{theorem}
%The proof of Theorem \ref{thm:convergence-alg-MM} can be obtained directly from Theorem \ref{thm: MMconvergence} where we provide a unified convergence analysis for a more general DCA. 

\begin{remark}[Adjusting Parameters]
	\label{remk:pdc}
	{\rm In practice, a small $\alpha >0$ is suggested to ensure strict decrease of the objective value and convergence 
		of $\{ \norm{\X^{k+1} - \X^k}_F \}$; %Algorithm \ref{alg:dc};%  
		if $f$ is strongly convex, one achieves these nice properties even if $\alpha =0$ based on the results of Theorem \ref{thm: MMconvergence}. The penalty parameter $c$ can be adaptively adjusted according to the rank of the sequence generated by Algorithm \ref{alg:dc}.} %This technique is extremely important in the implementation.
\end{remark}

\begin{remark}[Number of iterations of Algorithm \ref{alg:dc}]
	\label{prop:convergence-alg-MM}
	Let $\eta >0$ be the stopping tolerance and $F^*$ be the optimal value of problem \eqref{prob:nonconvex-pen-DC}. By using the inequality in Theorem \ref{thm:convergence-alg-MM}, it can be shown that if $\alpha >0$, then  Algorithm \ref{alg:dc} terminates in no more than $K$ iterations, where
	\[
	K = \left\lceil \frac{2\big(f(X^0) + c(\norm{X^0}_* - \norm{X^0}_{(r)}) - F^* \big)}{\alpha \eta^2}\right\rceil + 1.
	\] 
\end{remark}

\begin{remark}[Statistical properties]
		The statistical rate we derived in Theorem 2 does not carry over to the iterates of the DC algorithm here. Though we show in Theorem 5 that the DC algorithm can converge to a critical point, it remains unclear whethere this point is close to the global optimum and provably enjoys the statistical guarantee. Recently there have been many works conveying positive messages on the statistical properties of the non-convex optimization algorithms. For example, \citet{LWa15} showed that any stationary point of the composite objective function they considered lies within statistical precision of the true parameter. We hope to establish similar theory for the proposed DC approach in future research.
\end{remark}

Next, we discuss how to solve subproblems \eqref{subprob:mmalg}. \eqref{subprob:mmalg} is still a nuclear norm penalized convex optimization problem and is a special case of model \eqref{prob:gen-convex-nuc} with $g(\X) = f(\X) + \inprod{\W}{\X} + \frac{\alpha }{2}\norm{\X}_F^2$. Hence, Algorithm \ref{alg:sGS-ADMM} can directly solve these subproblems efficiently. When Algorithm \ref{alg:sGS-ADMM} is executed on this new function $g$, all computations, except for the update of $\bf\Xi$, have already been discussed in Section \ref{sec:optNuc}. To update $\bf\Xi$ in the process of executing Algorithm \ref{alg:sGS-ADMM} for solving \eqref{subprob:mmalg} with $g(\X) = \ell_n(\X) + \delta(\X\ge0) + \inprod{\W}{\X} + \frac{\alpha }{2}\norm{\X}_F^2$, we need to solve the following minimization problem for given ${\bf R}\in\Re^{p\times p}$ and $\sigma > 0$,
\[
\Z^* = \argmin_{\Z} \left\{ \sigma g(\Z) + \frac{1}{2} \norm{\Z - \sigma{\bf R}}^2\right\}.
\]  
$\Z^*$ here can be calculated by
\begin{equation*}
Z_{ij}^* = 
\left\{
\begin{aligned}
{}&\frac{(\sigma R_{ij} - W_{ij}) + \sigma\sqrt{(R_{ij} - W_{ij}/\sigma)^2 + 4 (\alpha + 1)n_{ij}/(n\sigma)}}{2(\alpha + 1)} \quad \mbox{if } (i,j)\in \Omega; \\[5pt]
{}&\sigma \max(R_{ij} - W_{ij}/\sigma,0) \quad  \mbox{if } (i,j)\in \overline{\Omega}.
\end{aligned}
\right.
\end{equation*}

\subsubsection {A unified analysis for the majorized iPDC algorithm.}
\label{sec:DCA}
%\begin{figure*}[tb]
%	%	\centering
%	\begin{subfigure}{.33\textwidth}
%		\label{figure:nu-etaF}
%		\includegraphics[width=\linewidth]{figs//nu-etaF.pdf}
%		%\caption{rank VS. mle}
%	\end{subfigure}
%	\begin{subfigure}{.33\textwidth}
%		\label{figure:nu-etaU}
%		\includegraphics[width=\linewidth]{figs//nu-etaU.pdf}
%		%\caption{rank VS. mle}
%	\end{subfigure}
%	\begin{subfigure}{.33\textwidth}
%		\label{figure:nu-etaV}
%		\includegraphics[width=\linewidth]{figs//nu-etaV.pdf}
%		%\caption{rank VS. mle}
%	\end{subfigure}
%	\caption{Comparison between {\bf rank} and {\bf nu} with accuracy measures $(\eta_F, \eta_U, \eta_V)$ versus the number of state jumps $n = C  p\log p$.}
%	\label{figure:nuvsrank}
%\end{figure*}
Due to the presence of the proximal term $\frac{\alpha}{2}\norm{\X - \X^k}^2$ in Algorithm \ref{alg:dc}, the classical DC analyses cannot be applied directly. In this subsection, we provide a unified convergence analysis for the majorized indefinite-proximal DC (majorized iPDC) algorithm which includes Algorithm \ref{alg:dc} as a special instance. 
Let $\mathbb{X}$ be a finite-dimensional real Euclidean space endowed with inner product $\inprod{\cdot}{\cdot}$ and induced norm $\norm{\cdot}$.
Consider the following optimization problem
\begin{equation}\label{eq:dca model}
\min_{x\in \mathbb{X}}\; \theta(x)\triangleq g(x) + p(x) - q(x),
\end{equation}
where $g:\mathbb{X}\to \Re$ is a continuously differentiable function (not necessarily convex) with a Lipschitz continuous gradient and Lipschitz modulus $L_g >0$, i.e.,
\[\norm{\nabla f(x) - \nabla f(x')} \le L_g \norm{x - x'}\quad \forall \, x, x'\in \mathbb{X}, \]
$p:\mathbb{X}\to (-\infty, +\infty]$ and $q:\mathbb{X}\to(-\infty, +\infty]$ are two proper closed convex functions.
It is not difficult to observe that penalized problem \eqref{prob:nonconvex-pen-DC} is a special instance of problem \eqref{eq:dca model}. 
For general model \eqref{eq:dca model}, one can only expect the DC algorithm converges to a critical point $\bar{x}\in \mathbb{X}$ of  \eqref{eq:dca model} satisfying 
\[ \left(\nabla g(\bar{x}) + \partial p(\bar{x}) \right) \cap \partial q(\bar{x}) \neq \emptyset.\]
%\end{equation*}
%Problem \eqref{eq:dca model} is a general model and covers many problems arising from machine learning, signal processing, and statistics. For example, in many data fitting problems, $g$ is typically chosen as a loss function measuring the deviation of a
%solution from the observations and $p - q$ is a regularizer intended to induce desirable structures in the
%solution. 
%When $g$ is a convex function, $\theta$ is a {\em DC} function and problem 
%\eqref{eq:dca model} is a DC programming problem. 

%The classic algorithm for DC programming problems is the DC algorithm (DCA), which was introduced by Pham Dinh and Le Thi \cite{tao1997convex}. 
%See \cite{tao2005dc,le2012exact,le2017stochastic,lethi2018} for more details and recent developments.
%
%In this subsection, based on the classic DCA, we propose a generalized DCA for solving problem \eqref{eq:dca model}. As one shall see later, our algorithm enjoys more flexibility than the classic DCA.  Meanwhile, a unified convergence analysis is provided for the proposed algorithm. We emphasize here that our convergence results can also recover those obtained for the classic DCA. %when $f\equiv 0$. 

%\textcolor{red}{Xudong: The next two paragraphs can be shortened/removed. We just need the main results here}

Since $g$ is continuously differentiable with Lipschitz continuous gradient, there exists a self-adjoint positive semidefinite linear operator $\cG:\mathbb{X} \to \mathbb{X}$ such that for any
$x, x'\in \mathbb{X}$,
\begin{equation*}\label{ineq:majorization}
g(x)\le \widehat{g}(x;x')\triangleq g(x') + \langle \nabla g(x'), x-x'\rangle  + \frac{1}{2}\|x - x'\|^2_{\mathcal{G}}.
\end{equation*}
%Indeed, the simplest choice of $\cG$ can be $L_g \cI$ where $\cI$ is the identity operator on $\mathbb{X}$.
%Clearly, $\widehat g$ is a majorization of $g$.
%We propose a generalized DCA (Algorithm \ref{alg:dca-general}) for solving \eqref{eq:dca model}.

\begin{algorithm}
	\caption{A majorized indefinite-proximal DC algorithm for solving problem \eqref{eq:dca model}}
	\label{alg:dca-general}
	Given initial point $x^0\in \mathbb{X}$ and stopping tolerance $\eta$, choose a self-adjoint, possibly
	indefinite, linear operator $\cT:\mathbb{X} \to \mathbb{X}$. %and $\mathcal{G}$ satisfying \eqref{ineq:majorization}.
	Iterate the following steps for $k=0,1,\ldots:$
	
	{\bf 1.} Choose $\xi^k \in \partial q({x}^k)$. Compute
	\begin{equation}\label{eq:subproblem}
	x^{k+1} \in \argmin_{x\in \mathbb{X}} \; \left\{ \widehat{\theta}(x;{x}^k) + \frac{1}{2}\|x - {x}^k\|^2_{\mathcal{T}}\right\}, 
	\end{equation}
	where $\widehat{\theta}(x;{x}^k) \triangleq \widehat{g}(x;{x}^k) + p(x) - \big(q({x}^k) + \langle x-{x}^k, \xi^k\rangle \big).$
	
	{\bf 2.}  If $\|{x}^{k+1}-{x}^k\|\le \eta$, stop.
\end{algorithm}
We present the majorized iPDC algorithm for solving \eqref{eq:dca model} in Algorithm \ref{alg:dca-general} and provide the following convergence results. 
%We shall emphasize here that our convergence results can also recover those obtained for the classic DCA.
\begin{theorem}[Convergence of iPDC]\label{thm: MMconvergence}
	Assume that {$\inf_{x\in \mathbb{X}}\theta(x)>-\infty$}. Let $\{x^k\}$ be the sequence generated by Algorithm \ref{alg:dca-general}. 
	If $x^{k+1} = x^k$ for some $k\ge 0$, then $x^k$ is a critical point of \eqref{eq:dca model}. If $\mathcal{G} + 2\mathcal{T}\succeq 0$, then
	any accumulation point of $\{x^{k}\}$, if exists, is a critical point of \eqref{eq:dca model}. In addition, if $ \mathcal{G} + 2\mathcal{T}\succ 0$, it holds that $\displaystyle\lim_{k\to\infty}\|{x}^{k+1} - {x}^k\| = 0$.
\end{theorem}
%The proof of Theorem \ref{thm: MMconvergence} and more discussions are provided in the supplementary material.
The proof of Theorem \ref{thm: MMconvergence} is provided in the supplementary material.

\begin{remark}
Here, we discuss the roles of linear operators $\cG$ and $\cT$. %The majorization technique of handling the smooth function $g$ and 
First, $\cG$ makes the subproblems \eqref{eq:subproblem} in Algorithm \ref{alg:dca-general} more amenable to efficient computations. Theorem \ref{thm: MMconvergence} shows the algorithm is convergent if $\cG + 2\cT \succeq 0$. This indicates that instead of adding the commonly used positive semidefinte or positive definite proximal terms, we allow $\cT$ to be indefinite for better practical performance. The computational benefit of using indefinite proximal terms has also been observed in \citep{sun2010majorized,li2016majorized}. %In fact, the introduction of indefinite proximal terms in the DC algorithm is motivated by this numerical evidence. 
As far as we know, Theorem \ref{thm: MMconvergence} provides the first rigorous convergence proof of the DC algorithms with indefinite proximal terms. Second, $\cG$ and $\cT$ also help to guarantee that the solutions of the subproblems \eqref{eq:subproblem} exist.
Since $\cG + 2\cT \succeq 0$ and $\cG \succeq 0$, we have that $2\cG + 2\cT \succeq 0$, i.e., $\cG + \cT \succeq 0$. %(the reverse direction holds when $\cT \succeq 0$).
Hence, $\cG+2\cT \succeq 0$ ($\cG+2\cT \succ 0$) implies that subproblems \eqref{eq:subproblem} 
are (strongly) convex. Third, the choices of $\cG$ and $\cT$ are very much problem dependent. The general principle is that $\cG + \cT$ should be as small as possible while  ensuring {$x^{k+1}$} is relatively easy to compute.
%Meanwhile, the choices of $\cG$ and $\cT$ are very much problem dependent. The general principle is that $\cG + \cT$ should be as small as possible while  {$x^{k+1}$} is still relatively easy to compute.
\end{remark}

\section{Simulation results}
\label{sec:num}

In this section, we conduct numerical experiments to validate our theoretical results. We first compare the proposed nuclear-norm regularized estimator and the rank-constrained estimator with previous methods in literature using synthetic data. We then use the rank-constrained method to analyze a dataset of Manhattan taxi trips to reveal citywide traffic patterns. All of our computational results are obtained by running {\sc Matlab} (version 9.5) on a windows workstation (8-core, Intel Xeon W-2145 at 3.70GHz, 64 G RAM).

\subsection{Experiments with simulated data}
We randomly draw the transition matrix $\P$ as follows. Let $\U_0, \V_0 \in \Re^{p\times r}$ be random matrices with i.i.d. standard normal entries and let
\[
\widetilde \U_{[i,:]} = (\U_0 \circ \U_0)_{[i,:]} / \norm{(\U_0)_{[i,:]}}_2^2 \mbox{ and } \widetilde \V_{[:,j]} = (\V_0 \circ \bV_0)_{[:,j]} / \norm{(\V_0)_{[:,j]}}_2^2, \quad i=1,\ldots,p, j=1,\ldots, r, 
\]
where $\circ$ is the Hadamard product and $\widetilde \U_{[i,:]}$ denotes the $i$-th row of $\widetilde\U$. The transition matrix $\P$ is obtained via $\P = \widetilde\U \widetilde\V^\top$. Then we simulate a Markov chain trajectory of length $n = {\rm round} (krp \log(p))$ on $p$ states, $\{X_0,\ldots, X_n\}$, with varying values of $k$. %Let $\P$ admit the following singular value decomposition (SVD), $\P = \U {\bf \Sigma} \V^\top$, where $\U, \V \in \Re^{p\times r}$ have orthonormal columns and the diagonal matrix ${\bf\Sigma}= {\rm diag}(\sigma_i) \in \Re^{r\times r}$ with $\sigma_i > 0$ for $i=1,\ldots, r$.

We compare the performance of four procedures: the nuclear norm penalized MLE, rank-constrained MLE, empirical estimator and spectral estimator. Here, the empirical estimator is the empirical count distribution matrix defined as follows: 
$$\tilde{\P} = \left(\tilde{\P}_{ij}\right)_{1\leq i, j\leq p}, \quad \tilde{\P}_{ij} = \left\{\begin{array}{ll}
\frac{\sum_{k =1}^n 1_{\{X_{k-1} = i, X_k = j\}}}{\sum_{k =1}^n 1_{\{X_{k-1} = i\}}}, & \quad  \text{when }\sum_{k =1}^n 1_{\{X_{k-1} = i\}} \geq 1;\\
\frac{1}{p}, & \quad \text{when } \sum_{k =1}^n 1_{\{X_{k-1} = i\}} = 0.
\end{array}\right.$$
The empirical estimator is in fact the unconstrained maximum likelihood estimator without taking into account the low-rank structure. The spectral estimator \citep[Algorithm 1]{zhang2018optimal} is based on a truncated SVD. In the implementation of the nuclear norm penalized estimator, the regularization parameter $\lambda$ in  \eqref{prob:convex-nuclear} is set to be $C\sqrt{{p\log p/n}}$ with constant $C$ selected by cross-validation. %is a fixed constant that is independent of the scale of the simulation.
%Specifically, given the ground truth $P$ and the simulation length $n$, we use Algorithm \ref{alg:sGS-ADMM} to compute the nuclear norm penalized estimators with twenty different values of $C$ chosen between $0.1$ and $2$. The experiment is repeat for $10$ times and the averaged recovery error is recorded for each choice. Then, the constant $C$ with the smallest associated averaged recovery error is chosen. \az{AZ: This sentence is not very clear. How do you select tuning parameter $\lambda$ in the convex method? Cross-validation?} \xd{Is the procedure I discussed cross-validation?} 
For each method, let $\widehat{\U}$ and $\widehat{\V}$ be the leading $r$ left and right singular vectors of the resulting estimator $\widehat \P$. We measure the statistical performance of $\widehat{\bP}$ through three quantities: 
\begin{align*}
 \eta_F := \norm{\P - \widehat \P}_F^2, ~~ \eta_{KL}:= D_{\KL}(\bP,\widehat\bP), \mbox{ and } \eta_{UV} := \max\bigl\{ \norm{\sin \Theta(\widehat{\U},\U)}_F^2, \norm{\sin \Theta(\widehat{\V},\V)}_F^2 \bigr\}.
\end{align*}

We consider the following setting with $p=1000$, $r = 10$, and $k\in [10, 100]$.
%; (b) $p=1000$, $k = 100$, $n = {\rm round}(10kp\log(p))$ and the rank $r \in [1, 10]$; (c) $r = 10$, $k=100$, $n = {\rm round}(kr10^3\log(10^3))$, and the dimension $p\in[500, 1500]$. 
The results are plotted in Figure \ref{figure:mlevsrank}.
%, \ref{figure:rmlevsrank}, and \ref{figure:pmlevsrank}, respectively.
One can observe from these results that for rank-constrained, nuclear norm penalized and spectral methods, $\eta_F, \eta_{KL}$ and $\eta_{UV}$ converge to zero quickly as the number of the state transitions $n$ increases, while the statistical error of the empirical estimator decreases in a much slower rate. Among the three estimators in the zoomed plots (second rows of Figure \ref{figure:mlevsrank}), the rank constrained estimator slightly outperforms the nuclear norm penalized estimator and the spectral estimator. %In particular, the estimation error of {\bf rank} is smaller than that of {\bf nu} when $n$ is large. 
This observation is consistent with our algorithmic design: the nuclear norm minimization procedure is actually the initial step of Algorithm \ref{alg:dc}; thus the rank-constrained estimator can be seen as a refined version of the nuclear norm regularized estimator.

We also consider the case where the invariant distribution $\pi$ is ``imbalanced'', i.e., we construct $\P$ such that $\min_{i=1,\ldots,p} \pi_i$ is quite small and the appearance of some states is significantly less than the others. Specifically, given $\gamma_1,\gamma_2 >0$, we generate a diagonal matrix $\D$ with i.i.d. beta-distributed (${\rm Beta}(\gamma_1, \gamma_2)$) diagonal elements.
After obtaining $\widetilde \U$ and $\widetilde \V$ in the same way as in the beginning of this subsection, we compute $\widetilde \P = \widetilde \U \widetilde \V^\top \D$. The ground truth transition matrix  $\P$ is obtained after a normalization of $\widetilde \P$. Then, we simulate a Markov chain trajectory of length $n = {\rm round}(krp\log(p))$ on $p$ states.
In our experiment, we set $p = 1000$, $r = 10$, $k\in [10,100]$, and $\gamma_1 = \gamma_2 = 0.5$. The detailed results are plotted in Figure \ref{figure:betamlevsrank}. As can be seen from the figure, under the imbalanced setting, the rank-constrained, nuclear norm penalized and spectral methods perform much better than the empirical approach in terms of all the three statistical performance measures ($\eta_F$, $\eta_{KL}$ and $\eta_{UV}$). In addition, the rank-constrained estimator exhibits a clear advantage over two other approaches.
%As one can observe from Figures , for {\bf svd}, {\bf nu} and {\bf rank}, the above three empirical error quantities are decreasing as . Besides, Figure \ref{figure:mlevsrank} and \ref{figure:svdvsrank} show that \texttt{rank} significantly outperforms \texttt{mle} and \texttt{nu} respectively.
%Combining the results in Theorems \ref{th:upper-bound} and \ref{thm:lower_bound}, we have shown that our rank-constrained estimator is much better than the plain {\bf mle} both theoretically and numerically.
%From Figure \ref{figure:svdvsrank}, one can see that our rank-constrained estimator outperforms {\bf svd} in terms of all three accuracy measurements. Recently, \cite{zhang2018optimal} studied the {\bf svd} approach and developed the total variation error bounds. The rank-constrained estimator and the KL-divergence error bounds studied in this paper are and harder to analyze and more meaningful for discrete distribution estimation. 
%via simulations. 
%Specifically, it is chosen among 10 different choices of the order $\sqrt{{p\log(p)/n}}$. We perform $100$ times simulations and choose the value of the parameter that has the smallest average error $\eta_F$. 

\begin{figure*}
	%	\centering
	\begin{subfigure}{.33\textwidth}
		\label{figure:mle-etaF}
		\includegraphics[width=\linewidth]{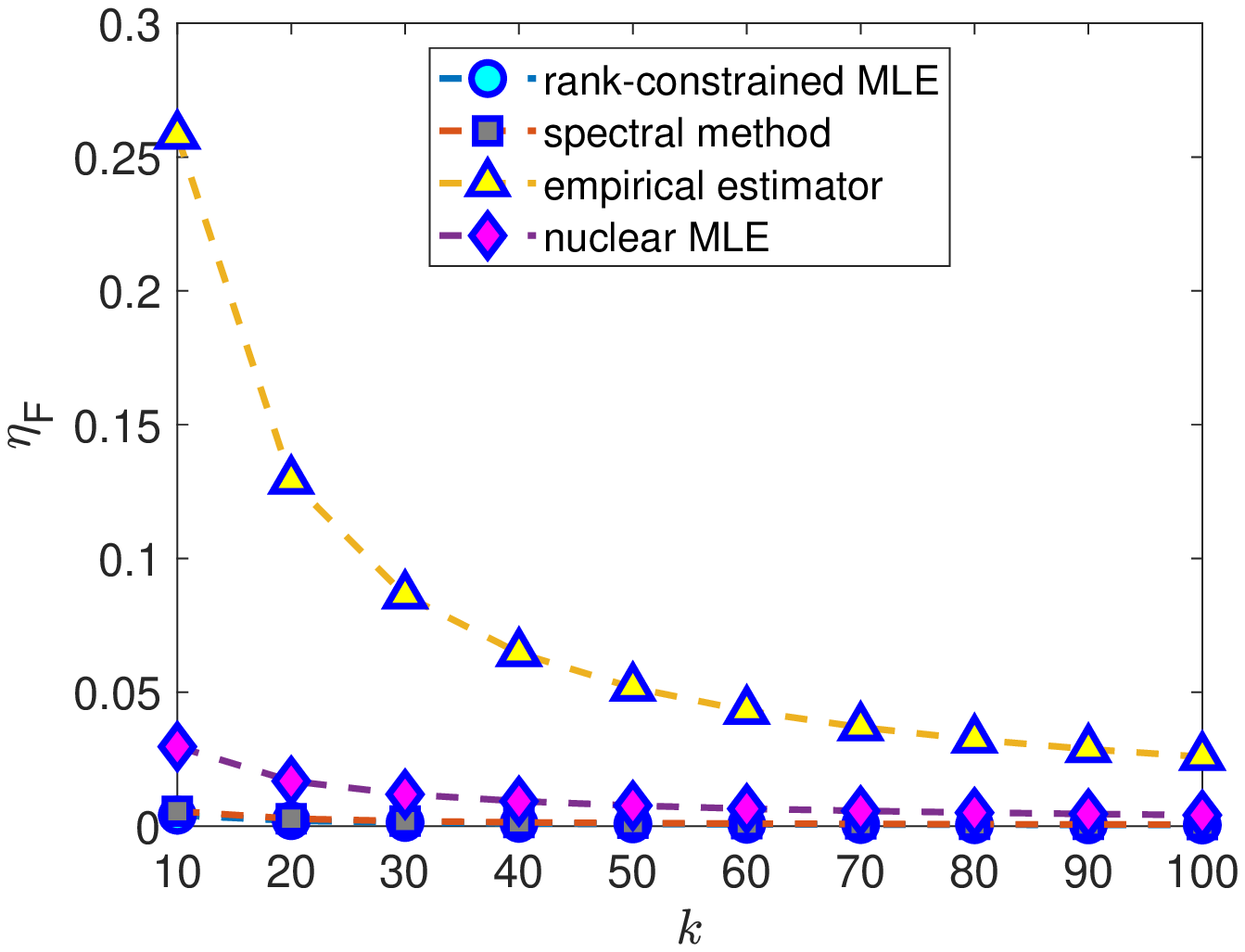}
		%\caption{rank VS. mle}
	\end{subfigure}
	\begin{subfigure}{.33\textwidth}
		\label{figure:mle-etaKL}
		\includegraphics[width=\linewidth]{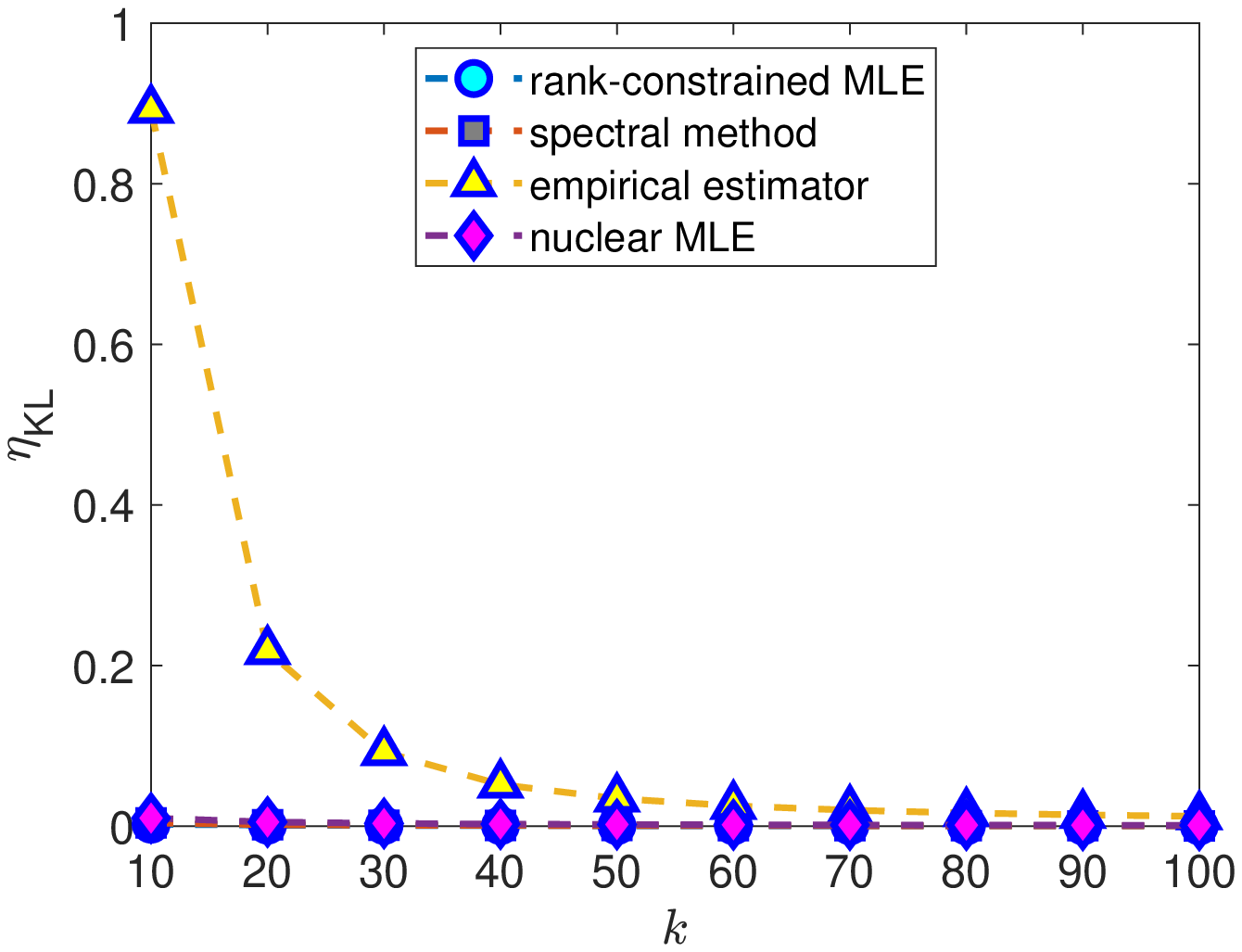}
		%\caption{rank VS. mle}
	\end{subfigure}
	\begin{subfigure}{.33\textwidth}
		\label{figure:mle-etaUV}
		\includegraphics[width=\linewidth]{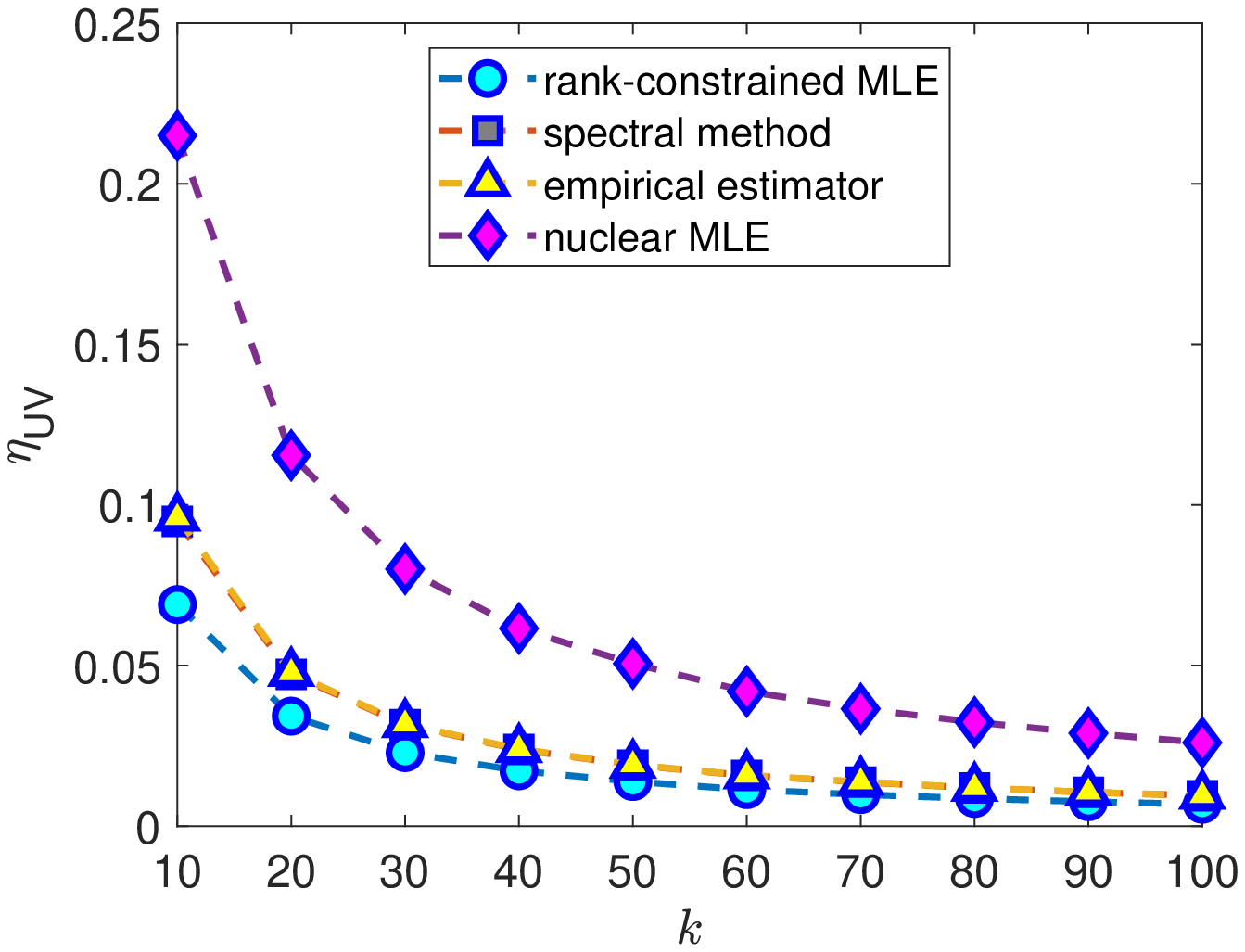}
		%\caption{rank VS. mle}
	\end{subfigure}
	%\caption{rank VS. mle}
	\\
	
	\begin{subfigure}{.33\textwidth}
		\label{figure:svd-etaF}
		\includegraphics[width=\linewidth]{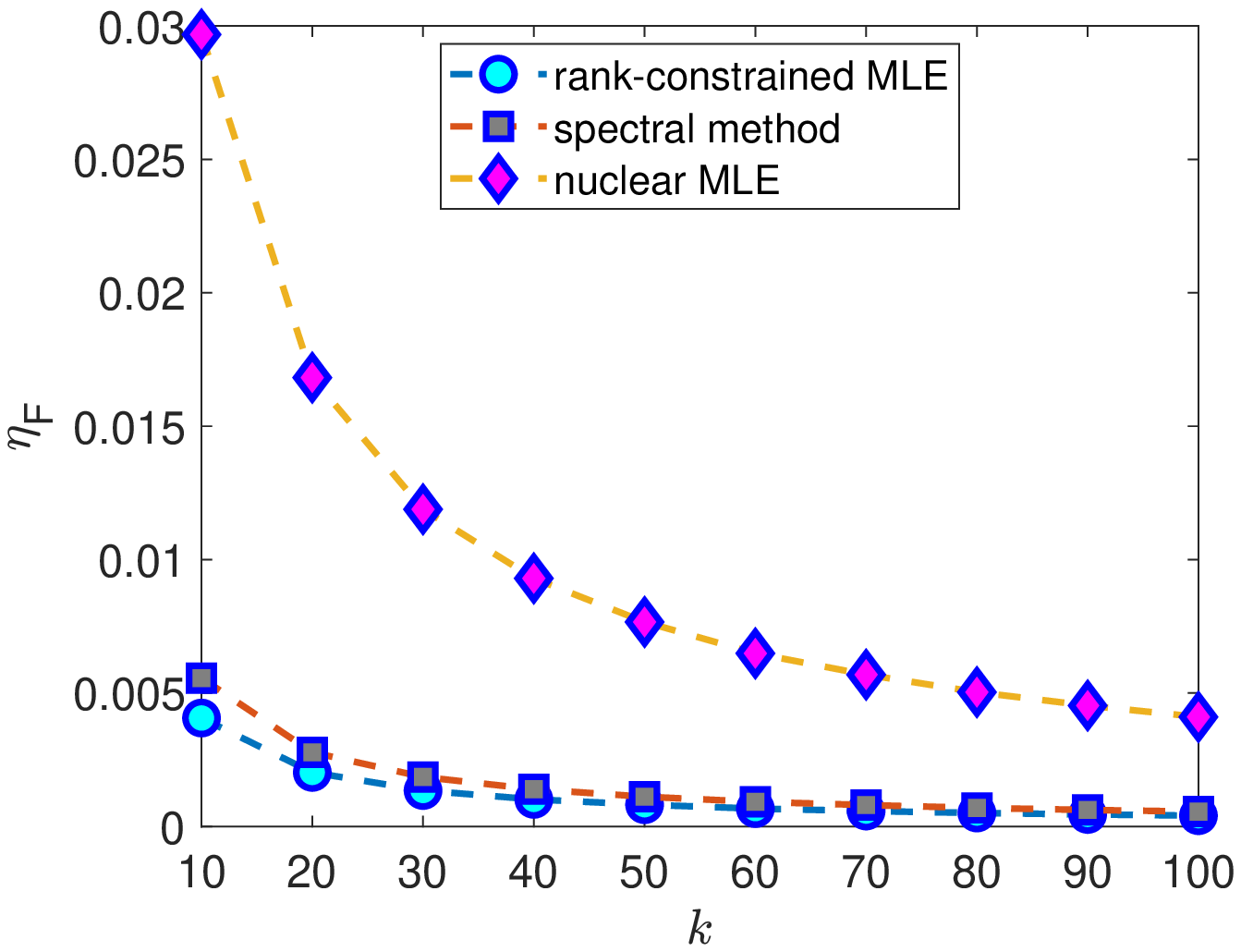}
		%\caption{rank VS. mle}
	\end{subfigure}
	\begin{subfigure}{.33\textwidth}
		\label{figure:svd-etaKL}
		\includegraphics[width=\linewidth]{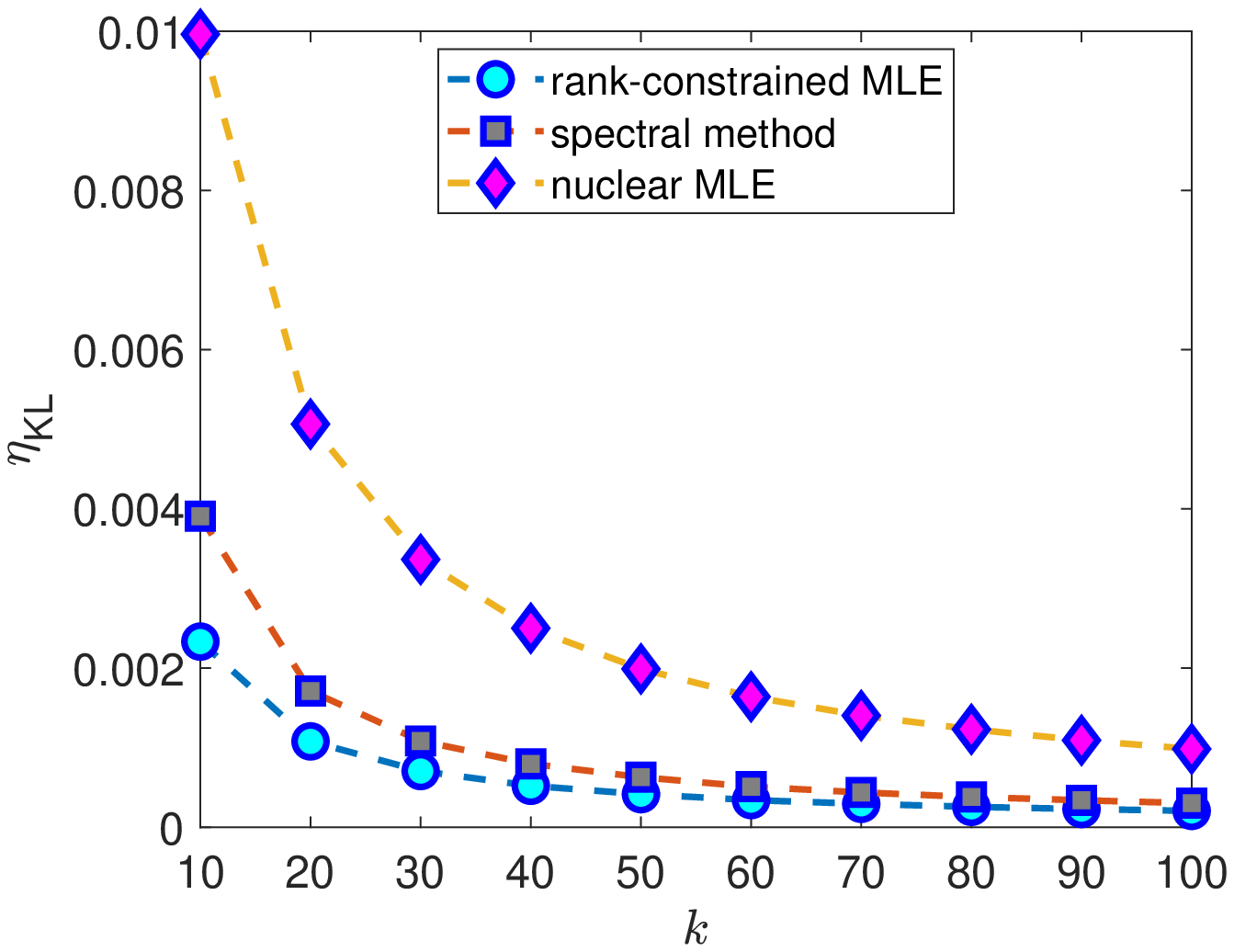}
	\end{subfigure}
	\begin{subfigure}{.33\textwidth}
		\label{figure:svd-etaUV}
		\includegraphics[width=\linewidth]{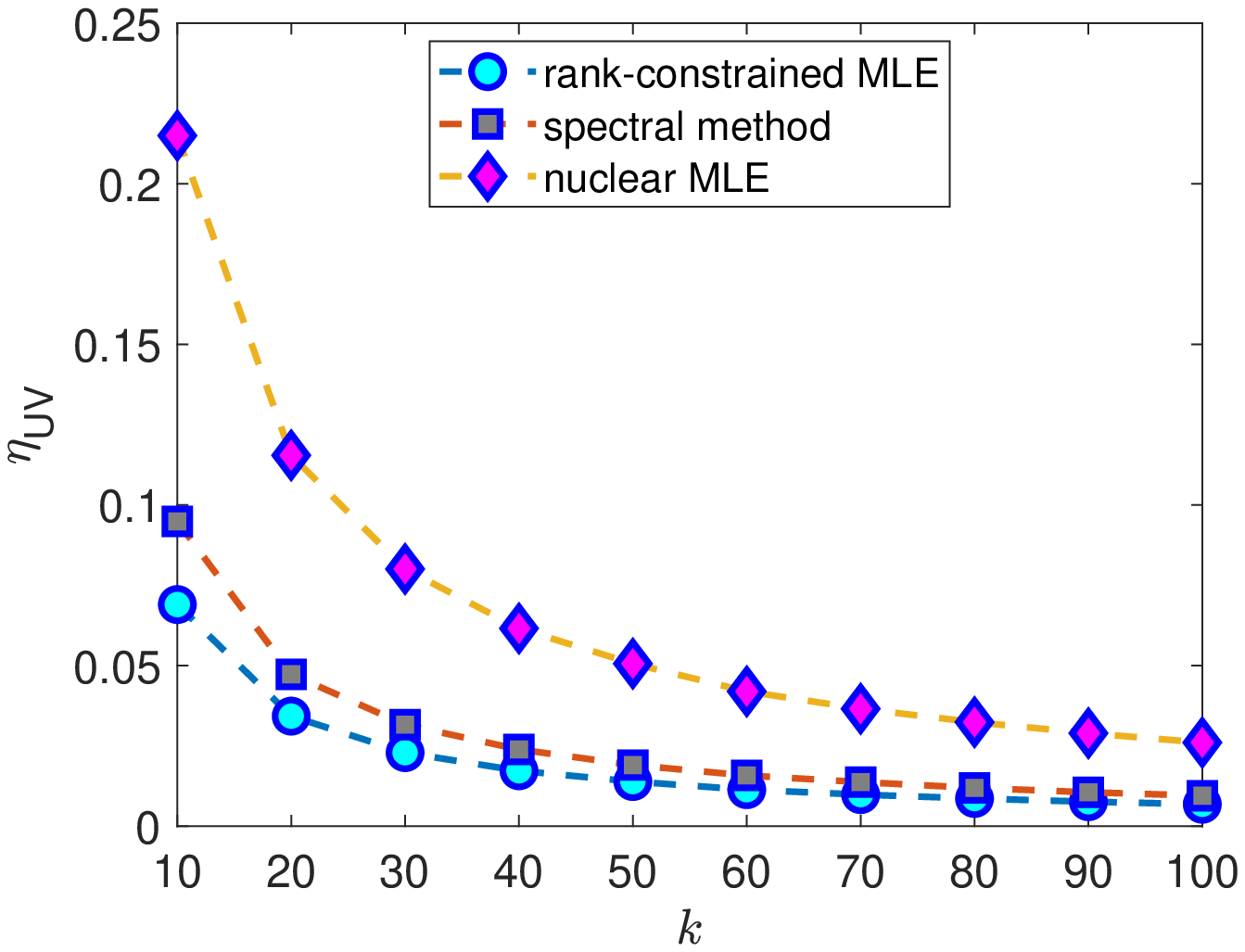}
		%\caption{rank VS. mle}
	\end{subfigure}
	
	\caption{The first row compares the rank-constrained estimator, nuclear norm penalized estimator, spectral method, and empirical estimator in terms of $\eta_F = \norm{\P - \widehat \P}_F^2, \eta_{KL}= D_{\KL}(\bP,\widehat\bP)$, and $\eta_{UV} = \max\bigl\{ \norm{\sin \Theta(\widehat{\U},\U,)}_F^2, \norm{\sin \Theta(\widehat{\V},\V)}_F^2 \bigr\}$. The second row provides the zoomed plots of the first row without the empirical estimator. Here, $n = {\rm round}(k  rp\log p)$ with $p = 1,000$, $r = 10$ and $k$ ranging from $10$ to $100$. 
		%\mw{MW: instead of "rank","nu","emp", please just give the full names "rank-constrained MLE", "nuclear MLE", "spectral method" to avoid confusion. There is enough space to display them. Also please give full expressions of $(\eta_F, \eta_U, \eta_V)$. What are the parameters of the experiment? Please state in this caption. Please move all figures into the same page.}
	}
	\label{figure:mlevsrank}
	
		\begin{subfigure}{.33\textwidth}
		\label{figure:betamle-etaF}
		\includegraphics[width=\linewidth]{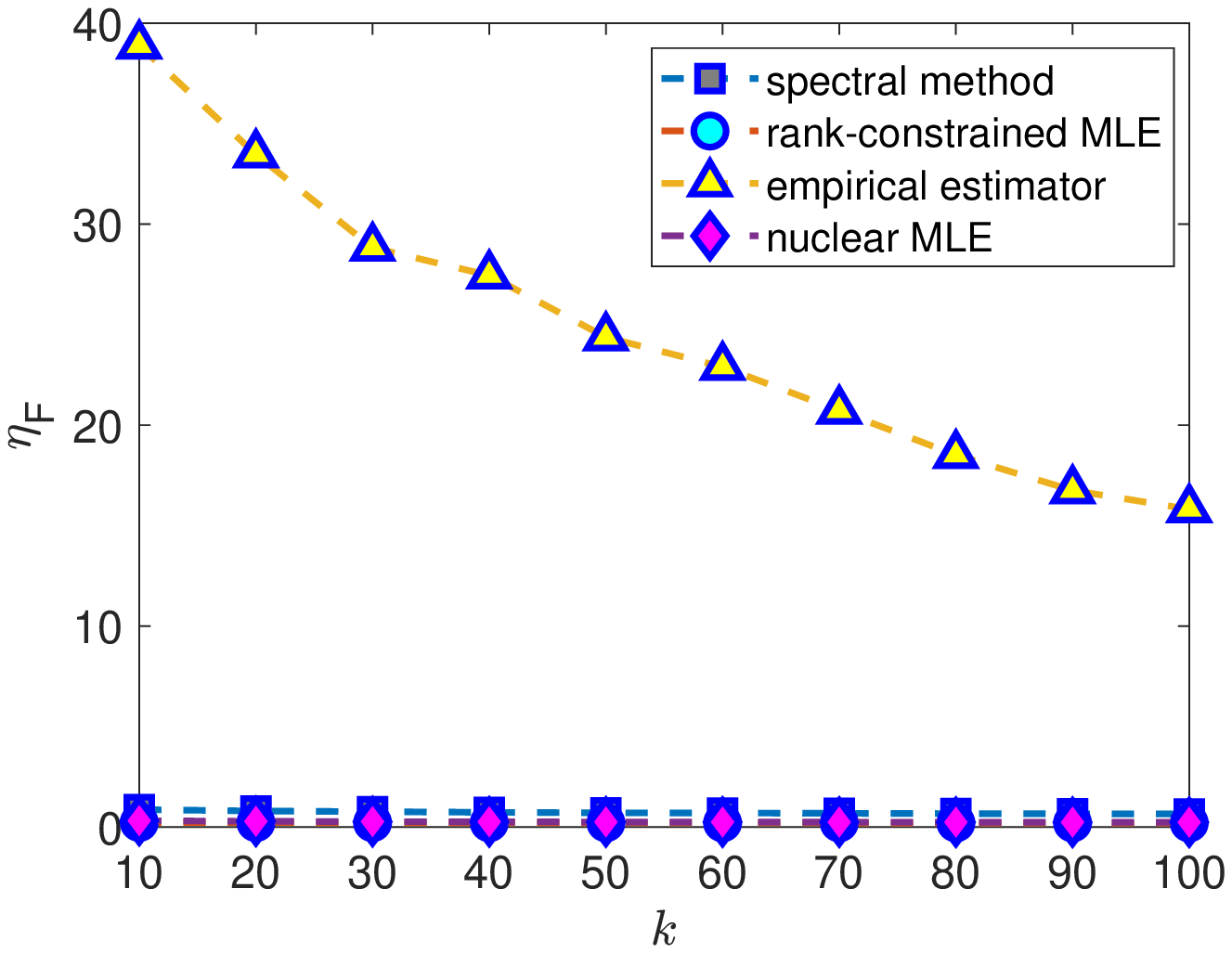}
		%\caption{rank VS. mle}
	\end{subfigure}
	\begin{subfigure}{.33\textwidth}
		\label{figure:betamle-etaKL}
		\includegraphics[width=\linewidth]{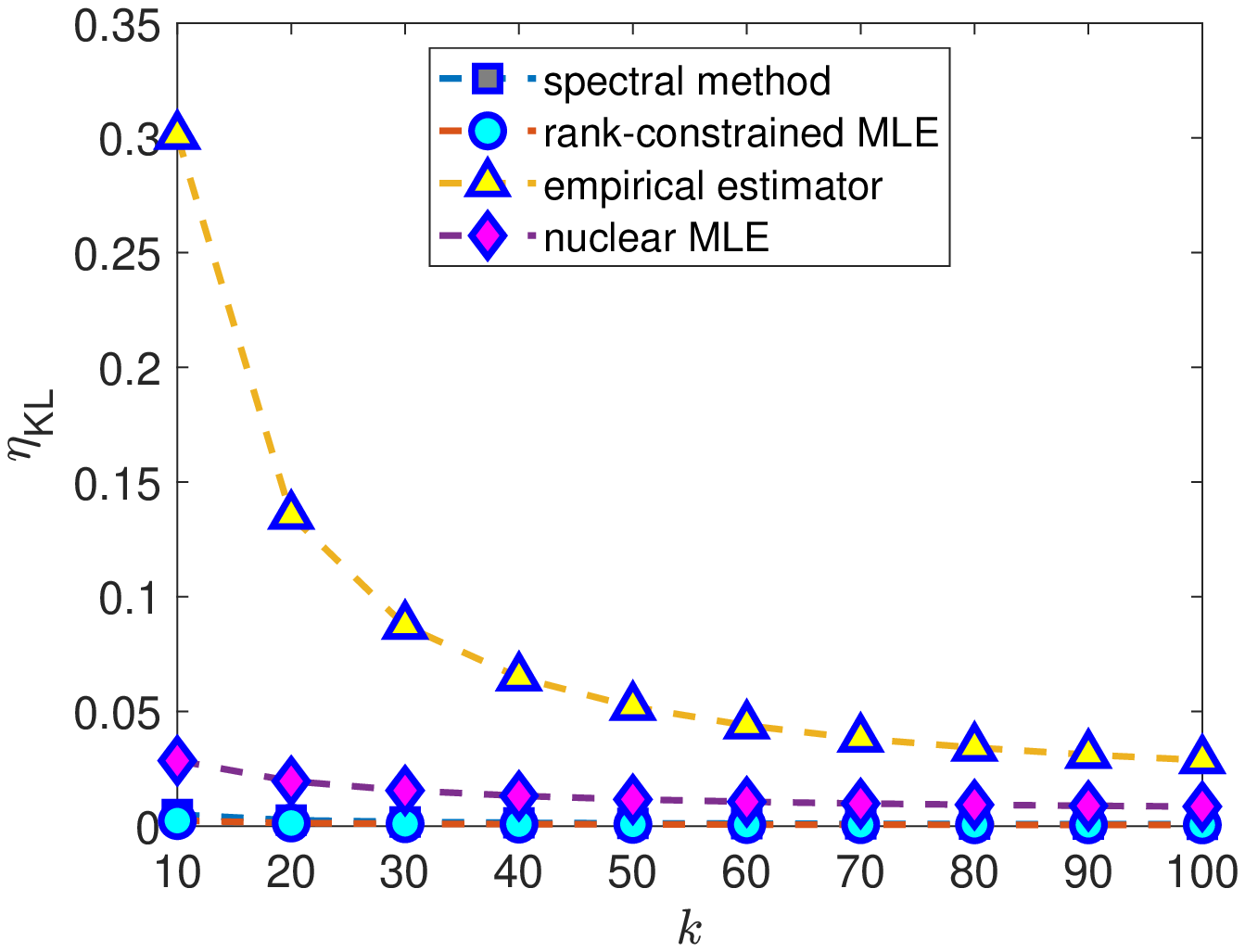}
		%\caption{rank VS. mle}
	\end{subfigure}
	\begin{subfigure}{.33\textwidth}
		\label{figure:betamle-etaUV}
		\includegraphics[width=\linewidth]{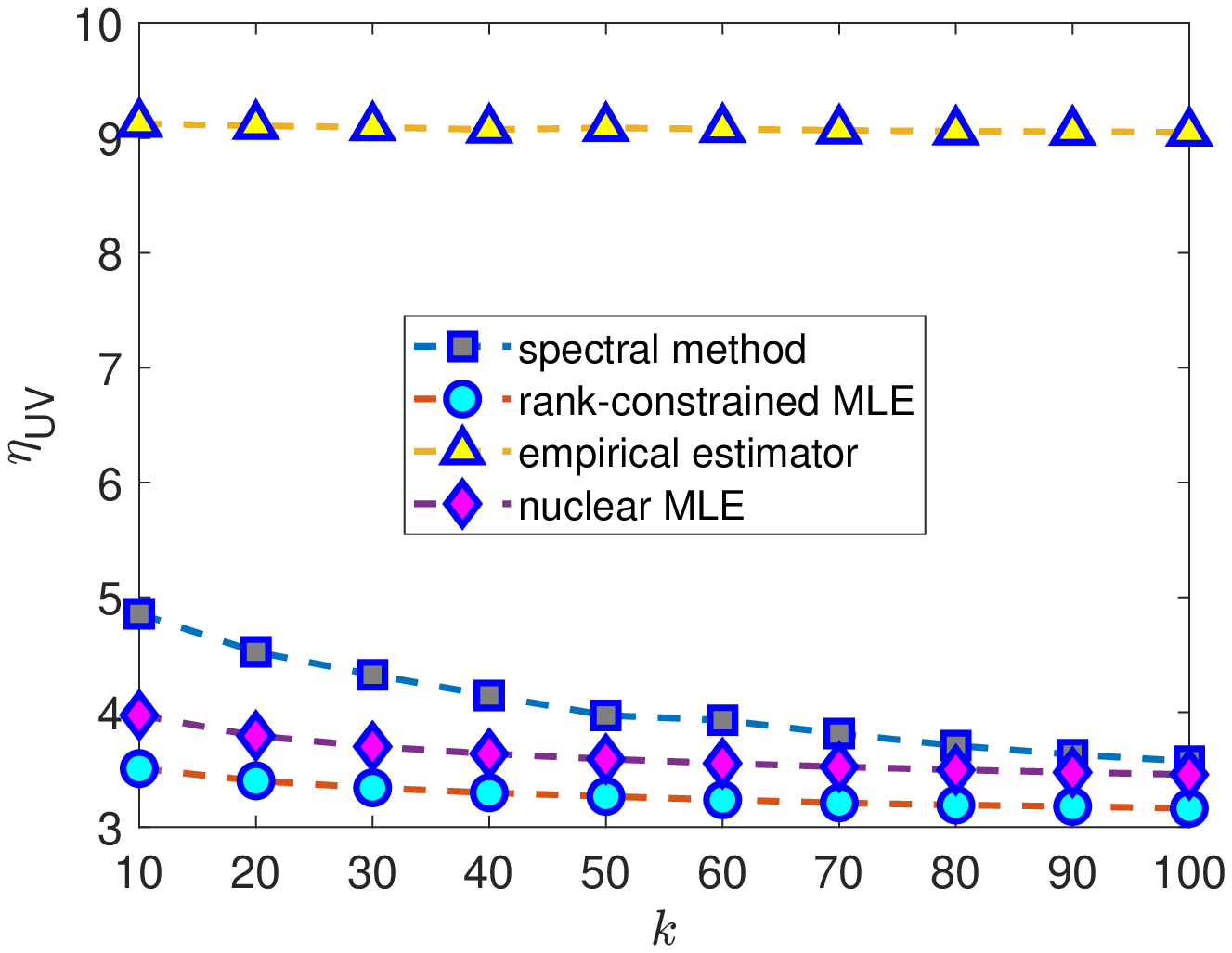}
		%\caption{rank VS. mle}
	\end{subfigure}
	%\caption{rank VS. mle}
	\\
	
	\begin{subfigure}{.33\textwidth}
		\label{figure:betasvd-etaF}
		\includegraphics[width=\linewidth]{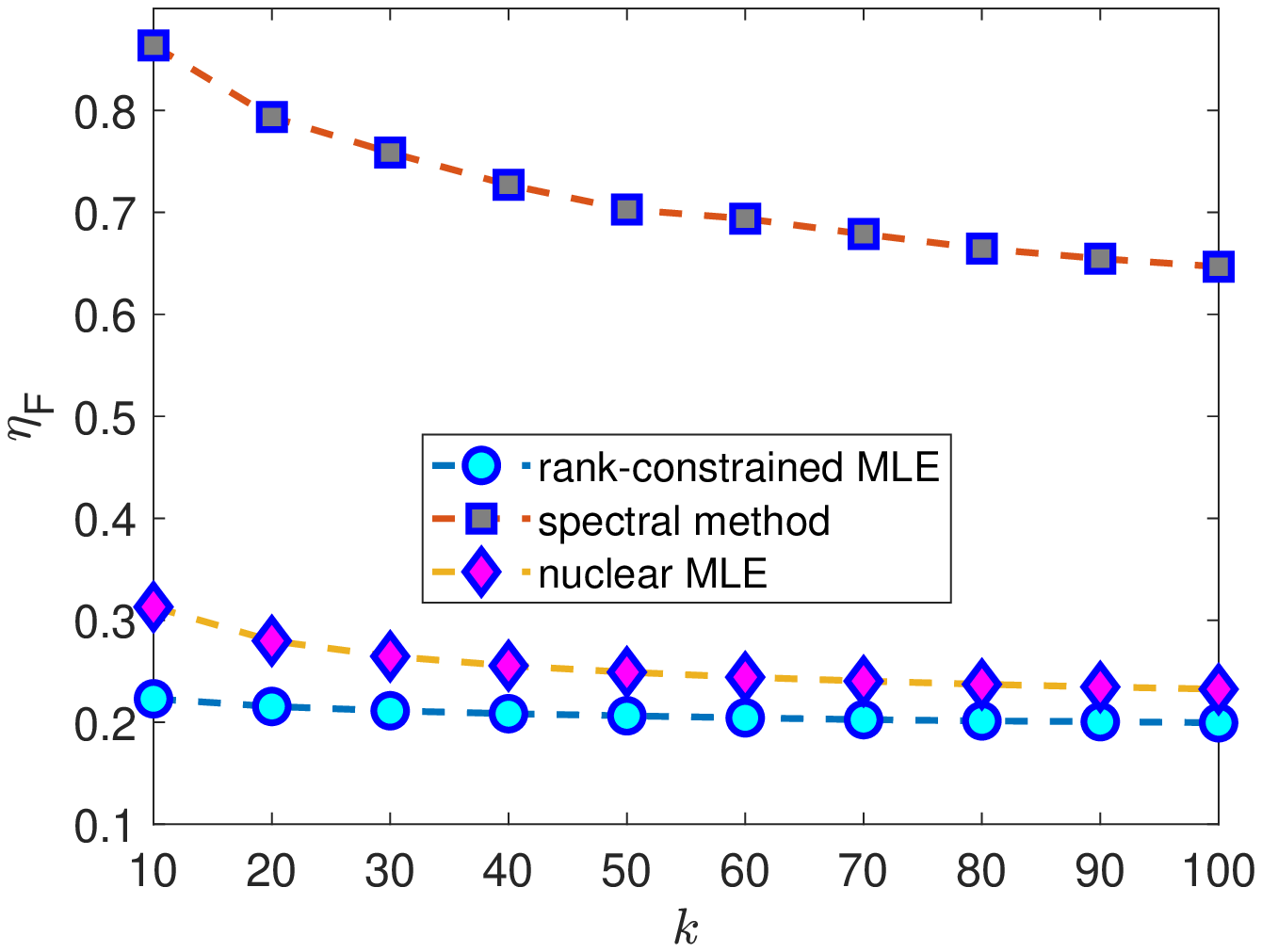}
		%\caption{rank VS. mle}
	\end{subfigure}
	\begin{subfigure}{.33\textwidth}
		\label{figure:betasvd-etaKL}
		\includegraphics[width=\linewidth]{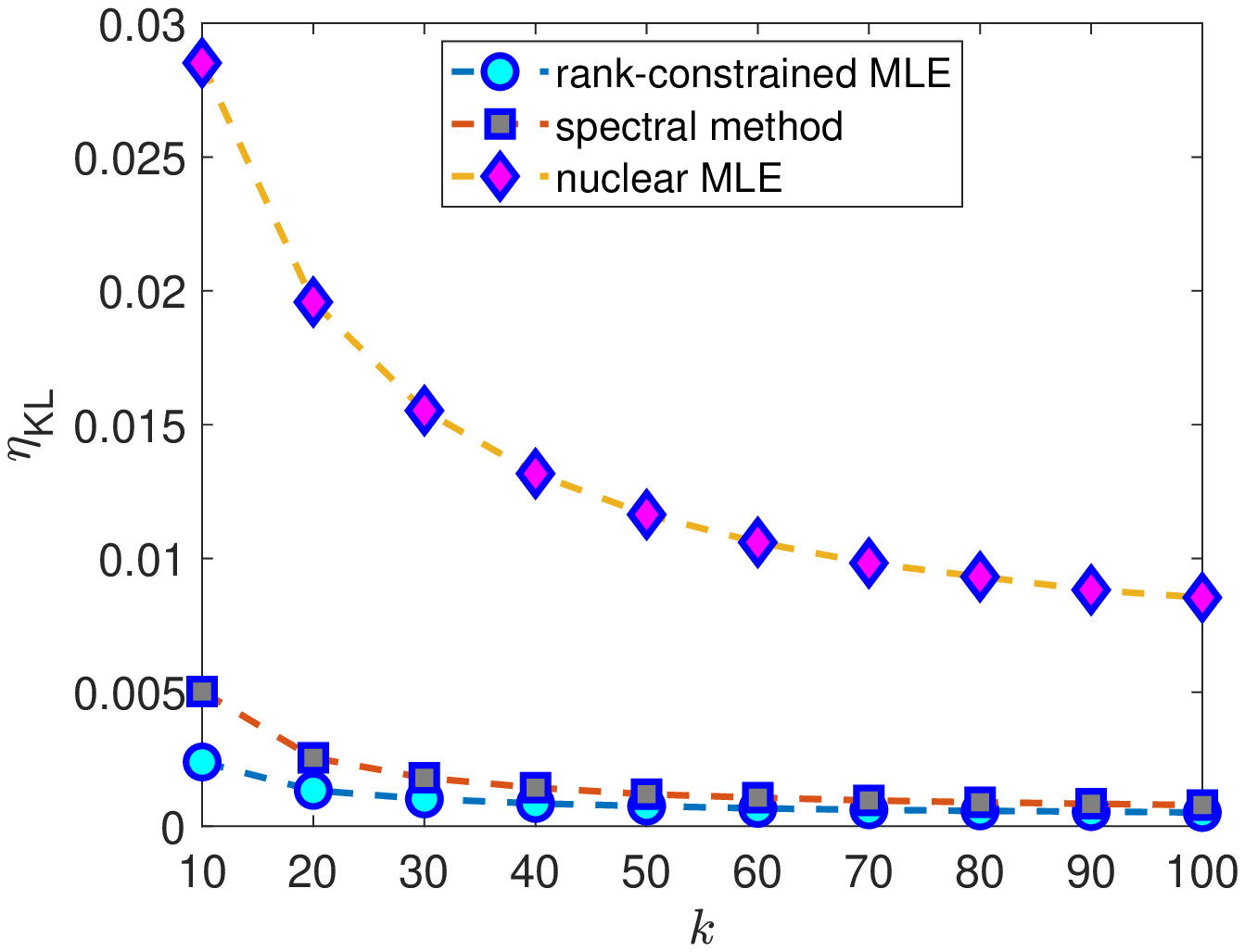}
	\end{subfigure}
	\begin{subfigure}{.33\textwidth}
		\label{figure:betasvd-etaUV}
		\includegraphics[width=\linewidth]{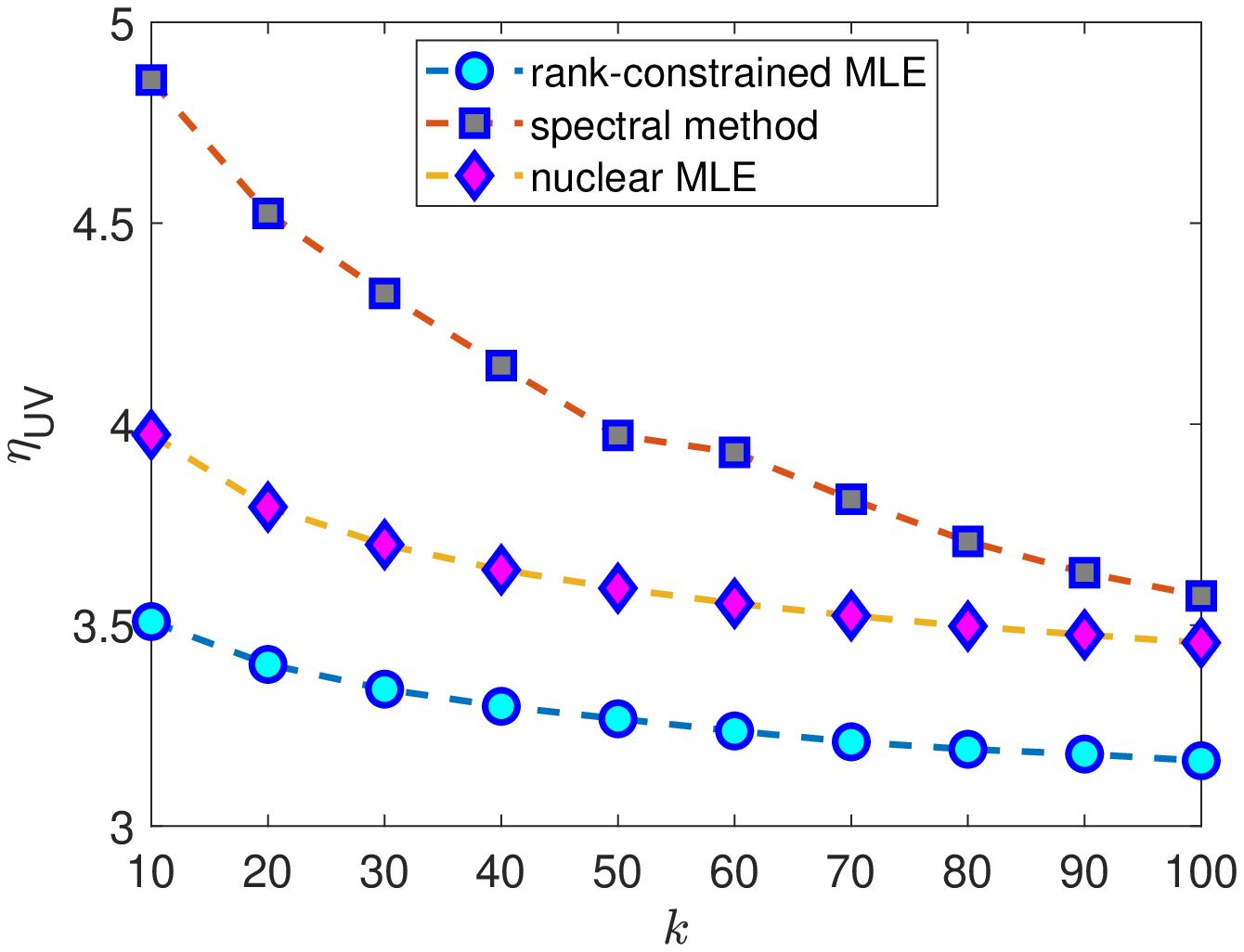}
		%\caption{rank VS. mle}
	\end{subfigure}
	
	\caption{The first row compares the rank-constrained estimator, nuclear norm penalized estimator, spectral method, and empirical estimator in terms of $\eta_F = \norm{\P - \widehat \P}_F^2, \eta_{KL}= D_{\KL}(\bP,\widehat\bP)$, and $\eta_{UV} = \max\bigl\{ \norm{\sin \Theta(\widehat{\U},\U,)}_F^2, \norm{\sin \Theta(\widehat{\V},\V)}_F^2 \bigr\}$ with imbalanced
		invariant distribution. The second row provides the zoomed plots of the first row without the empirical estimator. Here, $n = {\rm round}(k  rp\log p)$ with $p = 1,000$, $r = 10$ and $k$ ranging from $10$ to $100$. 
	}
	\label{figure:betamlevsrank}
\end{figure*}

\subsection{Experiments with Manhattan Taxi data}
\label{sec:taxi}

In this experiment, we analyze a real dataset of $1.1\times 10^7$ trip records of NYC Yellow cabs (Link: \url{https://s3.amazonaws.com/nyc-tlc/trip+data/yellow_tripdata_2016-01.csv}) in January 2016. Our goal is to partition the Manhattan island into several areas, in each of which the taxi customers share similar destination preference. This can provide guidance for balancing the supply and demand of taxi service and optimizing the allocation of traffic resources.

%Each entry of the dataset includes the GPS coordinates of the pick-up and drop-off
%locations, starting and ending time of the trip, distance and duration, and taxifares. 
We discretize the Manhattan island into a fine grid and model each cell of the grid as a state of the Markov chain; each taxi trip can thus be viewed as a state transition of this 
Markov chain \citep{yang2017dynamic, benson2017spacey, liu2012understanding}. For stability concerns, our model ignores the cells that have fewer than $1,000$ taxi visits. Given that the traffic dynamics typically vary over time, 
%We also take into account the time of picking up and dropping off in our analysis. 
we fit the MC under three periods of a day, i.e., $06:00\sim 11:59$ (morning), $12:00 \sim 17:59$ (afternoon) and $18:00 \sim 23:59$ (evening), where the number of the active states $p = 803$, $999$ and $1,079$ respectively. 
%\zzw{Not sure how you incorporate the time periods into the state space. Cartesian product of space and time?}. 
We apply the rank-constrained likelihood approach to obtain the estimator $\widehat \bP^r$ of the transition matrix, and then apply $k$-means to the left singular subspaces of $\widehat{\bP}^r$ to classify all the states into several clusters. %\zzw{left or right singular spaces?} 
Figure \ref{figure:lr4} presents the clustering result with $r = 4$ 
and $k = 4$ 
for the three periods of a day. %and \ref{figure:lr6}, respectively. \zzw{only need $r = 4$? }
%It can be observed from the figures that traffic dynamics exhibits different patterns across different time periods. 

First of all, we notice that the locations within the same cluster are close with each other in geographical distance. This is non-trivial: we do not have exposure to GPS location in the clustering analysis. This implies that taxi customers in neighboring locations have similar destination preference, which is consistent with common sense. Furthermore, to track the variation of the traffic dynamics over time, 
%one can observe that in the evening, the green cluster grows to cover the area that is on the right side of the Central Park. To explain this,  
Figure \ref{figure:pb4} visualizes the distribution of the destination choice that is correspondent to the center of the green cluster in the morning, afternoon and evening respectively. We identify the varying popular destinations in different periods of the day and provide corresponding explanations in the following table: 
\begin{table}[H]
	\centering 
	\def\arraystretch{0.7}	
	\begin{tabular}{c@{\hskip 1cm}p{7cm}@{\hskip 1cm}p{4cm}}
		\hline\hline
		Time & \hfil Popular Destinations & \hfil Explanation \\ \hline\hline
		Morning & New York--Presbyterian Medical Center,\vspace{-.3cm} \newline \hfil 42--59 St. Park Ave, Penn Station & hospitals, workplaces,\vspace{-.3cm} \newline \quad the train station \\ \hline
		Afternoon & \centering 66 St. Broadway & \hfil lunch, afternoon break,\vspace{-.3cm} \newline short trips \\ \hline
		Evening & \centering Penn Station & go home \\\hline\hline 
	\end{tabular}
\end{table}

Finally, it might be tempting to model the taxi trips by an HMM, where regions of Manhattan correspond to hidden states. However, such a region is always part of the current observation (i.e., location of taxi):  It is observable and is not a hidden state that has to be inferred from all past observations. As a result, although both HMM and the low-rank Markov model could apply to taxi trips, the low-rank Markov model is simpler and more accurate. 
%One problem with this model, however, is that each observation needs to be independent of all the historical observations and hidden states given the current hidden state. This implies that the consecutive taxi stops within a Manhattan region are independent of each other, which contradicts the practical situations. 

%Our procedure provides an informative partition of the Manhattan city according to traffic dynamics. 

\begin{figure*}[t!]
	\centering
	%	\begin{subfigure}{.23\textwidth}%[$t = 00:00 \sim 05:59$.]
	%		\label{figure:lrh0}
	%		\includegraphics[width=\linewidth]{figs//icml-lowrank-r4hh0}
	%	\end{subfigure}
	%\hspace{0.5cm}
	\begin{subfigure}{.32\textwidth}%[$t = 06:00 \sim 11:59$.]{
		\label{figure:lrh1}
		\includegraphics[width=\linewidth]{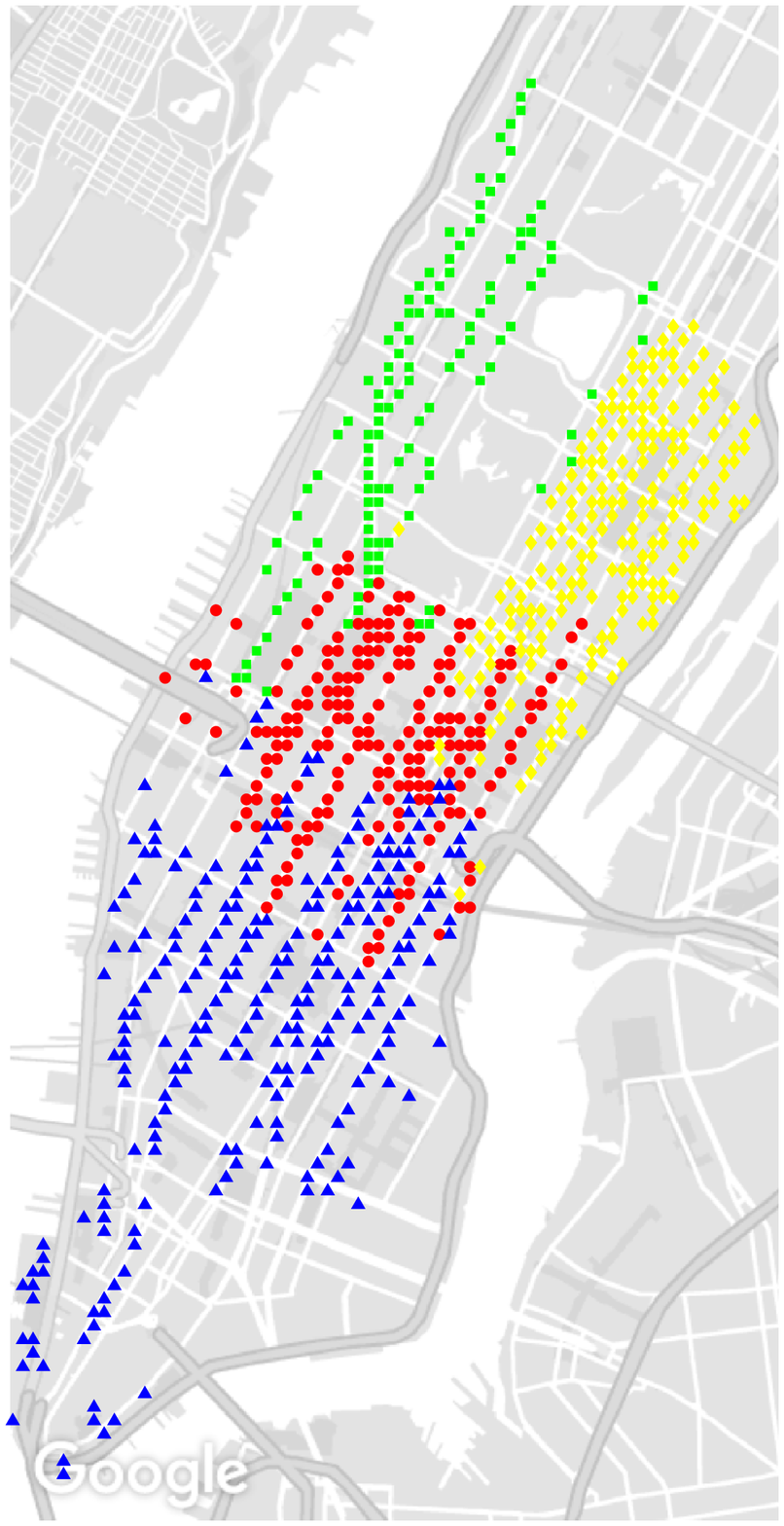}
	\end{subfigure} 
	\begin{subfigure}{.32\textwidth}%[$t = 12:00 \sim 17:59$.]{
		\label{figure:lrh2}
		\includegraphics[width=\linewidth]{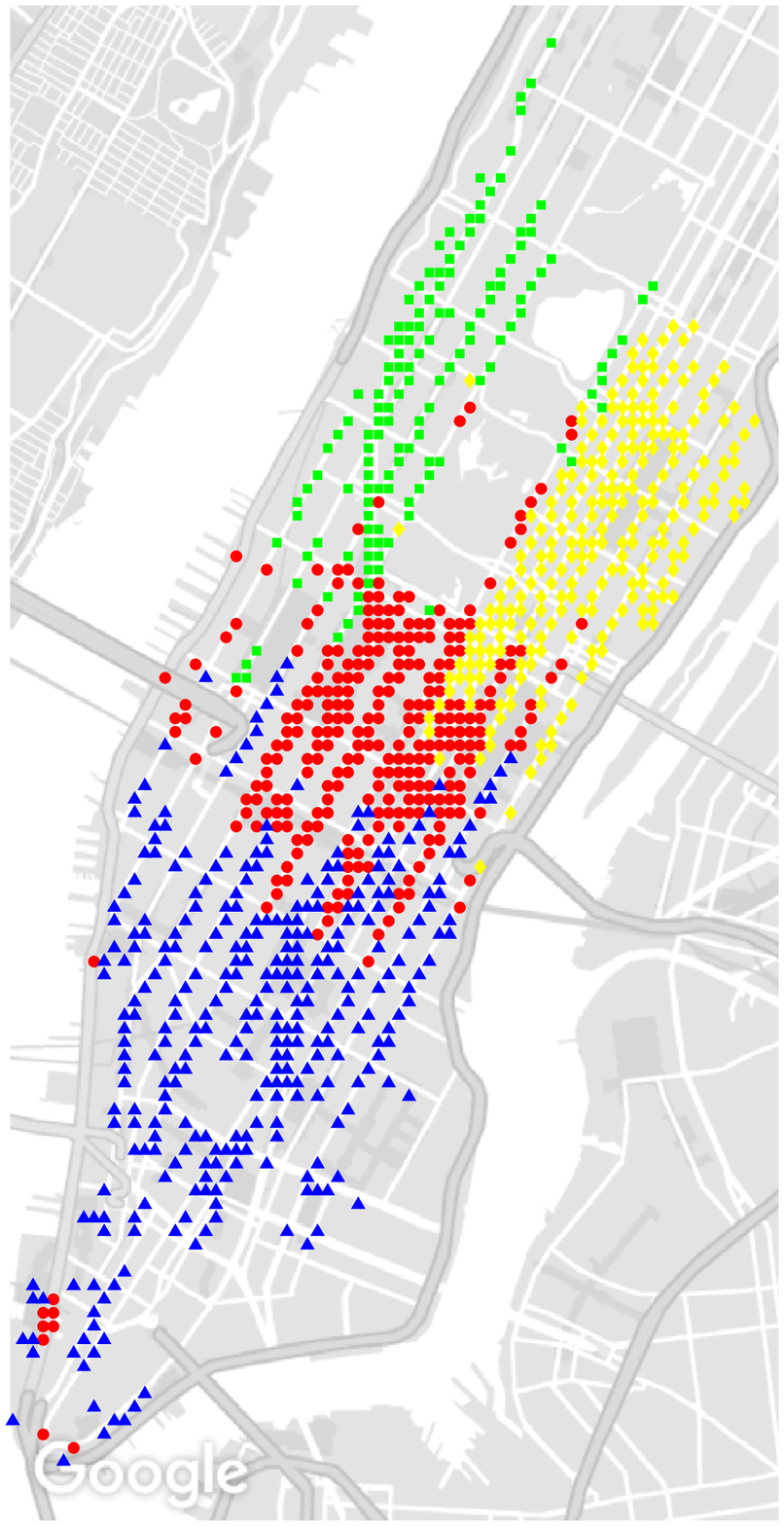}
	\end{subfigure}
	%\hspace{0.5cm}
	\begin{subfigure}{.32\textwidth}%[$t = 18:00 \sim 23:59$.]{
		\label{figure:lrh3}
		\includegraphics[width=\linewidth]{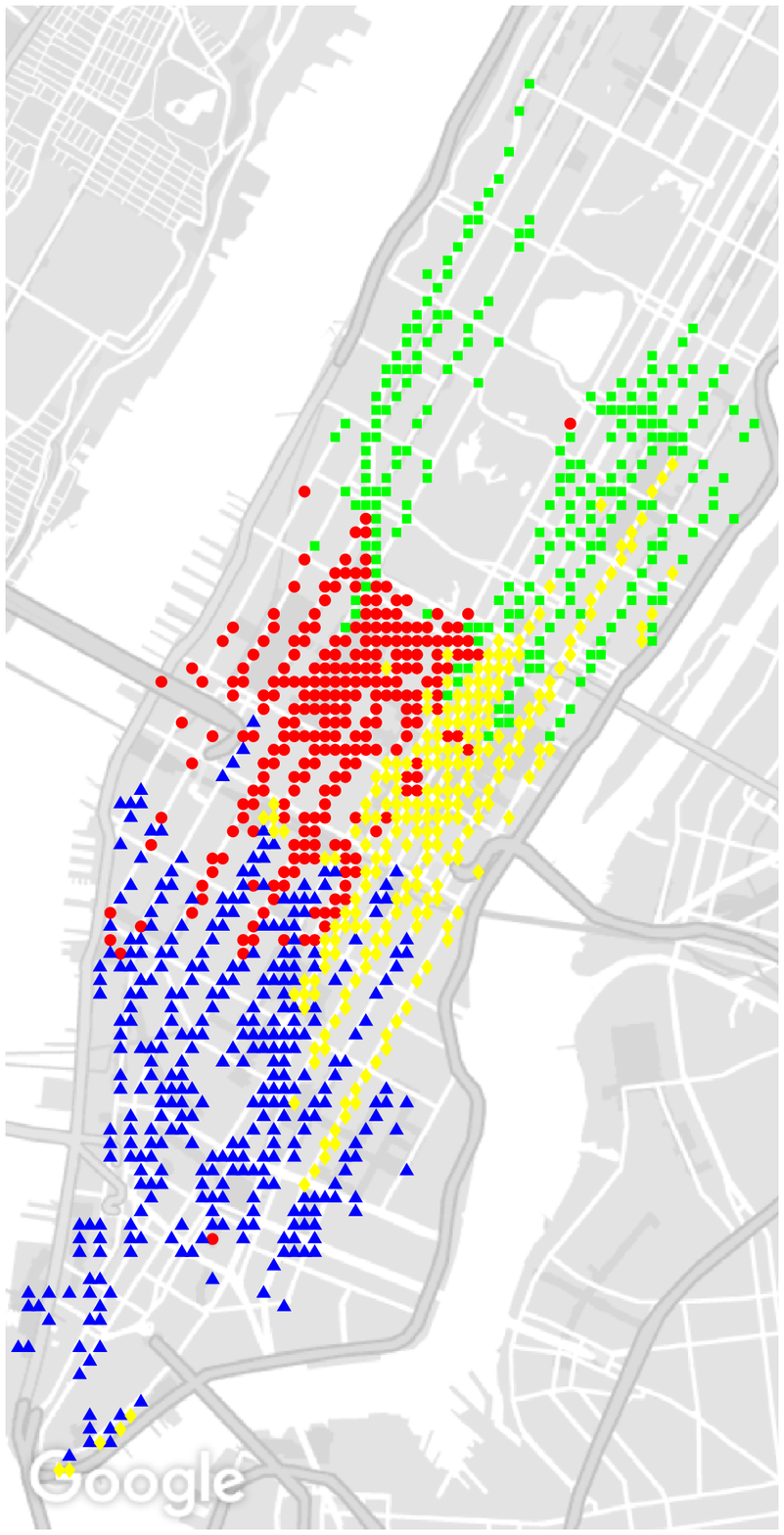}
	\end{subfigure}
	\caption{
		The meta-states compression of Manhattan traffic network via rank-constrained approach with $ r = 4$: mornings (left), afternoons (middle) and evenings (right). Each color or symbol represents a meta-state. One
		can see the day-time state aggregation
		results differ significantly from that of the evening time.}
	\label{figure:lr4}
\end{figure*}

\begin{figure*}[tb]
	\centering
	\begin{subfigure}{.32\textwidth}%[$t = 06:00 \sim 11:59$.]{
		\label{figure:lr6h1}
		\includegraphics[width=\linewidth]{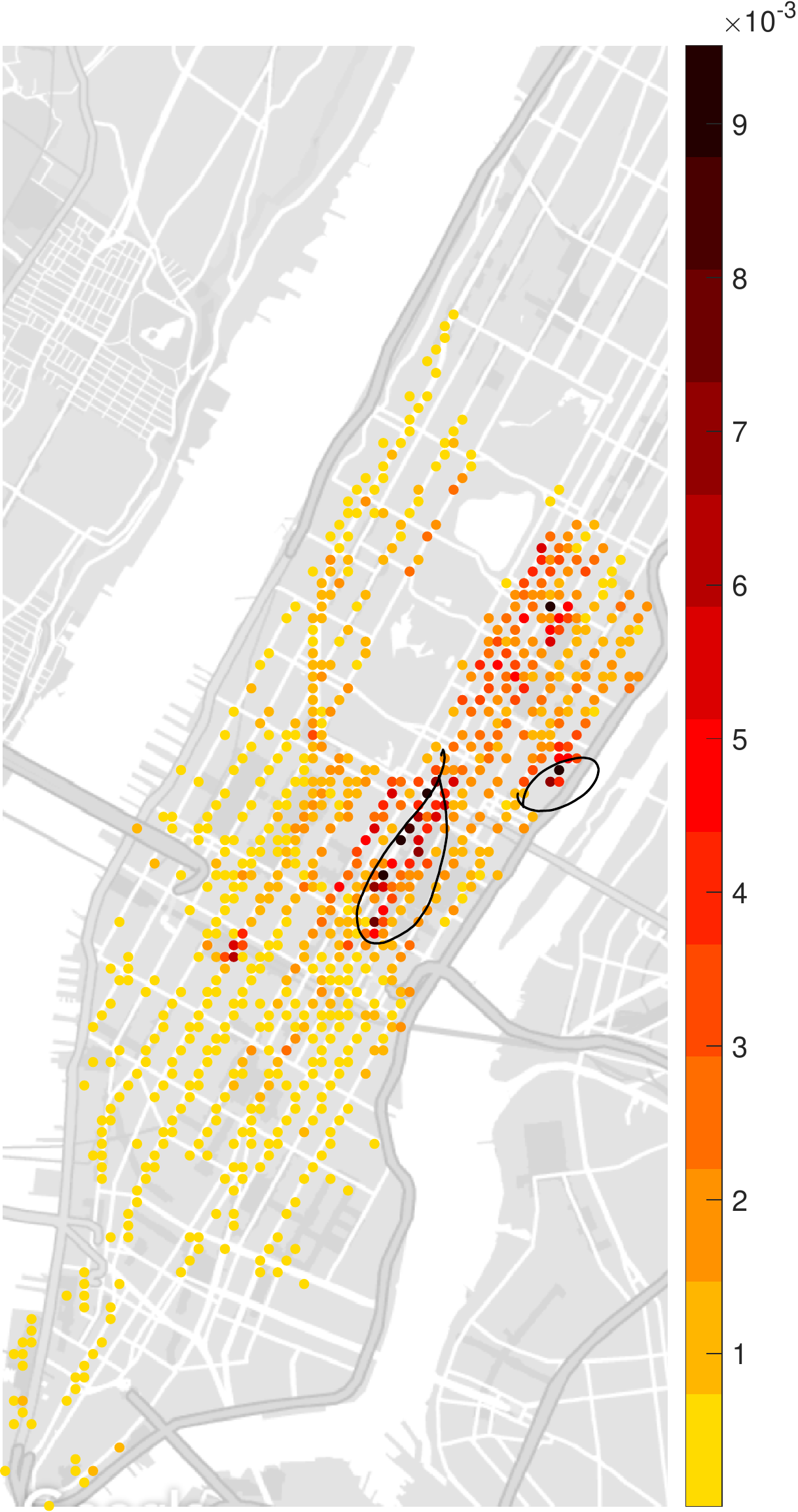}
	\end{subfigure} 
	\begin{subfigure}{.32\textwidth}%[$t = 12:00 \sim 17:59$.]{
		\label{figure:lr6h2}
		\includegraphics[width=\linewidth]{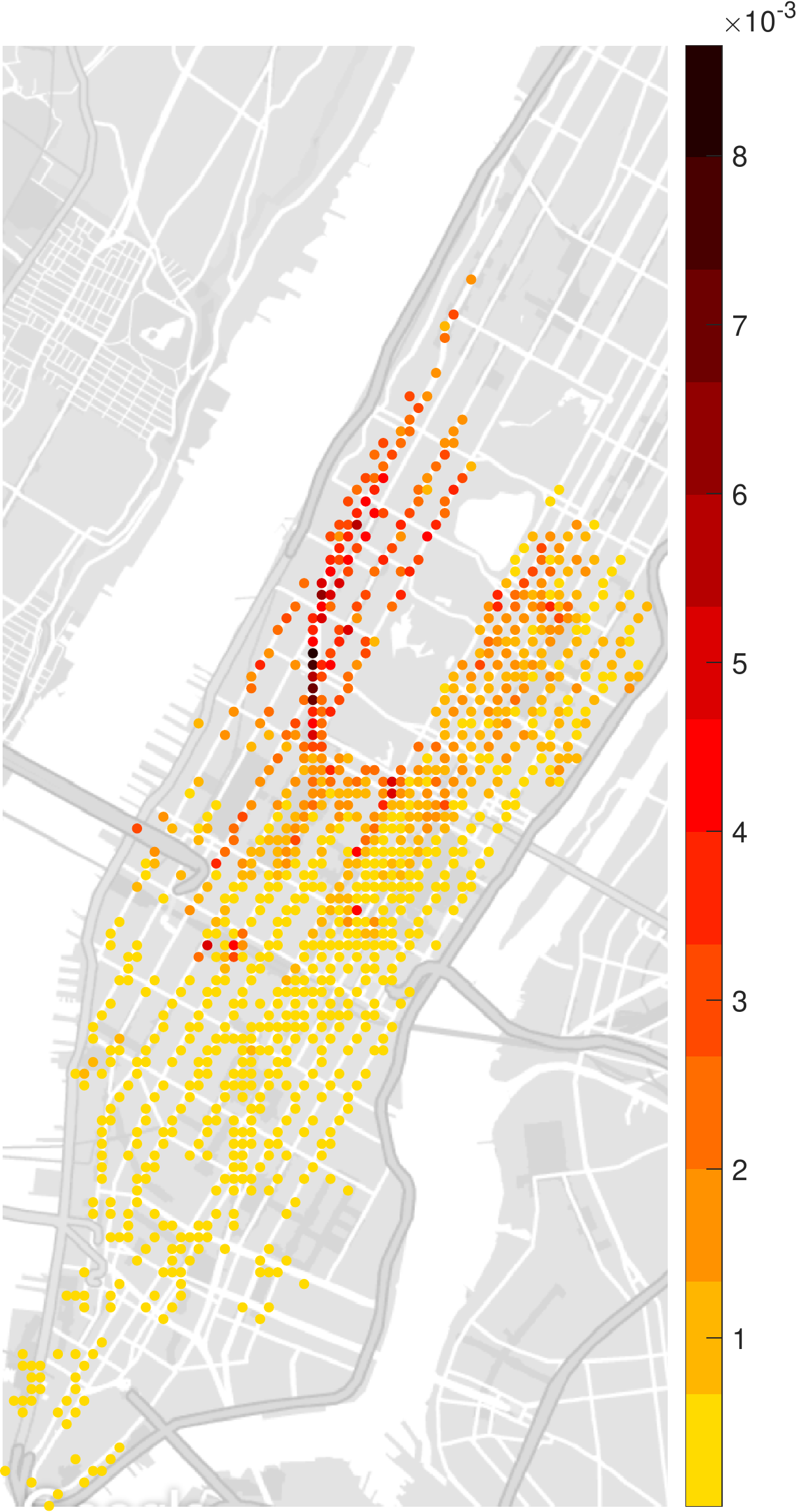}
	\end{subfigure}
	%\hspace{0.5cm}
	\begin{subfigure}{.32\textwidth}%[$t = 18:00 \sim 23:59$.]{
		\label{figure:lr6h3}
		\includegraphics[width=\linewidth]{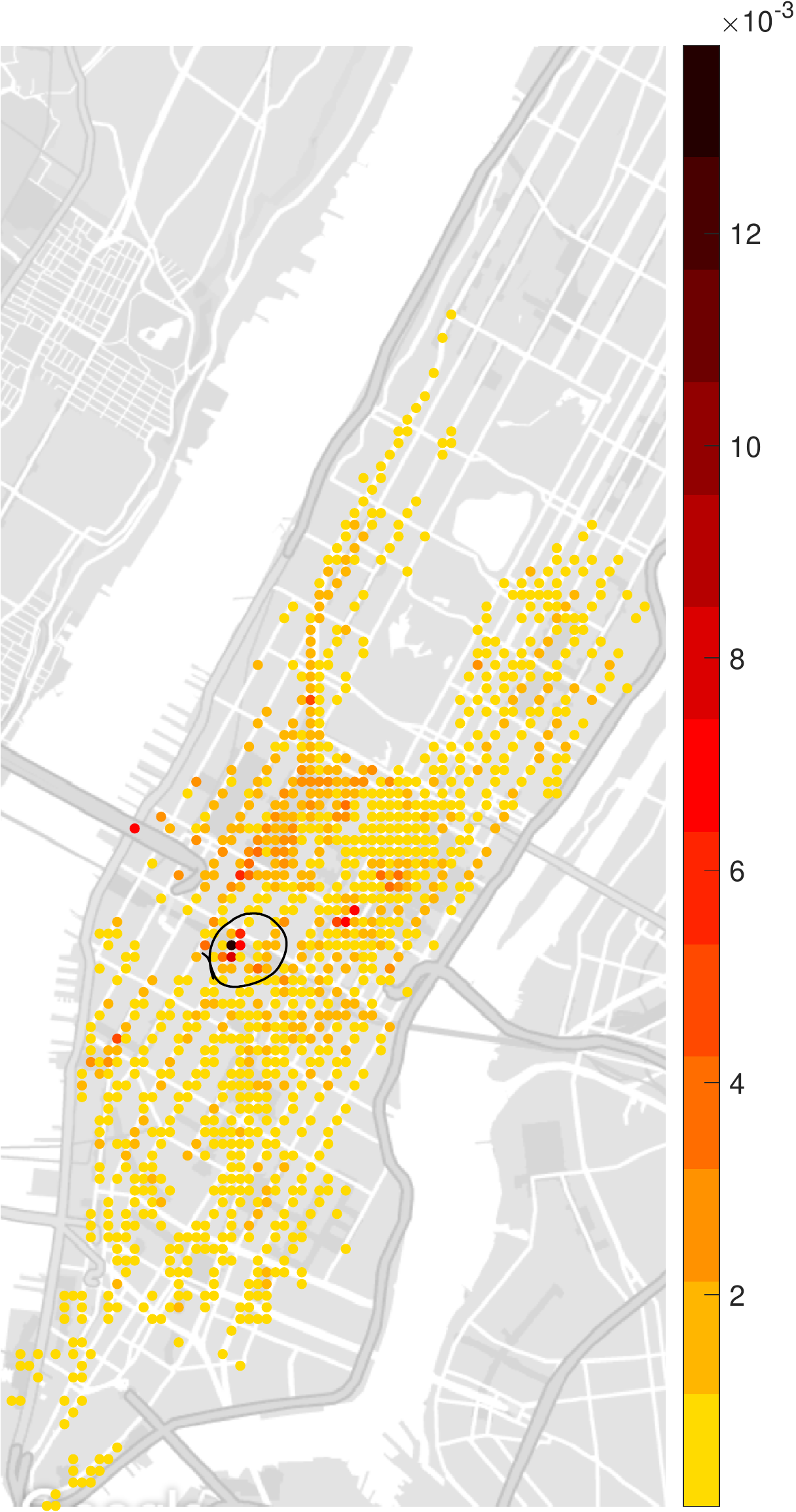}
	\end{subfigure}
	\caption{Visualization of the destination distributions corresponding to the pick-up locations in the green clusters in Figure \ref{figure:lr4}: mornings (left), afternoons (middle) and evenings (right).}
	\label{figure:pb4}
\end{figure*}

\section{Conclusion}
\label{sec:conclusion}
This paper studies the recovery and state compression of low-rank Markov chains from empirical trajectories via a rank-constrained likelihood approach. We provide statistical upper bounds for the $\ell_2$ risk and Kullback-Leiber divergence between the estimator and the true probability transition matrix for the proposed estimator. 
Then, a novel DC programming algorithm is developed to solve the associated rank-constrained optimization problem. The proposed algorithm non-trivially combines several recent optimization techniques, such as the penalty approach, the proximal DC algorithm, and the multi-block sGS-ADMM. We further study a new class of majorized indefinite-proximal DC algorithms for solving general non-convex non-smooth DC programming problems and provide a unified convergence analysis.
%, we proposed a penalty approach to transform the original difficult problem into a DC programming problem.
%Then, we solved this problem via a proximal DCA. \red{To Xudong, can you write this sentence like it is more novel?} 
%Experiments on both simulated data and real data on Manhattan taxi trips illustrate the merits of our approach.
Experiments on simulated data illustrate the merits of our approach.

% References here (outcomment the appropriate case)

% CASE 1: BiBTeX used to constantly update the references
%   (while the paper is being written).
%\bibliographystyle{informs2014} % outcomment this and next line in Case 1
%\bibliography{<your bib file(s)>} % if more than one, comma separated

% CASE 2: BiBTeX used to generate mypaper.bbl (to be further fine tuned)
%\input{mypaper.bbl} % outcomment this line in Case 2

%If you don't use BiBTex, you can manually itemize references as shown below.

%\begin{thebibliography}{informs2014}
\bibliographystyle{informs2014}
\bibliography{./arXiv-lowrankMarkov}

%\end{thebibliography}

%% Here starts the e-companion (EC)
%%%%%%%%%%%%%%%%%%%%%%%%%%%%%%%%%%%%%%%%%%%%%%%%%%%%%%%%%%
\ECSwitch

%\ECDisclaimer
%%%%%%%%%%%%%%%%%%%%%%%%%%%%%%%%%%%%%%%%%%%%%%%%%%%%%%%%%%

%%% Main head for the e-companion
\ECHead{Technical lemmas and proofs}

\section{Technical lemmas}

\begin{lemma}
		\label{lem:kl_to_l2}
		Given two discrete distributions $u, v \in \R^p$, if there exist $\alpha, \beta > 0$ such that $u_j \in \{0\}\cup [\alpha / p, \beta / p]$ and $v_j \in [\alpha / p, \beta / p]$ for any $j \in [p]$, then we have $$\DKL(u, v) \ge \{p\alpha / (2\beta ^ 2)\} \ltwonorm{u - v} ^ 2.$$ This implies that under Assumption \ref{asp:1}, for any $\bQ \in \cC$, $$\fnorm{\bP - \bQ} ^ 2 \le \frac{2\beta^2}{\alpha\pi_{\min}p}\DKL(\bP, \bQ).$$
\end{lemma}
\begin{proof}{Proof of Lemma \ref{lem:kl_to_l2}}
		\noindent By the mean value theorem, for any $j \in [p]$ such that $u_j \neq 0$, there exists $\xi_j \in [\alpha / p, \beta / p]$ such that 
		\[
			\log(v_j) - \log(u_j) = \frac{v_j - u_j}{u_j} - \frac{(v_j - u_j) ^ 2}{2\xi_j ^2}. 
		\]
		Therefore, 
		\[
			\begin{aligned}
				\DKL(u, v) & = \sum_{j: u_j \neq 0} u_j\log(u_j / v_j) = \sum_{j: u_j \neq 0} (u_j - v_j) + \sum_{j: u_j \neq 0} \frac{(u_j - v_j) ^ 2}{2 \xi_j ^ 2} \\
				& \ge  1 - \sum_{j: u_j \neq 0} v_j + \sum_{j: u_j \neq 0} \frac{p \alpha(u_j - v_j) ^ 2}{2\beta ^ 2} = \sum_{j: u_j = 0} v_j - u_j + \sum_{j: u_j \neq 0} \frac{p \alpha(u_j - v_j) ^ 2}{2\beta ^ 2} \\
				& \ge  \sum_{j: u_j = 0} \frac{p(v_j - u_j) ^ 2}{\beta} + \sum_{j: u_j \neq 0} \frac{p \alpha(u_j - v_j) ^ 2}{2\beta ^ 2}  \ge \frac{p\alpha}{2\beta ^ 2} \ltwonorm{u - v} ^ 2. 
			\end{aligned} 
		\]
		Then we have
		\[
			\fnorm{\bP - \bQ} ^ 2 = \sum_{i \in [p]} \ltwonorm{P_{i\cdot} - Q_{i\cdot}} ^ 2 \le \sum_{i \in [p]} \frac{2\beta ^ 2\pi_{i}}{p\alpha \pi_{\min}}\DKL(P_{i\cdot}, Q_{i\cdot}) = \frac{2\beta ^ 2}{p\alpha \pi_{\min}}\DKL(\bP, \bQ). 
		\]
\end{proof}
\section{Proof of Theorem \ref{thm:nuclear}}

\begin{proof}
	\noindent Given the definition of $\widehat \bP$, 
	\begin{equation}\label{ineq:tilde_D-P-hat-P}
		\widetilde{D}_{\KL}(\bP,\widehat{\bP}) = \frac{1}{n}\sum_{i=1}^n \langle\log(\bP) - \log(\widehat{\bP}), \bX_i\rangle = \ell_n(\widehat{\bP}) - \ell_n(\bP) \leq \lambda (\nnorm{\widehat \bP} - \nnorm{\bP}) \le \lambda \nnorm{\bP - \widehat \bP}.
	\end{equation}
	Then we have 
	\be
	\label{ineq:basic}
	\begin{aligned}
		D_{\KL}(\bP, \widehat\bP) & = \cL(\widehat\bP) - \cL(\bP) = \cL(\widehat\bP) - \ell_n(\widehat\bP) + \ell_n(\widehat\bP) - \ell_n(\bP) + \ell_n(\bP) - \cL(\bP) \\
		& \le \cL(\widehat \bP) - \ell_n(\widehat \bP) + \ell_n(\bP) - \cL(\bP) + \lambda \nnorm{\bP - \widehat \bP} \\
		& = D_{\KL}(\bP, \widehat\bP) - \widetilde D_{\KL}(\bP, \widehat\bP) + \lambda \nnorm{\bP - \widehat \bP}. 
	\end{aligned}
	\ee
	Define $\cE := \{\lambda \ge 2 \opnorm{\Pi_{\cN}(\nabla \ell_n(\bP))}\}$. If $\cE$ holds, then by Lemma \ref{lem:large_lambda} and then Lemma \ref{lem:kl_to_l2}, we obtain that 
	\[
		\begin{aligned}
			D_{\KL}(\bP, \widehat\bP) & \le D_{\KL}(\bP, \widehat\bP) - \widetilde D_{\KL}(\bP, \widehat\bP) + 4(2r)^{1 / 2}\lambda \fnorm{\bP - \widehat \bP} \\
			& 	\le  D_{\KL}(\bP, \widehat\bP) - \widetilde D_{\KL}(\bP, \widehat\bP) + 8\lambda \beta\biggl(\frac{r D_{\KL}(\bP, \widehat \bP)}{p\pi_{\min} \alpha}\biggr)^{\!1/2}. 
		\end{aligned}
	\]
	For any $\xi > 1$, an application of Lemma \ref{lem:uniform_law} with $\eta = \xi \pi_{\min}/ (rp\pi_{\max}\log p)$ yields
	\[
		\begin{aligned}
			\PP\biggl[ \biggl\{D_{\KL}(\bP, \widehat \bP) \le 16\lambda \beta\biggl(\frac{r D_{\KL}(\bP, \widehat \bP)}{p\pi_{\min} \alpha}\biggr)^{\!1/2} +  \frac{2C_1 r\pi_{\max}\beta ^ 2 p\log p}{\pi_{\min}\alpha ^ 3n} + \eta \biggr\} \cap \cE\biggr] \ge 1 - C_2 e^{- \xi} - \PP(\cE^c), 
		\end{aligned}
	\]	
%	\frac{2\alpha \xi ^ 2}{rp^2 \pi_{\max} \log p}
	where $C_1$ and  $C_2$ are exactly the same constants as in Lemma \ref{lem:uniform_law}. Some algebra yields that 
	\[
		\begin{aligned}
			\PP\biggl[ \biggl\{D_{\KL}(\bP, \widehat \bP) \le \frac{256\lambda ^ 2 \beta ^ 2 r}{p\pi_{\min} \alpha} +  \frac{2C_1 r\pi_{\max}\beta ^ 2 p\log p}{\pi_{\min}\alpha ^ 3n} + \eta \biggr\} \cap \cE\biggr] \ge 1 - C_2 e^{- \xi} - \PP(\cE^c). 
		\end{aligned}
	\]		
%	One can thus deduce that 
%	\[
%		\PP \biggl(\biggl\{D_{\KL}(\bP, \widehat\bP^r) \gtrsim \frac{ r \lambda ^ 2\beta ^ 2}{\alpha ^ 2} + \frac{2 r\pi_{\max}\beta ^ 2 p ^ 2\log p}{\alpha ^ 4n} + \frac{\alpha ^ 2\xi}{rp ^ 2 \pi_{\max}\log p}\biggr\} \cap \cE \biggr) \le 2C_2e^{- \xi}. 
%	\]
	By Lemma \ref{lem:gradient}, there exists a universal constant $C_3 > 0$ such that if we choose 
	\[
		\lambda = C_3\biggl\{\biggl(\frac{\xi p ^ 2\pi_{\max} \log p}{n\alpha}\biggr)^{1 / 2} + \frac{\xi p\log p}{n \alpha}\biggr\}, 
	\]
	then for any $\xi > 1$, whenever $n\pi_{\max}(1 - \rho_+) \ge \max(20, \xi ^ 2)\log p$, we have that  
	\[
		\begin{aligned}
			\PP \biggl( D_{\KL}(\bP, \widehat\bP^r) \gtrsim \frac{\xi r\pi_{\max}\beta ^ 2 p \log p}{\pi_{\min}\alpha ^ 3n} + \frac{\xi\pi_{\min}}{rp \pi_{\max}\log p}\biggr) \lesssim e^{- \xi} + p^{-(\xi - 1)} + p^{-10},
		\end{aligned}
	\]
	as desired. The Frobenius-norm error bound follows immediately by applying Lemma \ref{lem:kl_to_l2}. 
%	once we notice that 
%	\begin{equation}
%		\label{eq:f_kl_conversion}
%		\begin{split}
%			\|\widehat{\bP} - \bP\|_F^2 \leq & \sum_{i=1}^p \|P_{i\cdot} - \widehat{P}_{i\cdot}\|_2^2 \leq \sum_{i=1}^p \frac{2\beta^2}{\alpha p}\DKL(P_{i\cdot}, \widehat{P}_{i\cdot}) \leq  \sum_{i=1}^p \frac{2\beta^2}{\alpha^2 }\mu_i \DKL(P_{i\cdot}, \widehat{P}_{i\cdot}) = \frac{\beta^2}{\alpha^2}\DKL(\bP, \widehat{\bP}), 
%		\end{split}
%	\end{equation}
%	where we use \citet[][Lemma~4]{ZWa19} in the second inequality. 
\end{proof}

\section{Proof of Theorem \ref{thm:rank}}

\begin{proof}	
	\noindent Given the definition of $\widehat \bP^r$, 
	\begin{equation}\label{ineq:tilde_D-P-hat-P}
	\widetilde{D}_{\KL}(\bP,\widehat{\bP}^r) = \frac{1}{n}\sum_{i=1}^n \langle\log(\bP) - \log(\widehat{\bP}^r), \bX_i\rangle = \ell_n(\widehat{\bP}^r) - \ell_n(\bP) \leq 0.
	\end{equation}
	Then we have
	\be
	\label{ineq:basic}
	\begin{aligned}
		D_{\KL}(\bP, \widehat\bP^r) & = \cL(\widehat\bP^r) - \cL(\bP) = \cL(\widehat\bP^r) - \ell_n(\widehat\bP^r) + \ell_n(\widehat\bP^r) - \ell_n(\bP) + \ell_n(\bP) - \cL(\bP) \\
		& \le \cL(\widehat \bP^r) - \ell_n(\widehat \bP^r) + \ell_n(\bP) - \cL(\bP) = D_{\KL}(\bP, \widehat\bP^r) - \widetilde D_{\KL}(\bP, \widehat\bP^r). 
	\end{aligned}
	\ee
	For any $\xi > 1$, an application of Lemma \ref{lem:uniform_law} with $\eta = \pi_{\min}\xi / (rp\pi_{\max}\log p)$ yields
	\[
		\PP\biggl\{D_{\KL}(\bP, \widehat \bP^r) \ge \max\biggl(\frac{2C_1 r\pi_{\max} \beta^2 p \log p }{\pi_{\min}\alpha ^ 3n}, \frac{\xi\alpha ^ 2}{rp^ 2 \pi_{\max} \log p}\biggr)\biggr\} \le C_2e^{- \xi}, 
	\]
	as desired. The Frobenius-norm error bound immediately follows by Lemma \ref{lem:kl_to_l2}. 
%	\begin{equation*}
%	\begin{split}
%		\|\hat{\bP} - \bP\|_F^2 \leq & \sum_{i=1}^p \|P_{i\cdot} - \hat{P}_{i\cdot}\|_2^2 \leq \sum_{i=1}^p \frac{2\beta^2}{\alpha p}\DKL(P_{i\cdot}, \hat{P}_{i\cdot}) \\
%		\leq & \sum_{i=1}^p \frac{2\beta^2}{\alpha^2 }\mu_i \DKL(P_{i\cdot}, \hat{P}_{i\cdot}) = \frac{\beta^2}{\alpha^2}\DKL(\bP, \hat{\bP}).
%	\end{split}
%	\end{equation*}
%	Therefore, we have finished the proof for Theorem 1. %\ref{th:upper-bound}. 
	
\end{proof}

\section{Proof of Theorem \ref{thm:lower_bound}}

\begin{proof}
	~\hspace{-.5cm} To simplify the notation, assume without loss of generality that $p$ is a multiple of $4(r-1)$. For any $1\leq k \leq m$, consider
	\begin{equation}\label{eq:lower-bound-P^(k)2}
	\begin{split}
	\bP^{(k)} = & \begin{bmatrix}
	\frac{2-\alpha}{p} \bone_{p\times (p/2)} ~~~~  \frac{\alpha}{p} \bone_{p\times (p/2)} 
	\end{bmatrix} \\
	& + \frac{\eta(2-\alpha)}{2p} \begin{bmatrix}
	\bzero_{(p/2)\times (p/4)} & \bzero_{(p/2)\times (p/4)} & \bzero_{(p/2)\times (p/2)}\\
	~~ \bR^{(k)} ~~ \cdots  ~~~~~ \bR^{(k)} & -\bR^{(k)} ~~ \cdots ~~ -\bR^{(k)} & \bzero_{(p/4)\times (p/2)} \\
	\underbrace{-\bR^{(k)} ~~ \cdots  ~~ - \bR^{(k)}}_{l_0} & \underbrace{~~\bR^{(k)} ~~  \cdots ~~~~~~  \bR^{(k)}}_{l_0} & \bzero_{(p/4)\times (p/2)} \\
	\end{bmatrix},
	\end{split}
	\end{equation}
	where $l_0 = \frac{p}{4(r - 1)}$, $\bR^{(k)} \in \{0, 1\}^{(p/4)\times (r-1)}$, and $\eta$ is some positive value to be determined later. Let 
	\begin{equation}
	\mu := \biggl(\frac{2-\alpha}{p} 1_{p/2}^\top ~~ \frac{\alpha}{p} 1_{p/2}^\top\biggr)^\top.
	\end{equation}
	First of all, regardless of the value of $\bR^{(k)}$, one can see that for any $k \in [m]$, 
	\begin{enumerate}
		\item $\rank(\bP^{(k)})\leq r$; 
		\item $\mu^\top \bP^{(k)} = \mu^\top$, and hence $\mu$ is the invariant distribution of $\bP^{(k)}$; 
		\item $\bP^{(k)} \in \Theta$. 
	\end{enumerate} 
	Let $\{\bR^{(k)}\}_{k = 1}^m$ be i.i.d. matrices of independent Rademacher entries, i.e., for any $k \in [m]$, $\{R^{(k)}_{ij}\}_{i \in [n], j\in [d]}$ are independent Rademacher variables, and $\{\bR^{(k)}\}_{k \in [m]}$ are independent. For any $k \neq l$,  one can see that $\bigl\{\bigl|\bR^{(k)}_{ij} - \bR^{(l)}_{ij}\bigr|\bigr\}$ are i.i.d. uniformly distributed on $\{0, 2\}$, and that 
	$$\mathbb{E}\bigl|\bR^{(k)}_{ij} - \bR^{(l)}_{ij}\bigr| = 1, \quad \Var\bigl(\bigl|\bR^{(k)}_{ij} - \bR^{(l)}_{ij}\bigr|\bigr)= 1,\quad \bigl|\bigl|\bR^{(k)}_{ij} - \bR^{(l)}_{ij}\bigr| - 1\bigr| = 1.$$
	By Bernstein's inequality \citep[][Theorem~2.10]{BLM13}, for any $t>0$, 
	\begin{equation*}
	\bbP\biggl\{\biggl|\bigl\|\bR^{(k)} - \bR^{(l)}\bigr\|_1 - \frac{p(r-1)}{4}\biggr| \ge \biggl( \frac{p(r - 1)t}{2}\biggr)^{\!1/ 2} + t \biggr\} \leq 2e^{-t}. 
	\end{equation*}
	Let $t = p(r - 1)/ 64$ and $m = \lfloor \exp\{p(r-1)/128\} / 2^{1 / 2} \rfloor$. Since $p(r - 1) \ge 192\log 2$, we have that $m \ge 2$. Then a union bound yields that
	\begin{equation*}
	\begin{split}
	\bbP\biggl(\forall 1\leq k < l \leq m, ~ \frac{p(r - 1)}{8} \le \bigl\|\bR^{(k)} - \bR^{(l)}\bigr\|_1 \le \frac{3p(r - 1)}{8}\biggr) \ge 1 - 2m^2\exp\biggl(\frac{-p(r - 1)}{64}\biggr) > 0. 
	\end{split}
	\end{equation*}
	Hence, there exist $\bR^{(1)},\ldots, \bR^{(m)} \subseteq \{-1, 1\}^{(p / 4)\times (r-1)}$ such that
	\begin{equation}\label{ineq:P^k-P^l}
	\forall 1\leq k < l \leq m,~\frac{p(r - 1)}{8} \le \bigl\|\bR^{(k)} - \bR^{(l)}\bigr\|_1 \le \frac{3p(r - 1)}{8}, 
	\end{equation}
	which, given that $\fnorm{\bR^{(j)} - \bR^{(k)}} ^ 2 = 2\|\bR^{(j)} - \bR^{(k)}\|_1$, further implies that 
	\begin{equation}
	\forall 1\leq k < l \leq m,~\frac{p(r - 1)}{4} \le \fnorm{\bR^{(k)} - \bR^{(l)}}^2 \le \frac{3p(r - 1)}{4}. 
	\end{equation}
	%Similarly as the proof for previous part, there exists fixed matrices $\{\bR^{(1)}, \ldots, \bR^{(m)}\}\subseteq \{-1, 1\}^{(p/4)\times (r-1)}$ such that $m = \lfloor\exp(cp(r-1))\rfloor$ and
	%\begin{equation*}
	%	\forall 1\leq k < l \leq m, \quad  p(r-1)/4 \leq \left\|\bR^{(k)} - \bR^{(l)}\right\|_1 \leq 3p(r-1)/4.
	%\end{equation*}
	Now we have that 
	\begin{equation*}
	\begin{split}
	\left\|\bP^{(k)} - \bP^{(l)}\right\|_1 = \frac{2l_0\eta(2-\alpha)}{p}  \|\bR^{(k)} - \bR^{(l)}\|_1 \geq \frac{\eta p (2-\alpha)}{16}, \left\|\bP^{(k)} - \bP^{(l)}\right\|_1 \le \frac{3\eta p (2 - \alpha)}{16},  \\
	\fnorm{\bP^{(k)} - \bP^{(l)}} ^ 2 = \frac{l_0\eta ^ 2(2-\alpha) ^ 2}{p ^ 2}  \fnorm{\bR^{(k)} - \bR^{(l)}} ^ 2 \geq \frac{\eta ^ 2(2-\alpha) ^ 2}{16}, \left\|\bP^{(k)} - \bP^{(l)}\right\|_1 \le \frac{3\eta ^ 2(2-\alpha) ^ 2}{16}. 
	\end{split}
	\end{equation*}
	Besides, 
	\begin{equation*}
	\begin{split}
	D_{\KL} (\cX^{(k)} || \cX^{(l)}) & = n\sum_{i\in [p]} \pi_i D_{\KL}\bigl(\bP_{[i,:]}^{(k)},  \bP^{(l)}_{[i,:]}\bigr)
	=  n\sum_{i=(p/2)+1}^{p} \sum_{j=1}^{p/2} \frac{\alpha}{p} P_{ij}^{(k)}\log \bigl(P_{ij}^{(k)}/P_{ij}^{(l)}\bigr) \\
	& = \frac{2n\alpha}{p}\sum_{i=1}^{p/4} \frac{2-\alpha}{2} D_{\KL}\bigl(u^{(k)}_i, u^{(l)}_i\bigr),
	\end{split}
	\end{equation*}
	where $u^{(k)}_i = \frac{2}{p}1_{p/2} + \frac{\eta}{p}\left[\bR^{(k)}_{[i,:]} ~ \cdots ~ \bR^{(k)}_{[i,:]} ~ -\bR^{(k)}_{[i,:]} ~ \cdots ~ -\bR^{(k)}_{[i,:]}\right]$ corresponds to a $(p/2)$-dimensional distribution. By \citet[][Lemma~4]{ZWa19}, we have that 
	\begin{equation*}
	D_{\KL}\bigl(u^{(k)}_i, u^{(l)}_i\bigr) \leq \frac{3l_0\eta ^ 2}{p}\bigl\|\bR_{[i,:]}^{(k)} - \bR_{[i,:]}^{(l)}\bigr\|_2^2 = \frac{6 l_0\eta^2}{p}\bigl\|\bR_{[i,:]}^{(k)} - \bR_{[i,:]}^{(l)}\bigr\|_1.
	\end{equation*}
	Therefore, 
	\begin{equation*}
	\begin{split}
	D_{\KL}(\cX^{(k)}, \cX^{(l)}) & = \frac{6n\alpha(2-\alpha)l_0 \eta ^ 2}{p ^ 2}\sum_{i=1}^{p/4} \bigl\|\bR_{[i,:]}^{(k)} - \bR_{[i,:]}^{(l)}\bigr\|_1 \leq \frac{12n\alpha l_0\eta^2}{p^2}\|\bR^{(k)} - \bR^{(l)}\|_1 \\
	& \leq \frac{12n\alpha l_0 \eta^2}{p^2} \frac{3p(r-1)}{8} = \frac{9n\eta^2\alpha}{8}.
	\end{split}
	\end{equation*}
	%By Pinsker's inequality and then Fano's inequality \citep[][Lemma~3]{yu1997assouad}, we have that
	By Fano's inequality \citep[][Lemma~3]{yu1997assouad}, we have that
	\begin{equation*}
	\begin{split}
	%		\inf_{\widehat{\bP}} \sup_{\bP\in \cC} D_{\KL}(\bP, \widehat \bP) & \geq \frac{1}{2}\inf_{\widehat{\bP}} \sup_{\bP\in \cC} \bigl\|\widehat{\bP} - \bP\bigr\|^2_1 \geq \inf_{\widehat{\bP}} \sup_{\bP\in \{\bP^{(1)}, \ldots, \bP^{(m)}\}} \bigl\|\hat{\bP} - \bP\bigr\|^2_1\\ 
	\inf_{\widehat{\bP}} \sup_{\bP\in \Theta} \fnorm{\widehat{\bP} - \bP} ^ 2 \geq & \inf_{\widehat{\bP}} \sup_{\bP\in \{\bP^{(1)}, \ldots, \bP^{(m)}\}} \fnorm{\hat{\bP} - \bP} ^ 2 \geq \frac{\eta ^ 2(2-\alpha) ^ 2}{16} \left(1 - \frac{9n\eta^2\alpha - \log 2}{\log m}\right). 
	%		\geq \frac{p}{32}\cdot \sqrt{\frac{cpr}{18np\alpha} \wedge 1}\\
	%		\geq & cp \left(\sqrt{\frac{pr}{n}\cdot \frac{1}{p\theta_{\min}}}\wedge 1\right),
	\end{split}
	\end{equation*}
	There exist universal constants $c_1, c_2 > 0$ such that when $p(r - 1) \ge 192 \log 2$, choosing $\eta = c_1\{p(r - 1) / (n \alpha)\}^{\!1 / 2}$ yields that 
	\[
	\inf_{\widehat{\bP}} \sup_{\bP\in \Theta} \fnorm{\widehat{\bP} - \bP} ^ 2 \ge c_2\frac{p(r - 1)}{n\alpha}. 
	\]
	\quad $\square$
	
\end{proof}

\section{Proof of Theorem \ref{thm:uv}}
\begin{proof}
	~Let $\widehat{\bU}_{\perp}, \widehat{\bV}_{\perp} \in \Re^{p\times (p-r)}$ be the orthogonal complement of $\widehat{\bU}$ and $\widehat{\bV}$. Since $\bU, \bV, \widehat{\bU}$, and $\widehat{\bV}$ are the leading left and right singular vectors of $\bP$ and $\widehat{\bP}$, we have
	\begin{equation*}
	\begin{split}
		\|\widehat{\bP} - \bP\|_F \geq & \|\widehat{\bU}_{\perp}^\top(\widehat{\bP} - \bU\bU^\top \bP)\|_F = \|\widehat{\bU}_{\perp}^\top \bU\bU^\top\bP\|_F \geq \|\widehat{\bU}^\top_{\perp} \bU\|_F \sigma_r(\bU^\top \bP) = \|\sin\Theta(\widehat{\bU}, \bU)\|_F \sigma_r(\bP).
	\end{split}
	\end{equation*}
	Similar argument also applies to $\|\sin\Theta(\widehat{\bV}, \bV)\|$. Thus,
	$$\max\bigl(\|\sin\Theta(\widehat{\bU}, \bU)\|_F, \|\sin\Theta(\widehat{\bV}, \bV)\|_F\bigr)\le \min\biggl(\frac{\|\widehat{\bP}-\bP\|_F}{\sigma_r(\bP)}, r^{1 / 2}\biggr).$$
	The rest of the proof immediately follows from Theorem \ref{thm:nuclear}.
\end{proof}

\section{Proof of Lemma \ref{lem:large_lambda}}

\begin{proof}
	~By the inequality (52) in Lemma 3 in the Appendix of \citet{negahban2012restricted}, we have for any $\bDelta \in \Re^{p \times p}$, 
	\[
		\nnorm{\bP + \bDelta} - \nnorm{\bP} \ge \nnorm{\bDelta_{\overline\cM^\perp}} - \nnorm{\bDelta_{\overline \cM}} - 2\nnorm{\bP_{\cM^\perp}}. 
	\]
	Besides, 
	\[
		\begin{aligned}
		\ell_n(\bP + \bDelta) - \ell_n(\bP) & \ge \inn{\nabla \ell_n(\bP), \bDelta} = \inn{\Pi_{\cN}(\nabla\ell_n(\bP)), \bDelta} \ge - |\inn{\Pi_{\cN}(\nabla \ell_n (\bP)), \bDelta}| \\
		& \ge - \opnorm{\Pi_{\cN}(\nabla \ell_n(\bP))} \nnorm{\bDelta } \ge - \frac{\lambda}{2} \bigl(\nnorm{\bDelta_{\overline \cM}} + \nnorm{\bDelta_{\overline \cM^\perp}} \bigr). 
		\end{aligned}
	\]
	By the optimality of $\widehat\bP$, $\ell_n(\widehat\bP) + \lambda \nnorm{\widehat\bP} \le \ell_n(\bP) + \lambda \nnorm{\bP}$. Therefore, 
	\[
		\lambda \bigl(\nnorm{\bDelta_{\overline \cM}} + 2\nnorm{\bP_{\cM^\perp}} -\nnorm{\bDelta_{\overline\cM^\perp}} \bigr) \ge \lambda (\nnorm{\bP} - \nnorm{\widehat\bP}) \ge - \frac{\lambda}{2}\bigl( \nnorm{\widehat\bDelta_{\overline\cM}} + \nnorm{ \widehat\bDelta_{\overline\cM^\perp}}\bigr), 
	\]
	from which we deduce that  
	\[
	\nnorm{ \widehat\bDelta_{\overline\cM^\perp}} \le 3\nnorm{\widehat\bDelta_{\overline\cM}} + 4\nnorm{\bP_{\cM^\perp}}. 
	\]
\end{proof}

\section{Proof of Lemma \ref{lem:gradient}}

\begin{proof}
		~Some algebra yields that 
	\be
	\label{eq:log_likelihood}
	\nabla \ell_n(\bQ ) = \frac{1}{n} \sum\limits_{i=1}^n -\frac{\bX_i}{\inn{\bQ, \bX_i}}. 
	\ee
	For ease of notation, write $\bZ_i := - \bX_i / \inn{\bP, \bX_i}$. Note that $\bZ_i$ is well-defined almost surely. Besides, 
	\[
	\EE(\bZ_i | \bZ_{i - 1}) = \EE(\bZ_i | \bX_{i-1}) = \sum\limits_{j=1}^p -\frac{e_{X_{i - 1}}e_{j}^\top}{P_{X_{i - 1}, j}} P_{X_{i - 1}, j} = -e_{X_{i - 1}} 1^\top. 
	\]
	Thus $\opnorm{\bZ_i - \EE(\bZ_i | \bZ_{i - 1})} \le p / \alpha + \sqrt{p} =: R < \infty$. Define $\bS_k := \sum_{i=1}^k \bZ_i - \EE(\bZ_i| \bZ_{i - 1})$, then $\{\bS_k\}_{k = 1}^n$ is a matrix martingale. In addition, 
	\[
	\begin{aligned}
	\EE & \bigl\{(\bZ_i -  \EE(\bZ_i | \bZ_{i - 1}))^\top (\bZ_i - \EE (\bZ_i | \bZ_{i - 1})) | \{\bS_k\}_{k = 1}^{i -1}\bigr\}  = \EE \bigl\{(\bZ_i -  \EE(\bZ_i | \bZ_{i - 1}))^\top (\bZ_i - \EE (\bZ_i | \bZ_{i - 1})) | \bZ_{i - 1}\bigr\}  \\
	& = \EE (\bZ^\top_i \bZ_i | \bZ_{i - 1}) - \EE(\bZ_i | \bZ_{i - 1})^\top \EE (\bZ_i | \bZ_{i - 1}) = \biggl(\sum_{j = 1}^p \frac{e_je_j^\top}{P_{X_{i - 1}, j}}\biggr) -  11^\top =: \bW^{(1)}_i, 
	\end{aligned}
	\]
	and similarly, 
	\[
	\begin{aligned}
	\EE  \bigl\{(\bZ_i -  \EE(\bZ_i | \bZ_{i - 1}))(\bZ_i - \EE (\bZ_i | \bZ_{i - 1}))^\top | \{\bS_k\}_{k = 1}^{i -1} \bigr\} =  \biggl(\sum\limits_{j = 1}^ p \frac{ e_{X_{i - 1}}e_{X_{i - 1}}^\top}{P_{X_{i - 1}, j}} \biggr) -  p e_{X_{i - 1}}e_{X_{i - 1}}^\top =: \bW^{(2)}_i. 	
	\end{aligned}
	\]
	Write $\opnorm{\sum_{i=1}^n \bW^{(1)}_i}$ as $W_n^{(1)}$, $\opnorm{\sum_{i = 1}^n \bW^{(2)}_i}$ as $W_n^{(2)}$ and $\max(W_n^{(1)}, W_n^{(2)})$ as $W_n$. By the matrix Freedman inequality \citep[Corollary~1.3]{tropp2011freedman}, for any $t \ge 0$ and $\sigma^2 > 0$, 
	\be
	\label{eq:matrix_freedman}
	\PP ( \opnorm{\bS_n} \ge t, W_n\le \sigma^2 ) \le 2p \exp\biggl(-\frac{t^2 /2 }{\sigma^2 + Rt / 3}\biggr). 
	\ee
	Now we need to choose an appropriate $\sigma^2$ so that $W_n \le \sigma^2$ holds with high probability. Note that $W^{(1)}_n \le np(\alpha^{-1} + 1)$ and $W_n^{(2)} \le (p^2 \alpha^{-1} - p) \sup_{j\in [p]} \sum_{i = 1}^n \ind_{\{X_i = s_j\}}$. In the following we derive a bound for $\sup_{j\in [p]} \sum_{i = 1}^n \ind_{\{X_i = s_j\}}$. For any $j \in [p]$, by \citet[Theorem~1.2]{JFS18}, which is a variant of Bernstein's inequality for Markov chains, we have that 
	\be
	\label{eq:mc_bernstein}
	\PP\biggl\{\frac{1}{n}\sum_{i=1}^n (\ind_{\{X_i = s_j\}} - \pi_j) > \epsilon \biggr\} \le \exp\biggl( -\frac{n\epsilon ^ 2}{2(A_1\beta / p + A_2\epsilon)}\biggr), 
	\ee
	where 
	\[
	A_1 = \frac{1 + \max(\rho_+, 0)}{1 - \max(\rho_+, 0)}~~~\textnormal{and}~~~A_2 = \frac{1}{3}\ind_{\{\rho_{+} \le 0\}} +  \frac{5}{1 - \rho_+} \ind_{\{\rho_+ > 0\}}. 
	\]
	Some algebra yields that for any $\xi > 0$, 
	\[
	\PP\biggl\{\frac{1}{n}\sum_{i=1}^n \ind_{\{X_i = s_j\}} - \pi_j > \biggl(\frac{4A_1\xi}{np}\biggr)^{1 / 2} + \frac{4A_2\xi}{n} \biggr\} \le \exp( - \xi). 
	\]
	A union bound over $j \in [p]$ yields that 
	\[
	\PP\biggl\{ \sup_{j \in [p]} \frac{1}{n}\sum_{i=1}^n (\ind_{\{X_i = s_j\}}- \pi_j) > \biggl(\frac{4A_1\xi \log p}{np}\biggr)^{1 / 2} + \frac{4A_2\xi \log p}{n}\biggr\} \le p^{-(\xi - 1)}, 
	\]
	which implies that 
	\[
	\PP\biggl\{\sup_{j \in [p]} \frac{1}{n}\sum\limits_{i=1}^n  \ind_{\{X_i = s_j\}} > \pi_{\max} + \biggl(\frac{4A_1\xi \log p}{np}\biggr)^{1 / 2} + \frac{4A_2\xi \log p}{n}\biggr\} \le p^{-(\xi - 1)}. 
	\]
	Therefore, whenever $n \pi_{\max} (1 - \rho_+) \ge 2 \log p$, we have that 
	\[
		\PP\biggl(\sup_{j \in [p]} \frac{1}{n} \sum_{i = 1}^n 1_{\{X_i = s_j\}} \gtrsim \pi_{\max}\biggr) \le \exp\biggl(- \frac{n\pi_{\max} (1 - \rho_+)}{2}\biggr). 
	\]
%	Given that $n \ge cp \log p$ for some universal constant $c > 0$, we can find $C_1 > 0$ that depends on $\alpha, \beta$ such that for any $\xi > 1$, 
%	\be
%	\label{eq:unionbd}
%	\PP\biggl(\sup_{j \in [p]} \frac{1}{n}\sum\limits_{i=1}^n  \ind_{\{X_i = s_j\}} > \frac{C_1 \xi}{p(1 - \rho_+)} \biggr) \le p^{-(\xi - 1)}, 
%	\ee
%	which further implies that there exists a $C_2>0$ depending on $\alpha, \beta$ such that a
	Combining this with the bounds of $W_n^{(1)}$ and $W_n^{(2)}$, we have that 
	\[
		\PP\biggl(W_n \ge \frac{C_1np ^ 2\pi_{\max}}{\alpha}\biggr) \le \exp\biggl(- \frac{n\pi_{\max} (1 - \rho_+)}{2}\biggr), 
	\]
	where $C_1$ is a universal constant. Now choosing $\sigma^2 = C_1 np^2 \pi_{\max} / \alpha$, we deduce that for any $t \ge 0$, 
	\[
	\begin{aligned}
		\PP(\opnorm{\bS_n} \ge t) & = \PP(\opnorm{\bS_n} \ge t, W_n \le \sigma^2 ) + \PP(\opnorm{\bS_n} \ge t, W_n > \sigma^2 ) \\
		& \le \PP(\opnorm{\bS_n} \ge t, W_n \le \sigma^2 ) + \PP(W_n > \sigma^2 ) \\
		& \le 2p \exp\biggl(-\frac{t^2 /2 }{\sigma^2 + Rt / 3}\biggr) + \exp\biggl(- \frac{n\pi_{\max} (1 - \rho_+)}{2}\biggr). 
	\end{aligned}
	\]
	Equivalently, for any $\xi > 1$, 
	\[
		\PP\biggl\{\biggl\|\frac{1}{n}\bS_n\biggr\|_2 \gtrsim \biggl(\frac{\xi p ^ 2\pi_{\max} \log p}{n\alpha}\biggr)^{1 / 2} + \frac{\xi p\log p}{n \alpha}\biggr\} \le 4p^{-(\xi - 1)} + \exp\biggl(- \frac{n\pi_{\max} (1 - \rho_+)}{2}\biggr). 
	\]	
	Finally, observe that for any $i \in [n]$, $\Pi_{\cN}(\EE(\bZ_i | \bZ_{i - 1})) = \Pi_{\cN}( -e_{X_{i - 1}} 1^\top) = \bzero$. Therefore, $\Pi_{\cN}(\nabla\ell_n(\bP)) = n^{-1}\bS_n$ and the final conclusion then follows. 	
\end{proof}

\section{Proof of Lemma \ref{lem:uniform_law}}
\begin{proof}
	\noindent We first split $\mathcal{C}(\eta)$ as the union of the sets 
	\begin{equation*}
	\mathcal{C}_l := \left\{\bQ \in \mathcal{C}(\eta): 2^{l-1}\eta \leq \DKL(\bP, \bQ) \leq 2^l\eta, ~ \rank(\bQ)\leq r\right\}, \quad l=1,2,3,\ldots. 
	\end{equation*}
	Define
	\begin{equation*}
	\begin{split}
	\gamma_l = & \sup_{\bQ \in \mathcal{C}_l} \bigl|\DKL(\bP, \bQ) - \widetilde D_{\KL}(\bP, \bQ)\bigr| \\
	= & \sup_{\bQ \in \mathcal{C}_l} \biggl|\frac{1}{n}\sum_{i=1}^n \langle \log(\bP) - \log(\bQ), \bX_i \rangle - \mathbb{E}\langle \log(\bP) - \log(\bQ), \bX_i \rangle\biggr|.
	\end{split}
	\end{equation*}
	First, we wish to apply \citet[][Theorem~7]{Ada08} to bound $|\gamma_l - \EE \gamma_l|$. Adamczak's bound entails the following asymptotic weak variance
	\[
	\sigma^2 := \sup_{\bQ\in\cC_l} \Var \biggl\{\sum_{i = S_1 + 1}^{S_2} \langle \log(\bP) - \log(\bQ), \bX_i \rangle - \mathbb{E}\langle \log(\bP) - \log(\bQ), \bX_i \rangle \biggr\} / \EE T_2. 
	\]
	We have that 
	\[
	\begin{aligned}
	\sigma^2  & \le \sup_{\bQ\in\cC_l}\EE \biggl[\biggl\{ \sum_{i = S_1 + 1}^{S_2} \langle \log(\bP) - \log(\bQ), \bX_i \rangle - \mathbb{E}\langle \log(\bP) - \log(\bQ), \bX_i \rangle \biggr\}^{\!2}\biggr] / \EE T_2\\
	& = \frac{1}{2} \sup_{\bQ\in\cC_l} \sum\limits_{j = 1}^{\infty} \EE \biggl[\biggl\{ \sum_{i = S_1 + 1}^{S_2} \langle \log(\bP) - \log(\bQ), \bX_i \rangle - \mathbb{E}\langle \log(\bP) - \log(\bQ), \bX_i \rangle \biggr\}^{\!2} \ind_{\{T_2 = j\}}\biggr] \\
	& \le \frac{1}{2}  \sum\limits_{j = 1}^{\infty}  {4j^2\log^ 2(\beta / \alpha) }\PP(T_2 = j) = {2 \log^2(\beta / \alpha) \EE (T_2^2)} = {8\log ^ 2(\beta / \alpha)}. 
	\end{aligned}
	\]
	By \citet[][Theorem~7]{Ada08}, there exists a universal constant $K > 1$ such that for any $\xi > 0$, 
	\[
	\PP\biggl\{ |\gamma_l - \EE \gamma_l| \ge K \EE\gamma_l + {2\log(\beta / \alpha)}\biggl(\frac{2K\xi }{n}\biggr)^{\! 1 / 2} +  \frac{16 K\log(\beta / \alpha)\xi \log n }{n}\biggr\} \le K e^{- \xi}.  
	\]
	Since $n^{- 1/ 2} \ge 2 n^{-1}\log n$ for any positive integer $n$, we have that
	\be
	\label{eq:gamma_l_uniform_concentration}
	\PP\biggl\{ |\gamma_l - \EE \gamma_l| \ge K \EE\gamma_l + {11K\log(\beta / \alpha)}\biggl(\frac{\xi }{n}\biggr)^{\! 1 / 2}\biggr\} \le K e^{- \xi}.  	
	\ee
	Next, we bound $\EE \gamma_l$. Let $\{\varepsilon_i\}_{i=1}^n$ be $n$ independent Rademacher random variables. By a symmetrization argument, 
	\begin{equation*}
	\begin{split}
	\mathbb{E}\gamma_l = & \mathbb{E}\biggl(\sup_{\bQ \in \mathcal{C}_l} \biggl|\frac{1}{n}\sum_{i=1}^n \langle \log (\bP) - \log(\bQ), \bX_i\rangle - \mathbb{E}\langle \log(\bP) - \log(\bQ), \bX_i \rangle\biggr|\biggr)\\
	\leq & 2\mathbb{E}\biggl(\sup_{\bQ \in \mathcal{C}_l}\biggl|\frac{1}{n}\sum_{i=1}^n \varepsilon_k \langle \log(\bP) - \log(\bQ), \bX_i\rangle\biggr|\biggr).
	\end{split}
	\end{equation*}
	Let $\phi_i(t) = (\alpha / p)\{\log(\langle\bP, \bX_i\rangle + t) - \log(\langle \bP, \bX_i\rangle)\}$. Then $\phi_{i}(0) = 0$ and $|\phi_i'(t)| \leq 1$ for all $t$ such that $t + \langle\bP, \bX_i\rangle \geq \alpha/p$. In other words, $\phi_i$ is a contraction map for $t \geq \min_{j, k \in [p]}(P_{jk} - \alpha/p)$. By the contraction principle (Theorem 4.12 in \cite{ledoux2013probability}),
	%	\begin{equation}\label{ineq:expected-gamma_l}
	%	\begin{split}
	%	\mathbb{E}(\gamma_l) \leq & \frac{2p}{\alpha}\mathbb{E}\biggl(\sup_{\bQ \in \mathcal{C}_l}\biggl|\frac{1}{n}\sum_{i=1}^n\varepsilon_i \phi_i\left(\langle \bQ-\bP, \bX_i\rangle\right)\biggr|\biggr) \leq  \frac{4p}{\alpha}\mathbb{E}\biggl(\sup_{\bQ \in \mathcal{C}_l}\biggl|\frac{1}{n}\sum_{i = 1}^n \varepsilon_k\langle \bQ-\bP, \bX_i\rangle\biggr|\biggr)\\
	%	\leq & \frac{4p}{\alpha}\mathbb{E}\biggl(\sup_{\bQ \in \mathcal{C}_l}\biggl\|\frac{1}{n}\sum_{i = 1}^n \varepsilon_k\bD^{- 1/ 2}_{\pi}\bX_i\biggr\| \|\bD^{1 / 2}_\pi(\bQ-\bP)\|_\ast\biggr) \\
	%	\leq & \frac{4p}{\alpha} \mathbb{E}\biggl\|\frac{1}{n}\sum_{i = 1}^n \varepsilon_i\bD^{-1 / 2}_{\pi}\bX_i\biggr\|  \sup_{\bQ \in \mathcal{C}_l} \|\bD^{1 / 2}_{\pi}(\bQ - \bP)\|_\ast. 
	%	\end{split}
	%	\end{equation}
	\begin{equation}\label{ineq:expected-gamma_l}
		\begin{split}
			\mathbb{E}\gamma_l \leq & \frac{2p}{\alpha}\mathbb{E}\biggl(\sup_{\bQ \in \mathcal{C}_l}\biggl|\frac{1}{n}\sum_{i=1}^n\varepsilon_i \phi_i\left(\langle \bQ-\bP, \bX_i\rangle\right)\biggr|\biggr) \leq  \frac{4p}{\alpha}\mathbb{E}\biggl(\sup_{\bQ \in \mathcal{C}_l}\biggl|\frac{1}{n}\sum_{i = 1}^n \varepsilon_i\langle \bQ-\bP, \bX_i\rangle\biggr|\biggr)\\
			\leq & \frac{4p}{\alpha}\mathbb{E}\biggl(\sup_{\bQ \in \mathcal{C}_l}\biggl\|\frac{1}{n}\sum_{i = 1}^n \varepsilon_i\bX_i\biggr\|_2 \|\bQ-\bP\|_\ast\biggr) \leq  \frac{4p}{\alpha} \mathbb{E}\biggl\|\frac{1}{n}\sum_{i = 1}^n \varepsilon_i\bX_i\biggr\|_2  \sup_{\bQ \in \mathcal{C}_l} \|\bQ - \bP\|_\ast. 
		\end{split}
	\end{equation}	
	By Lemma \ref{lem:kl_to_l2}, 
	%	\begin{equation}\label{ineq:Q-P-nuclear}
	%	\begin{split}
	%	 \sup_{\bQ \in \mathcal{C}_l} & \|\bD^{1 / 2}_{\pi}(\bQ - \bP)\|_\ast  \leq \sup_{\bQ \in \mathcal{C}_l} (2r)^{1 / 2} \|\bD^{1 / 2}_{\pi}(\bQ - \bP)\|_F \leq \biggl( 2r\sum_{j = 1}^p \pi_j \| Q_{j\cdot} - P_{j\cdot}\|_2^2\biggr) ^ {\!1 / 2} \\
	%	& \le \biggl( 2r\sum_{j = 1}^p \pi_j\| Q_{j\cdot} - P_{j\cdot}\|_1^2\biggr) ^ {\!1 / 2} \le \biggl(4r \sum_{j = 1}^p \pi_j D_{\KL}(P_{j\cdot}, Q_{j\cdot})\biggr)^{\!1 / 2} = 2(rD_{\KL} (\bP, \bQ))^{1 / 2}, 
	%	\end{split}
	%	\end{equation}
	\begin{equation}
	\label{ineq:Q-P-nuclear}
	\begin{split}
		\sup_{\bQ \in \mathcal{C}_l}  \|\bQ - \bP\|_\ast & \leq \sup_{\bQ \in \mathcal{C}_l} (2r)^{1 / 2} \fnorm{\bQ - \bP}  \le 2\beta \biggl(\frac{2^l \eta r}{p\alpha \pi_{\min}}\biggr)^{\!1 / 2}. 
	\end{split}
	\end{equation}	
	%	where the last inequality uses Pinsker's inequality. 
	Hence, the remaining task is to bound $\mathbb{E}\|n^{-1}\sum_{i = 1}^n \varepsilon_i\bX_i\|$. From now on, we denote $\varepsilon_i \bX_i$ by $\bZ_i$. One can see that $(\bZ_i)_{i = 1}^n$ is a martingale difference sequence. We wish to apply the matrix Freedman inequality \citep[Corollary~1.3]{tropp2011freedman} to bound the average of $(\bZ_i)_{i = 1}^n$. We have that
	\begin{equation*}
	\begin{split}
	\biggl\|\sum_{i = 1}^n \mathbb{E} \bigl(\bZ_i^\top \bZ_i \vert X_{i - 1}\bigr) \biggr\|_2 = & \biggl\|\sum_{i=1}^n \sum_{j=1}^p  P_{X_{i - 1}, j} (e_{X_{i - 1}} e_j^\top)^\top (e_{X_{i - 1}}e_j^\top)\biggr\|_2 = \biggl\|\sum_{j=1}^p  \sum_{i=1}^n P_{X_{i - 1}, j} e_je_j^\top\biggr\|_2 \\
	= & \max_{j \in [p]} \sum_{i = 1}^n P_{X_{i - 1}, j} =: W_n^{(1)}
	\end{split}
	\end{equation*} 
	and that 
	\begin{equation*}
	\begin{split}
	\biggl\|\sum_{i = 1}^n \mathbb{E} \bigl(\bZ_i\bZ_i^\top\vert X_{i - 1}\bigr) \biggr\|_2 = & \biggl\|\sum_{i=1}^n \sum_{j=1}^p P_{X_{i - 1}, j} e_{X_{i - 1}} e_{X_{i - 1}}^\top \biggr\|_2 = \biggl\|\sum_{i=1}^n e_{X_{i - 1}} e_{X_{i - 1}}^\top\biggr\|_2 \\
	= & \max_{j \in [p]} \sum_{i = 1}^n 1_{\{X_{i - 1} = j\}} =: W_n^{(2)}. 
	\end{split}
	\end{equation*}	
	We first bound $W_n^{(1)}$. Note that for any $j \in [p]$, $\EE (P_{X_{i - 1}, j}) = \pi_j $, and that
	\[
	\Var_{\pi}(P_{X_{i - 1}, j}) = \sum_{k = 1}^p \pi_k (P_{kj} - \pi_j)^2 = \sum_{k = 1}^p \pi_k P^2_{kj} - \pi_j ^ 2 \le \pi_j(1 - \pi_j). 
	\]
	By a variant of Bernstein's inequality for Markov chains \citep[Theorem~1.2]{JFS18}, we have that for any $j \in [p]$, 
	\[
	\label{eq:mc_bernstein}
	\PP\biggl(\frac{1}{n} \sum_{i = 1}^n P_{X_{i - 1}, j} - \pi_j > \epsilon \biggr) \le \exp\biggl\{ -\frac{n\epsilon ^ 2}{2(A_1\pi_j + A_2\epsilon)}\biggr\}, 
	\]
	where
	\[
	A_1 := \frac{1 + \max(\rho_+, 0)}{1 - \max(\rho_+, 0)}~~~\textnormal{and}~~~A_2 := \frac{1}{3}\ind_{\{\rho_{+} \le 0\}} +  \frac{5}{1 - \rho_+} \ind_{\{\rho_+ > 0\}}. 
	\]
	A union bound yields that 
	\be
	\label{eq:w1}
	\PP\bigl\{W_n^{(1)} \ge n\pi_{\max} + (4nA_1 \pi_{\max}\xi \log p)^{1 / 2} + {4A_2\xi \log p} \bigr\} \le p^{-(\xi - 1)}. 
	\ee
	Next we bound $W_n^{(2)}$. Note that $W_n^{(2)} \le \max_{j \in [p]} \sum_{i =1}^n \ind_{\{X_{i - 1} = s_j\}}$. Similarly, by \citet[Theorem~1.2]{JFS18}, for any $j \in [p]$, 
	%	and $W_n^{(2)} \le (p^2 \alpha^{-1} - p) \sup_{j\in [p]} \sum_{i = 1}^n \ind_{\{X_i = s_j\}}$.
	%	 In the following we derive a bound for $\sup_{j\in [p]} \sum_{i = 1}^n \ind_{\{X_i = s_j\}}$. For any $j \in [p]$, by \citet[Theorem~1.2]{JFS18}, which is a variant of Bernstein's inequality for Markov chains, we have that 
	\be
	\label{eq:mc_bernstein}
	\PP\biggl\{\frac{1}{n}\sum_{i=1}^n \ind_{\{X_{i - 1} = s_j\}} - \pi_j > \epsilon \biggr\} \le \exp\biggl\{ -\frac{n\epsilon ^ 2}{2(A_1\pi_j + A_2\epsilon)}\biggr\}, 
	\ee
	Some algebra yields that for any $\xi > 0$, 
	\[
	\PP\biggl\{\frac{1}{n}\sum_{i=1}^n \ind_{\{X_{i - 1} = s_j\}} - \pi_j > \biggl(\frac{4A_1\pi_j\xi }{n}\biggr)^{1 / 2} + \frac{4A_2\xi}{n} \biggr\} \le \exp( - \xi). 
	\]
	By a union bound over $j \in [p]$, 
	\[
	\PP\biggl\{ \max_{j \in [p]} \frac{1}{n}\sum_{i=1}^n \ind_{\{X_{i - 1} = s_j\}} > \pi_{\max} + \biggl(\frac{4A_1\pi_{\max}\xi \log p}{n}\biggr)^{1 / 2} + \frac{4A_2\xi \log p}{n}\biggr\} \le p^{-(\xi - 1)}, 
	\]
	which further implies that 
	\be
	\label{eq:w2}
	\PP\bigl\{ W_n^{(2)} \ge n\pi_{\max} + (4nA_1 \pi_{\max}\xi \log p)^{1 / 2} + {4A_2\xi \log p} \bigr\} \le p^{-(\xi - 1)}. 
	\ee
	Define $W_n := \max(W_n^{(1)}, W_n^{(2)})$. Let $\bS_n :=  \sum_{i = 1}^n \varepsilon_i\bX_i$. By matrix Freedman's inequality \citep[][Corollary~1.3]{tropp2011freedman}, for any $t \ge 0$ and $\sigma^2 > 0$, 
	\be
	\label{eq:matrix_freedman}
	\PP ( \opnorm{\bS_n} \ge t, W_n\le \sigma^2 ) \le 2p \exp\biggl(-\frac{t^2 /2 }{\sigma^2 + t / 3}\biggr). 
	\ee
	Now we need to choose an appropriate $\sigma^2$ so that $W_n \le \sigma^2$ holds with high probability.
	Given that $\rho_+ > 0$ and $n\pi_{\max} \ge 10\xi \log p / (1 - \rho_+)$, combining \eqref{eq:w1} and \eqref{eq:w2} yields that 
	%	\be
	%		\label{eq:w}	
	%		\PP\biggl\{ W_n \ge n\pi_{\max} + \biggl(\frac{8n\pi_{\max}\xi \log p}{1 - \rho_+}\biggr)^{1 / 2} + \frac{20\xi \log p}{1 - \rho_+} \biggr\} \le 2p^{-(\xi - 1)}. 
	%	\ee
	\be
	\label{eq:w}	
	\PP\bigl( W_n \ge 4n\pi_{\max}\bigr) \le 2p^{-(\xi - 1)}. 
	\ee
	Now choosing $\sigma^2 = 4n\pi_{\max}$ in \eqref{eq:matrix_freedman}, we deduce that 
	\[
	\begin{aligned}
	\PP(\opnorm{\bS_n} \ge t) & = \PP(\opnorm{\bS_n} \ge t, W_n \le \sigma^2 ) + \PP(\opnorm{\bS_n} \ge t, W_n > \sigma^2 ) \\
	& \le \PP(\opnorm{\bS_n} \ge t, W_n \le \sigma^2 ) + \PP(W_n > \sigma^2 ) \\
	& \le 2p \exp\biggl(-\frac{t^2 /2 }{\sigma^2 + t / 3}\biggr) + 2p^{-(\xi - 1)}. 
	\end{aligned}
	\]
	Choose $\xi = n \pi_{\max} (1 - \rho_+) / (10 \log p)$. As long as $n\pi_{\max}(1 - \rho_+)\ge \max (20 \log p, \log n)$, we have that 
	\be
	\label{eq:expectation_sn_op}
	\EE \biggl\|\frac{1}{n}\bS_n\biggr\|_2 \lesssim \biggl(\frac{\pi_{\max}\log p}{n}\biggr)^{\!1 /2}. 		
	%		\EE \biggl\|\frac{1}{n}\bS_n\biggr\|_2 \lesssim \biggl(\frac{\pi_{\max}\log p}{n}\biggr)^{\!1 /2} + \exp\{-5n\pi_{\max}(1 - \rho_+)\}. 
	%		% Ziwei: log p / n << the first term on the RHS		
	\ee
	Combining \eqref{ineq:expected-gamma_l}, \eqref{ineq:Q-P-nuclear} and \eqref{eq:expectation_sn_op} yields that
	\begin{equation*}
	\mathbb{E} \gamma_l \lesssim \frac{\beta}{\alpha ^ {3 / 2}}\biggl(\frac{2^{l}\eta \pi_{\max}rp\log p}{\pi_{\min}n}\biggr)^{\!1 /2}. 
	\end{equation*}
	Then combining this with \eqref{eq:gamma_l_uniform_concentration} yields that   
	\begin{equation*}
	\PP\biggl\{\gamma_l \gtrsim \frac{\beta}{\alpha ^ {3 / 2}}\biggl(\frac{2^{l}\eta \pi_{\max}rp\log p}{\pi_{\min}n}\biggr)^{\!1 /2} + \log(\beta / \alpha)\biggl(\frac{\xi }{n}\biggr)^{\! 1 / 2} \hspace{0cm} \biggr\} \lesssim e^{-\xi}. 
	\end{equation*}
	Let $\xi = 2^l \eta \pi_{\max} rp \log p / \pi_{\min}$. Then there exist universal constants $C_1, C_2 > 0$ such that
	\begin{equation*}
	\PP\biggl\{\gamma_l \ge \frac{C_1\beta}{\alpha ^ {3 / 2}}\biggl(\frac{2^{l}\eta \pi_{\max}rp\log p}{\pi_{\min}n}\biggr)^{\!1 /2} \biggr\} \le C_2\exp\biggl\{- \frac{(2l + 1) \eta \pi_{\max} rp \log p}{\pi_{\min}}\biggr\}. 	
	\end{equation*}
	We can thus deduce that there exists a universal constant $C_3 > 0$ such that 
	\begin{equation*}
	\begin{split}
	&\PP\biggl( |\widetilde{D}_{\KL}(\bP, \bQ) - \DKL(\bP, \bQ)|> \frac{1}{2}\DKL(\bP, \bQ) + \frac{C_3\pi_{\max} \beta ^ 2rp\log p}{\pi_{\min}\alpha ^ 3n} \biggr)\\
	\leq & \sum_{l=0}^\infty P\left(\exists \bQ \in \mathcal{C}_l, ~\left|\widetilde{D}_{\mathrm{KL}}(\bP, \bQ) - \DKL(\bP, \bQ)\right|> 2^{l - 2}\eta + \frac{C_3\pi_{\max} \beta ^ 2rp\log p}{\pi_{\min}\alpha ^ 3n}\right)\\
	\leq & \sum_{l=0}^\infty P\biggl\{\gamma_l \ge \frac{C_1\beta}{\alpha ^ {3 / 2}}\biggl(\frac{2^l\eta\pi_{\max}rp\log p}{\pi_{\min}n}\biggr)^{\!1 /2}\biggr\}\\
	\leq & C_2 \sum_{l=0}^\infty \exp\biggl\{- \frac{(2l + 1) \eta \pi_{\max} rp \log p}{\pi_{\min}}\biggr\} \le 2C_2 \exp\biggl(- \frac{\eta \pi_{\max} rp \log p}{\pi_{\min}}\biggr). 
	\end{split}
	\end{equation*}
	where we use the Cauchy-Schwarz inequality in the second step. 
\end{proof}

\section{Alternative statistical error analysis}
\label{sec:alt}

\subsection{Main results}

In this section, we provide an alternative proof strategy that follows \citet{NRW12} to bound the statistical error of $\widehat \bP$ and $\widehat \bP ^ r$. This strategy resolves the inconsistency issue of Theorems \ref{thm:nuclear} and \ref{thm:rank} when $n \gg \{rp \pi_{\max}(\log p) \beta/ (\pi_{\min} \alpha ^ {\!3 / 2})\} ^ 2$. For any $R>0$, define a constraint set $\cC (\beta, R, \kappa):= \{\bDelta \in \Re^{p \times p}: \supnorm{\bDelta} \le \beta / p, \fnorm{\bDelta}\le R, \nnorm{\bDelta} \le \kappa r^{1 / 2}\fnorm{\bDelta}  \}$. An important ingredient of this statistical analysis is the localized restricted strong convexity \citep{NWa11, FLS18} of the loss function $\ell_n(\bP)$ near $\bP$. This property allows us to bound the distance in the parameter space by the difference in the objective function value. Define the first-order Taylor remainder term of the negative log-likelihood function $\ell_n(\bQ)$ around $\bP$ as
\[
	\delta\ell_n(\bQ; \bP)  := \ell_n(\bQ) - \ell_n(\bP) - \nabla\ell_n(\bP)^\top (\bQ - \bP). 
\]
The following lemma establishes the desired local restricted strong convexity. 
\begin{lemma}
	\label{lem:restricted_strong_convexity}
	Under Assumption \ref{asp:1}, there exists a universal constant $K$ such that for any $\xi > 1$, it holds with probability at least $1 - K\exp(-\xi)$ that for any $\bDelta \in \cC(\beta, R, \kappa)$, 
	\be
	\begin{aligned}
	\delta\ell_n(\bP + \bDelta; \bP) \ge \frac{\alpha ^2}{8\beta^2}\fnorm{\bDelta}^2 & - 8R\biggl({\frac{3K\xi }{n}}\biggr)^{1 / 2} - \frac{8K\xi\alpha^2 \log n }{\beta^2 n} - \frac{Kp\kappa R}{\beta}\biggl(\frac{r\pi_{\max}\log p}{n}\biggr) ^ {\!1 / 2}. 
\end{aligned}
	\ee
\end{lemma}

Now we present the statistical rates of $\widehat \bP$ and $\widehat \bP ^ r$. 

\begin{theorem}[Alternative statistical guarantee for $\widehat \bP$]
	\label{thm:nuclear_alt}
	Under the same assumptions of Theorem \ref{thm:nuclear}, there exists a universal constant $C_1 > 0$, such that for any $\xi > 1$, if we choose 
	\[
	\lambda =  C_1 \biggl\{\biggl(\frac{\xi p ^ 2\pi_{\max} \log p}{n\alpha}\biggr)^{\! 1 / 2} + \frac{\xi p\log p}{n \alpha}\biggr\}, 
	\]
	then whenever $n\pi_{\max}(1 - \rho_+) \ge \max\{\max(20, \xi ^ 2) \log p, \log n\}$, we have with probability at least $1 - K\exp(-\xi) - 4p^{-(\xi - 1)} - p ^{-1}$ that 
	\[
	\begin{aligned}
	\fnorm{\widehat\bP - \bP} \lesssim \frac{\beta ^ 2}{\alpha ^ 2}\biggl(\frac{\xi r p ^ 2 \pi_{\max} \log p}{n \alpha}\biggr) ^ {\!1 / 2}~~\text{and}~~D_{\mathrm{KL}}(\bP, \widehat\bP) \lesssim \frac{\xi \beta ^ 6 \pi_{\max} r p ^ 2  \log p}{n \alpha ^ 7},
	\end{aligned}
	\]
	where $K$ is the same constant as in Lemma \ref{lem:restricted_strong_convexity}. 
\end{theorem}

\begin{theorem}[Alternative statistical guarantee for $\widehat \bP ^ r$]
	\label{thm:rank_alt}
	Under the same assumptions of Theorem \ref{thm:nuclear}, there exists a universal constant $C_1 > 0$, for any $\xi > 1$, we have with probability at least $1 - K\exp(-\xi) - 4p^{-(\xi - 1)} - p ^{-1}$ that 
	\[
	\fnorm{\widehat\bP ^ r- \bP} \lesssim \frac{\beta ^ 2}{\alpha ^ 2}\biggl(\frac{\xi r p ^ 2 \pi_{\max} \log p}{n \alpha}\biggr) ^ {\!1 / 2} ~~\text{and}~~D_{\mathrm{KL}}(\bP, \widehat\bP ^ r) \lesssim \frac{\xi \beta ^ 6 \pi_{\max} r p ^ 2  \log p}{n \alpha ^ 7}, 
	\]
	where $K$ is the same constant as in Lemma \ref{lem:restricted_strong_convexity}. 
\end{theorem}

One can see from the theorems above that the derived error bounds converge to zero as $n$ goes to infinity. Nevertheless, their dependence on $\alpha$ and $\beta$ is worse than those in Theorems \ref{thm:nuclear} and \ref{thm:rank} when $n \lesssim \{rp \pi_{\max}(\log p) \beta/ (\pi_{\min} \alpha ^ {\!3 / 2})\} ^ 2$. This is why we do not present this result in the main text. 

\subsection{Proof of Lemma \ref{lem:restricted_strong_convexity}}

\begin{proof}
	~Given any $\bDelta \in \cC(\beta, R, \kappa)$, it holds that for some $0\le v \le 1$ that 
	\be
	\label{eq:rsc_first_step}
	\begin{aligned}
		\delta \ell_n(\bP + \bDelta; \bP) & = \frac{1}{2} \vec(\bDelta)^\top \bH_n(\bP + v\bDelta)\vec(\bDelta) = \frac{1}{2n} \sum\limits_{i=1}^n \frac{\inn{\bX_i, \bDelta}^2}{\inn{\bP + v\bDelta, \bX_i}^2} \\
		& \ge \frac{1}{2n}\sum\limits_{i=1}^n \frac{p ^ 2}{4\beta^2} \inn{\bDelta, \bX_i}^2. 
	\end{aligned}
	\ee
	Define 
	\[
	\Gamma_n := \sup_{\bDelta \in \cC(\beta, R, \kappa)} \biggl |\frac{1}{n} \sum\limits_{i=1}^n \inn{\bDelta, \bX_i}^2 - \EE(\inn{\bDelta, \bX_i}^2) \biggr |. 
	\]
	We first bound the deviation of $\Gamma_r$ from its expectation $\EE \Gamma_r$. Note that $\{\bX_i\}_{i = 1}^n$ is a Markov chain on $\cM := \{e_je_k^\top\}_{j, k = 1}^p$. Here we apply a tail inequality for suprema of unbounded empirical processes due to \citet[][Theorem~7]{Ada08}. To apply this result, we need to verify that $\{\bX_i\}_{i = 1}^n$ satisfies the ``minorization condition'' as stated in Section 3.1 of \citet{Ada08}. Below we characterize a specialized version of this condition. 
	
	\begin{condition}[Condition 1 (minorized condition).]
		\label{con:minorization}
		We say that a Markov chain $\cX$ on $\cS$ satisfies the minorized condition if there exist $\delta > 0$, a set $\cC\subset \cS$ and a probability measure $\nu$ on $\cS$ for which $\forall_{x \in \cC} \forall_{\cA \subset \cS} \PP(x, \cA) \ge \delta \nu(\cA)$ and $\forall_{x \in \cS} \exists_{n \in \NN} \PP^n(x, \cC) > 0$. 
	\end{condition} 
	One can verify that the Markov chain $\{\bX_i\}_{i = 1}^n$ satisfies Condition \ref{con:minorization} with $\delta = 1/ 2$, $\cC = \{e_1e_2^\top\}$ and $\nu(e_je_k^\top) = P_{jk}1_{\{j = 2\}}$ for $j, k \in [p]$.
	
	Now consider a new Markov chain $\{(\widetilde \bX_i, R_i)\}_{i=1}^n$ constructed as follows. Let $\{R_i\}_{i=1}^n$ be i.i.d. Bernoulli random variables with $\EE R_1 = \delta$. For any $i \in \{0, \ldots, n - 1\}$, at step $i$, if $\bX_i \notin \cC$, we sample $\widetilde \bX_{i + 1}$ according to $\PP(\widetilde \bX_i, \cdot)$; otherwise, the distribution of $\widetilde \bX_i$ depends on $R_i$: if $R_i = 1$, the chain regenerates in the sense that we draw $\widetilde \bX_i$ from $\nu$, and if $R_i = 0$, we draw $\widetilde \bX_i$ from $(\PP(\bX_i, \cdot) - \delta \nu(\cdot)) / (1 - \delta)$. One can verify that the sequence $\{\widetilde\bX_i\}_{i=1}^n$ has exactly the same distribution as the original Markov chain $\{\bX_i\}_{i=1}^n$. Define $T_1 := \inf\{n > 0: R_n = 1\}$ and $T_{i + 1} := \inf\{n > 0: R_{T_1 + \ldots + T_i + n} = 1\}$ for $i \ge 0$. Note that $\{T_i\}_{i \ge 0}$ are i.i.d. Geometric random variables with $\EE T_1 = 2$ and $\|T_1\|_{\psi_1} \le 4$. Let $S_0 := -1$, $S_j := T_1 + \ldots + T_j$ and $\cY_j := \{ \widetilde \bX_i\}_{i = S_{j - 1} + 1}^{S_j}$ for $j \ge 1$. Based on our construction, we deduce that $\{\cY_j\}_{j \ge 1}$ are independent. Thus we chop the original Markov chain $\{X_i\}_{i \in [n]}$ into independent sequences. Finally, Adamazak's bound entails the following asymptotic weak variance
	\[
	\sigma^2 := \sup_{\bDelta\in\cC(\beta, R, \kappa)} \Var \biggl\{\sum_{i = S_1 + 1}^{S_2} \inn{\bDelta, \bX_i}^2 - \EE(\inn{\bDelta, \bX_i}^2) \biggr\} / \EE T_2. 
	\]
	We have
	\[
	\begin{aligned}
	\sigma^2  & \le \sup_{\bDelta\in\cC(\beta, R, \kappa)}\EE \biggl[\biggl\{ \sum_{i = S_1 + 1}^{S_2} \inn{\bDelta, \bX_i}^2 - \EE(\inn{\bDelta, \bX_i}^2) \biggr\}^2\biggr] / \EE T_2\\
	& = \frac{1}{2} \sup_{\bDelta\in\cC(\beta, R, \kappa)} \sum\limits_{j = 1}^{\infty} \EE \biggl[\biggl\{ \sum_{i = S_1 + 1}^{S_2} \inn{\bDelta, \bX_i}^2 - \EE(\inn{\bDelta, \bX_i}^2) \biggr\}^2 \ind_{\{T_2 = j\}}\biggr] \\
	& \le \frac{1}{2}  \sum\limits_{j = 1}^{\infty}  \frac{j^2R^2\beta^4}{p^4}\PP(T_2 = j) = \frac{R^2 \beta^4\EE (T_2^2)}{2p^4} = \frac{3\beta^4R^2}{p^4}. 
	\end{aligned}
	\]
	By \citet[][Theorem~7]{Ada08}, there exists a universal constant $K$ such that for any $\xi > 0$, 
	\be
	\label{eq:1.7}
	\PP\biggl\{ |\Gamma_n - \EE \Gamma_n| \ge K \EE\Gamma_n + \frac{R\beta^2}{ p^2}\biggl(\frac{3K\xi }{n}\biggr)^{1 / 2} +  \frac{64K\xi\alpha^2 \log n }{np^2}\biggr\} \le K \exp(-\xi).  
	\ee
		
	Next, by the symmetrization argument and Ledoux-Talagrand contraction inequality \citep{ledoux2013probability}, for $n$ independent and identically distributed Rademacher variables $\{\gamma_i\}_{i=1}^n$, when $n \pi_{\max}(1 - \rho_+) \ge \max(20\log p, \log n)$, we have that
	\be
	\label{eq:16}
	\begin{aligned}
		\EE \Gamma_n & \le 2\EE \sup_{\substack{\fnorm{\bDelta} \le R, \\ \bDelta \in \cC(\beta, R, \kappa)}} \biggl | \frac{1}{n} \sum\limits_{i=1}^n \gamma_i \inn{\bDelta, \bX_i}^2 \biggr | \le \frac{8\beta}{p}~\EE \sup_{\substack{\fnorm{\bDelta} \le R, \\ \bDelta \in \cC(\beta, R, \kappa)}} \biggl | \inn{\frac{1}{n}\sum\limits_{i=1}^n \gamma_i \bX_i, \bDelta}  \biggr| \\
		& \le \frac{8\beta\nnorm{\bDelta}}{p} ~\EE \biggl\| \frac{1}{n}\sum\limits_{i=1}^n \gamma_i \bX_i \biggr \|_{2}  \le \frac{8\kappa\beta r^{1 / 2} R}{p}\EE \biggl\| \frac{1}{n}\sum\limits_{i=1}^n \gamma_i \bX_i \biggr\|_2 \le \frac{8\kappa\beta R}{p}\biggl(\frac{r\pi_{\max}\log p}{n}\biggr)^{\!1 / 2}, 
	\end{aligned}
	\ee
	where the penultimate inequality is due to the fact that $\bDelta \in \cC(\beta, R, \kappa)$, and where the last inequality is due to \eqref{eq:expectation_sn_op}. 
	
	Finally, 
	\be
	\EE \inn{\bDelta, \bX_i}^2 = \sum\limits_{1 \le j,k \le d} \pi_j P_{jk}\Delta^2_{jk} \ge \frac{\alpha^2}{p^2} \fnorm{\bDelta}^2. 
	\ee
	Combining all the bounds above, we have for any $\xi > 1$, with probability at least $1 - K \exp(-\xi)$, 
	\be
	\begin{aligned}
		\delta\ell_n(\bP + \bDelta; \bP) \ge \frac{\alpha ^2}{8\beta^2}\fnorm{\bDelta}^2 & - 8R\biggl({\frac{3K\xi }{n}}\biggr)^{1 / 2} - \frac{8K\xi\alpha^2 \log n }{\beta^2 n} - \frac{Kp\kappa R}{\beta}\biggl(\frac{r\pi_{\max}\log p}{n}\biggr) ^ {\!1 / 2}. 
	\end{aligned}
	\ee
\end{proof}

\subsection{Proof of Theorem \ref{thm:nuclear_alt}}

\begin{proof}
	~For a specific $R$ whose value will be determined later, we construct an intermediate estimator $\widehat\bP_{\eta} $ between $\widehat\bP$ and $\bP$:
	\[
	\widehat \bP_{\eta} = \bP + \eta (\widehat \bP - \bP), 
	\]
	where $\eta = 1$ if $\fnorm{\widehat \bP - \bP} \le R$ and $\eta  = R/\fnorm{\widehat\bP - \bP}$ if $\fnorm{\widehat\bP - \bP} > R$. For any $\xi > 1$, there exists a universal constant $C > 0$ such that when 
	\[
	\lambda = C\biggl\{\biggl(\frac{\xi p ^ 2\pi_{\max}\log p}{n\alpha}\biggr)^{1 / 2} + \frac{\xi p\log p}{n \alpha}\biggr\}, 
	\] 
	we have by Lemmas \ref{lem:restricted_strong_convexity} and \ref{lem:gradient} that with probability at least $1 - K\exp(-\xi) - 4p^{-(\xi - 1)} - p ^{-1}$, 
	\be
	\label{eq:stat_error_fnorm}
	\begin{aligned}
		\frac{\alpha ^2}{8\beta^2} & \fnorm{\bDelta}^2 - 8R\biggl({\frac{3K\xi }{n}}\biggr)^{1 / 2} - \frac{8K\xi\alpha^2 \log n }{\beta^2 n} - \frac{Kp\kappa R}{\beta}\biggl(\frac{r\pi_{\max}\log p}{n}\biggr) ^ {\!1 / 2}\\
		& \le \delta\ell_n(\widehat\bP_{\eta}; \bP) \le - \inn{\Pi_{\cN}(\nabla\cL_n(\bP)), \widehat \bDelta_{\eta}}  + \lambda ( \nnorm{\bP} - \nnorm{\widehat\bP_{\eta}} ) \\
		& \le - \inn{\Pi_{\cN}(\nabla\cL_n(\bP)), \widehat \bDelta_{\eta}}  + \lambda\nnorm{\widehat\bDelta_{\eta}} \le (\opnorm{\Pi_{\cN}(\nabla \cL_n(\bP))} + \lambda) \nnorm{\widehat \bDelta_{\eta}}    \\
		& \le 8\lambda \nnorm{[\widehat \bDelta_{\eta}]_{\overline\cM}} \le 8\lambda\sqrt{r}\fnorm{\widehat\bDelta_{\eta}},
	\end{aligned}
	\ee
	where $K$ is the same universal constant as in Theorem \ref{lem:restricted_strong_convexity}. Some algebra yields that 
	\be
	\label{eq:5.24}
	\begin{aligned}
		\fnorm{\widehat \bDelta_{\eta}}^2  \lesssim \frac{\beta ^ 2}{\alpha ^ 2} \max \biggl\{ \frac{\lambda ^ 2r\beta ^ 2}{\alpha ^ 2}, R\biggl({\frac{\xi}{n}}\biggr)^{1 / 2}, \frac{\xi \alpha ^ 2 \log n}{\beta  ^ 2n}, \frac{pR}{\beta}\biggl(\frac{r\pi_{\max}\log p}{n}\biggr) ^ {\!1 / 2}\biggr\}. 
	\end{aligned} 
	\ee
	Letting $R ^ 2$ be greater than the RHS of the inequality above, we can find a universal constant $C_4 > 0$ such that
	\[
	\begin{aligned}
	R \ge \frac{C_4\beta ^ 2}{\alpha ^ 2}\biggl(\frac{\xi r p ^ 2 \pi_{\max} \log p}{n \alpha}\biggr) ^ {\!1 / 2} =: R_0. 
	\end{aligned}
	\]
	Choose $R = R_0$. Therefore, $\fnorm{\widehat \bDelta_{\eta}} \le R$ and $\widehat \bDelta_{\eta} = \widehat \bDelta$. We can thus reach the conclusion. As to the KL-Divergence, by \citet[][Lemma~4]{zhang2018optimal}, we deduce that 
	\be
	\label{eq:l2_to_kl}
	\DKL(\widehat\bP, \bP) = \sum\limits_{j = 1}^p \pi_j \DKL(\bP_{j\cdot}, \widehat\bP_{j\cdot}) \le \sum\limits_{j = 1}^p \frac{\beta^2 }{2 \alpha^2}\ltwonorm{\bP_{j\cdot} - \widehat\bP_{j\cdot}}^2 = \frac{\beta^2}{2\alpha^2} \fnorm{\widehat\bP - \bP}^2, 
	\ee
	from which we attain the conclusion.  
\end{proof}

\subsection{Proof of Theorem \ref{thm:rank_alt}}

\begin{proof}	
	~Define ${\widehat\bDelta}(r) := \widehat \bP^r - \bP$. Since $\rank(\bP) \le r$ and $\rank(\widehat\bP^r) \le r$, $\rank(\widehat\bDelta(r))\le 2r$. Thus $\fnorm{\widehat \bDelta(r)} \le (2r)^{1 / 2} \nnorm{\widehat \bDelta(r)}$. 
	%	By Lemma \ref{lem:restricted_strong_convexity}, we obtain that
	%	\[
	%		abc
	%	\]
	%	By replacing \eqref{eq:16} with the following bound
	%	\be
	%	\begin{aligned}
	%		\EE \Gamma_r & \le 2\EE \sup_{\substack{\fnorm{\bDelta} \le R, \\ \rank(\bDelta) \le 2r}} \bigl | \frac{1}{n} \sum\limits_{i=1}^n \gamma_i \inn{\bDelta, \bX_i}^2 \bigr | \le \frac{8\beta}{p} \cdot \EE \sup_{\substack{\fnorm{\bDelta} \le R, \\ \rank(\bDelta) \le 2r}} \bigl| \inn{\frac{1}{n}\sum\limits_{i=1}^n \gamma_i \bX_i, \bDelta}  \bigr| \\
	%		& \le \frac{8\beta\nnorm{\bDelta}}{p} \cdot \EE \opnorm{ \frac{1}{n}\sum\limits_{i=1}^n \gamma_i \bX_i} \le \frac{8\sqrt{2r}\beta \fnorm{\bDelta}}{p} \cdot \EE \opnorm{ \frac{1}{n}\sum\limits_{i=1}^n \gamma_i \bX_i}, 
	%	\end{aligned}
	%	\ee
	%	we can borrow the proof of Lemma \ref{lem:restricted_strong_convexity} to deduce that there exists a constant $C$ depending only on $\alpha$ and $\beta$ such that for any $\xi > 1$, it holds with probability at least $1 - K\exp(-\xi)$ that for any $\fnorm{\bDelta} \le R$ and $\rank(\bDelta) \le 2r$, 
	%	\be
	%	\begin{aligned}
	%		\delta\cL_n(\bP + \bDelta; \bP) \ge \frac{\alpha ^2}{8\beta^2}\fnorm{\bDelta}^2 & - 8R\sqrt{\frac{3K\xi }{n}} - \frac{8K\xi\alpha^2 \log n }{\beta^2 n} \\
	%		& - \frac{CKR\sqrt{r}}{\beta}\Bigl(\sqrt{\frac{p \log p}{n(1 - \rho_+)}} + \frac{p \log p}{n}\Bigr). 
	%	\end{aligned}
	%	\ee 
	Now we follow the proof strategy of Theorem \ref{thm:nuclear} to establish the statistical error bound for $\widehat \bP^r$. Similarly, for a   specific $R > 0$ whose value will be determined later, we can construct an intermediate estimator $\widehat\bP^r_{\eta} $ between $\widehat\bP^r$ and $\bP$:
	\[
	\widehat \bP^r_{\eta} = \bP + \eta (\widehat \bP^r - \bP), 
	\]
	where $\eta = 1$ if $\fnorm{\widehat \bP^r - \bP} \le R$ and $\eta  = R/\fnorm{\widehat\bP^r - \bP}$ if $\fnorm{\widehat\bP^r - \bP} > R$. Let $\widehat \bDelta_\eta(r) := \widehat \bP^r_\eta - \bP$. Since $\widehat \bDelta_\eta(r) \in \cC(\beta, R, \sqrt{2})$, applying Lemma \ref{lem:restricted_strong_convexity} yields that 
	\be
	\begin{aligned}
		\frac{\alpha ^2}{8\beta^2} & \fnorm{\bDelta}^2 - 8R\biggl({\frac{3K\xi }{n}}\biggr)^{1 / 2} - \frac{8K\xi\alpha^2 \log n }{\beta^2 n} - \frac{Kp\kappa R}{\beta}\biggl(\frac{r\pi_{\max}\log p}{n}\biggr) ^ {\!1 / 2}\\
		& \le \delta\ell_n(\widehat\bP^r_{\eta}; \bP) \le - \inn{\Pi_{\cN}(\nabla\ell_n(\bP)), \widehat \bDelta_{\eta}(r)} \le \opnorm{\Pi_{\cN}(\nabla \cL_n(\bP))} \nnorm{\widehat \bDelta_{\eta}(r)}\\
		& \le \sqrt{2r}\opnorm{\Pi_{\cN}(\nabla \cL_n(\bP))} \fnorm{\widehat \bDelta_{\eta}(r)}, 
	\end{aligned}
	\ee
	which futher implies that there exists $C_1$ depending only on $\alpha$ and $\beta$ such that 
	\[
	\fnorm{\widehat \bDelta_{\eta}(r)}^2  \le C_1 \max \biggl\{r\opnorm{\Pi_{\cN}(\nabla \cL_n(\bP))} ^ 2, R\biggl({\frac{\xi}{n}}\biggr)^{1 / 2}, \frac{\xi \alpha ^ 2 \log n}{\beta  ^ 2n}, \frac{pR}{\beta}\biggl(\frac{r\pi_{\max}\log p}{n}\biggr) ^ {\!1 / 2}\biggr\}. 
	\]
	By a contradiction argument as in the proof of Theorem \ref{thm:nuclear}, we can choose an appropriate $R$ large enough such that $\widehat \bP^r_\eta = \widehat \bP^r $ and attain the conclusion. 
	
\end{proof}

\section{Proof of Proposition \ref{prop:penlowrank}}%\ref{prop:penlowrank}
\begin{proof}
	% problem 5
	%	\begin{equation}\label{prob:gen-nonconvex-lowrank}
	%	\min \, \left\{ f(X) \,\mid\, \cA X = b,\, {\rm rank}(X)\le r \right\}, 
	%	\end{equation}
	~Since ${\rm rank}(\X_c^*)\le r$, we know that $\X_c^*$ is in fact a feasible solution to the original problem (5) and $\norm{\X_c^*}_{*} - \norm{\X_c^*}_{(r)} = 0$. Therefore, for any feasible solution $X$ to
	(5), it holds that 
	\begin{align*} 
	f(\X_c^*) ={}& f(\X_c^*) + c(\norm{\X_c^*}_{*} - \norm{\X_c^*}_{(r)})\\[5pt]
	\le{}& f(\X) + c(\norm{\X}_* - \norm{\X}_{(r)})
	= f(\X).
	\end{align*}
	This completes the proof of the proposition.
\end{proof}

\section{Convergence and $o(1/k)$ non-ergodic iteration complexity of Algorithm 1 (sGS-ADMM)}

Before deriving the desired results of Algorithm \ref{alg:sGS-ADMM} for solving problem \eqref{prob:D}, we present some notation and definitions for the subsequent analysis. Assume that the solution sets of \eqref{prob:gen-convex-nuc} and \eqref{prob:D} are nonempty. Then, the primal-dual solution pairs associated with problems \eqref{prob:gen-convex-nuc} and \eqref{prob:D} satisfy the following Karush-Kuhn-Tucker (KKT) system: 
%% the constraint qualification for (P) is satisfied
\begin{equation}
\label{KKT}
0 \in R({\bf X}, {\bf \Xi}, {\bf S}), \quad \cA({\bf X}) = b, \quad 
{\bf \Xi} + \cA^*(y) + {\bf S} = 0,
\end{equation}
with 
\begin{equation*}
R({\bf X}, {\bf \Xi}, {\bf S}): = \begin{pmatrix}
{\bf \Xi} + \partial g({\bf X}) \\[5pt]
{\bf X} + \partial \delta(\norm{{\bf S}}_

2 \le c)
\end{pmatrix}, \quad ({\bf X}, {\bf \Xi}, {\bf S})\in {\rm dom}\,g \times \Re^{p\times p}\times \left\{{\bf S}\in \Re^{p\times p}\mid \norm{{\bf S}}_2 \le c \right\}.
\end{equation*}
Define the KKT residual function $D:{\rm dom}\,g \times \Re^{p\times p}\times \Re^n \times \left\{{\bf S}\in \Re^{p\times p}\mid \norm{{\bf S}}_2 \le c \right\} \to [0, +\infty)$ as
\[
D({\bf X}, {\bf \Xi}, y, {\bf S}):= {\rm dist}^2(0, R({\bf X}, {\bf \Xi}, {\bf S})) + \norm{\cA({\bf X}) - b}^2 + 
\norm{{\bf \Xi} + \cA^*(y) + {\bf S}}^2.
\]
We say $({\bf X}, {\bf \Xi}, y, {\bf S})\in {\rm dom}\,g \times \Re^{p\times p}\times \Re^n \times \left\{{\bf S}\in \Re^{p\times p}\mid \norm{{\bf S}}_2 \le c \right\}$ be an $\epsilon$-approximate primal-dual solution pair for problems \eqref{prob:gen-convex-nuc} and \eqref{prob:D} if
$D({\bf X}, {\bf \Xi}, y, {\bf S}) \le \epsilon$. We show in the following theorem the global convergence and the $o(1/k)$ iteration complexity results of Algorithm sGS-ADMM.
\begin{theorem}
	\label{thm:sGS-ADMM}
	Suppose that the solution sets of \eqref{prob:gen-convex-nuc} and \eqref{prob:D} are nonempty. Let $\{({\bf\Xi}^k,y^k,\S^k,\X^k)\}$ be the sequence generated by Algorithm \ref{alg:sGS-ADMM}. If $\tau\in(0,(1+\sqrt{5}\,)/2)$, then the sequence $\{({\bf\Xi}^k,y^k,\S^k)\}$ converges to an optimal solution of \eqref{prob:D} and $\{\X^k\}$ converges to an optimal solution of \eqref{prob:gen-convex-nuc}.
	Moreover, there exist a constant $\omega >0$ such that 
	\[
	\min_{1\le i \le k} \left\{ D({\bf X}^k, {\bf \Xi}^k, y^k, {\bf S}^k) \right\} \le \frac{\omega}{k}, \ \forall\, k\ge 1, \quad{\rm and}  \quad 
	\lim_{k\to \infty} \left\{ k \min_{1\le i \le k} \left\{ D({\bf X}^k, {\bf \Xi}^k, y^k, {\bf S}^k) \right\} \right\}= 0.
	\]
\end{theorem}

\begin{proof}
	~In order to use \citep[Theorem 3]{li2016schur}, we need to write problem \eqref{prob:D} as following
	\begin{equation*} 
	\begin{array}{rll}
	\min  & g^*(-{\bf\Xi}) - \inprod{b}{y} + {\delta( \norm{\S}_2 \le c)}  \\
	\mbox{s.t.} & \cF ({\bf\Xi}) + \cA_1^*(y) + \cG(\S)  = 0,
	\end{array}
	\end{equation*}
	where $\cF, \cA_1$ and $\cG$ are linear operators such that for all $({\bf \Xi}, y, \S) \in \Re^{p\times p} \times \Re^n \times \Re^{p\times p}$, $\cF({\bf\Xi}) = {\bf \Xi}$, 
	$\cA_1^*(y) = \cA^*(y)$ and $\cG(\S) = \S$. 
	Clearly, $\cF = \cG = \cI$  where $\cI:\Re^{p\times p} \to \Re^{p\times p}$ is the identity map.
	Therefore, we have $\cA_1\cA_1^* \succ 0$ and $\cF\cF^* = \cG\cG^* = \cI \succ 0$. Hence, the assumptions and conditions in  \citep[Theorem 3]{li2016schur} are satisfied. The convergence results thus follow directly. Meanwhile, the non-ergodic iteration complexity results follows from \citep[Theorem 6.1]{chen2017efficient}.
\end{proof}

%\section{Convergence of Algorithm 1 (sGS-ADMM) and its proof}
%\begin{theorem}
%	\label{thm:sGS-ADMM}
%	Suppose that the solution sets of \eqref{prob:gen-convex-nuc} and \eqref{prob:D} are nonempty. Let $\{({\bf\Xi}^k,y^k,\S^k,\X^k)\}$ be the sequence generated by Algorithm \ref{alg:sGS-ADMM}. If $\tau\in(0,(1+\sqrt{5}\,)/2)$, then the sequence $\{({\bf\Xi}^k,y^k,\S^k)\}$ converges to an optimal solution of \eqref{prob:D} and $\{\X^k\}$ converges to an optimal solution of \eqref{prob:gen-convex-nuc}.
%\end{theorem}
%
%\begin{proof}
%	~In order to use \citep[Theorem 3]{li2016schur}, we need to write problem \eqref{prob:D} as following
%	\begin{equation*} 
%	\begin{array}{rll}
%	\min  & g^*(-{\bf\Xi}) - \inprod{b}{y} + \delta(\S\mid \norm{\S}_2 \le c)  \\
%	\mbox{s.t.} & \cF ({\bf\Xi}) + \cA_1^*(y) + \cG(\S)  = 0,
%	\end{array}
%	\end{equation*}
%	where $\cF, \cA_1$ and $\cG$ are linear operators such that for all $({\bf \Xi}, y, \S) \in \Re^{p\times p} \times \Re^n \times \Re^{p\times p}$, $\cF({\bf\Xi}) = {\bf \Xi}$, 
%	$\cA_1^*(y) = \cA^*(y)$ and $\cG(\S) = \S$. 
%	Clearly, $\cF = \cG = \cI$  where $\cI:\Re^{p\times p} \to \Re^{p\times p}$ is the identity map.
%	Therefore, we have $\cA_1\cA_1^* \succ 0$ and $\cF\cF^* = \cG\cG^* = \cI \succ 0$. Hence, the assumptions and conditions in  \citep[Theorem 3]{li2016schur} are satisfied. The convergence results thus follow directly.
%\end{proof}

\section{Proof of Theorems  \ref{thm:convergence-alg-MM} and \ref{thm: MMconvergence}} %\ref{thm:convergence-alg-MM} and \ref{thm: MMconvergence}
We only need to prove Theorem \ref{thm: MMconvergence} 
as Theorem \ref{thm:convergence-alg-MM} 
is a special incidence. To prove Theorem  \ref{thm: MMconvergence}, 
we first introduce the following lemma.	
\begin{lemma}\label{lemma:decrease}
	Suppose that $\{ {x}^k \}$ is the sequence generated by Algorithm 3. % \ref{alg:dca-general} . 
	Then  $\theta({x}^{k+1})\le \theta({x}^k) - \frac{1}{2}\|{x}^{k+1} - {x}^k\|^2_{\mathcal{G} + 2\mathcal{T}}$.
\end{lemma}
\begin{proof}
	%Part (i) can be easily obtained from the inequalities \eqref{ineq:majorization} and \eqref{ineq:convexity} that
	%\begin{equation}\label{eq:sandwich}
	%\theta({x}^{k+1}) \le \wh{\theta}({x}^{k+1};{x}^k) \le \wh{\theta}({x}^k;{x}^k) = \theta({x}^k),\quad \forall\; k\geq 0.
	%\end{equation}
	%(10) subproblem of generalized DCA
~For any $k\geq 0$, by the optimality condition of problem (10) at
	${x}^{k+1}$, we know that  
	there exist $\eta^{k+1}\in \partial p({x}^{k+1})$  such that
	\begin{equation*}\label{eq:major-k-optimality}
	0 =  \nabla   g ({x}^k) +   (\mathcal{G} + \mathcal{T})({x}^{k+1} -
	x^{k})  + \eta^{k+1}-\xi^k  = 0.
	\end{equation*}
	Then for any $k\ge 0$, we deduce
	\begin{equation*}
	\begin{array}{rl}
	& \theta({x}^{k+1}) - \theta({x}^k)
	\le \widehat{\theta}({x}^{k+1};{x}^k) - \theta({x}^k)\\[0.1in]
	= & p(x^{k+1}) - p(x^k) + \langle {x}^{k+1} - {x}^k , \nabla g({x}^k)-\xi^k \rangle +
	\frac{1}{2} \|{x}^{k+1} - {x}^k\|^2_{\mathcal{G}} \\[0.1in]
	\le &  \langle  \nabla g(x^k)+\eta^{k+1} -\xi^k, {x}^{k+1} - {x}^k\rangle+
	\frac{1}{2} \|{x}^{k+1} - {x}^k\|^2_{\mathcal{G}}
	\\[0.1in]
	=  &  - \frac{1}{2}\|x^{k+1} - x^k\|^2_{\mathcal{G} + 2\mathcal{T}}.
	\end{array}
	\end{equation*}
	This completes the proof of this lemma.
\end{proof}

%	Lemma \ref{lemma:decrease} indicates that if $\cG + 2\cT \succeq 0$, then the objective value, along the sequence generated by Algorithm \ref{alg:dca-general}, is non-increasing. In the other words, we proved that there is a sufficient decrease between two consecutive iterations. In fact, $\cG+2\cT \succeq 0$ ($\cG+2\cT \succ 0$) also implies that subproblem \eqref{eq:subproblem} is a (strongly) convex problem.  
Now we are ready to prove Theorem \ref{thm: MMconvergence}.
\begin{proof}
	~From the optimality condition at $x^{k+1}$, we have that 
	\[ 0 \in \nabla   g ({x}^k) +   (\mathcal{G} + \mathcal{T})({x}^{k+1} -
	x^{k})  +  \partial p(x^{k+1})-\xi^k.\]
	Since $x^{k+1} = x^k$, this implies that 
	\[ 0 \in \nabla   g ({x}^k) +  \partial p(x^{k})- \partial q(x^k), \]
	i.e., $x^k$ is a critical point.
	Observe that the sequence $\{\theta (x^{k})\}$ is non-increasing since
	$${\theta}(x^{k+1})  \le \widehat{\theta}(x^{k+1}; x^{k})  \le \widehat{\theta}(x^{k}; x^{k}) =\theta(x^{k}), \quad k\geq 0.$$
	Suppose that there exists a subsequence $\{x^{k_j}\}$ that converging to $\bar{x}$, i.e.,  one of the accumulation points of $\{x^k\}$.
	By Lemma \ref{lemma:decrease} and the assumption that $\mathcal{G} + 2\mathcal{T}\succeq 0$, we know that for all $x\in \mathbb{X}$
	\begin{align*}
	&\widehat{\theta}(x^{k_{j+1}};x^{k_{j+1}}) = \theta(x^{k_{j+1}}) \\
	\le &\theta(x^{k_j+1})\le \widehat{\theta}(x^{k_j+1};x^{k_j})\le \widehat{\theta}(x;x^{k_j}).
	\end{align*}
	By letting $j\to\infty$ in the above inequality, we obtain that
	$$
	\widehat{\theta}(\bar{x};\bar{x})\le  \widehat{\theta}(x;\bar{x}).
	$$
	By the optimality condition of $\widehat{\theta}(x; \bar{x})$, we have that
	there exists $\bar{u}\in \partial p(\bar{x})$ and $\bar{v}\in \partial q(\bar{x})$ such that
	$$
	0 \in \nabla g(\bar{x}) + \bar{u} - \bar{v}.  
	$$
	This implies that $\left(\nabla g(\bar{x}) + \partial p(\bar{x})  \right)\cap \partial q(\bar{x})\neq \emptyset$.
	To establish the rest of this proposition, 
	we obtain from Lemma 1 that
	\begin{align*}
	&\lim_{t \to + \infty}\frac{1}{2} \sum_{i=0}^t \|{x}^{k+1}-{x}^k\|^2_{\mathcal{G}+2\mathcal{T}} \\
	\le {}&\liminf_{t\to +\infty} \big( \theta(x^0)
	-\theta({x}^{k+1})\big) \le \theta(x^0) < +\infty \,,
	\end{align*}
	which implies $  \lim_{i\to +\infty} \|{x}^{k+1} -
	x^{i}\|_{\mathcal{G}+2\mathcal{T}} = 0.$ The proof of this theorem is thus complete by the positive definiteness of the operator $\mathcal{G} + 2\mathcal{T}$.
\end{proof}
\end{document}